\documentclass[jmp, amssymb, amsmath, preprint, preprintnumbers, superscriptaddress]{revtex4-1}
\usepackage{bbm, amsthm, changes, graphicx, dsfont, verbatim, appendix, physics}
\usepackage[hyperindex, breaklinks]{hyperref}
\hypersetup{
     colorlinks=true,        
     citecolor=blue,    
     filecolor=blue,      		
     urlcolor=blue,           	
    runcolor=cyan,
}

\def\Z2{\mathbb Z_2}
\def\zero{{\bf 0}}
\def\u{{\bf u}}
\def\v{{\bf v}}
\def\w{{\bf w}}

\def\N{\mathcal N}
\def\S{\mathcal S}
\def\I{\mathrm i}
\def\unity{\mathbbm 1}
\def\C{\mathcal C}
\def\pr{{\rm prob}}
\def\V{\mathcal V}
\def\Ls{L_{\mathrm s}}
 
\newcommand{\SympForm}[2]{\langle #1 , #2 \rangle}

\theoremstyle{definition}
\newtheorem{Lemma}{Lemma}
\newtheorem{Definition}[Lemma]{Definition}
\newtheorem{Theorem}[Lemma]{Theorem}

\begin{document}
\title{Mixing and localisation in random time-periodic quantum circuits of Clifford unitaries}

\author{Tom Farshi}
\email[]{t.farshi.17@ucl.ac.uk}
\affiliation{Department of Computer Science, University College London, United Kingdom}
\affiliation{Department of Physics and Astronomy, University College London, United Kingdom}
\author{Daniele Toniolo}
\affiliation{Department of Computer Science, University College London, United Kingdom}
\affiliation{Department of Physics and Astronomy, University College London, United Kingdom}
\author{Carlos E.~Gonz\'alez-Guill\'en} 
\affiliation{ Dept Matem\'atica Aplicada a la Ingenier\'\i a Industrial, Universidad Polit\'ecnica de Madrid, Spain.}
\author{\'Alvaro M.~Alhambra}
\affiliation{ Max-Planck-Institut für Quantenoptik, Garching, Germany}
\affiliation{Perimeter Institute for Theoretical Physics, Canada}

\author{Lluis Masanes} 
\affiliation{Department of Computer Science, University College London, United Kingdom}
\affiliation{London Centre for Nanotechnology, University College London, United Kingdom}

\date\today 

\begin{abstract}\noindent
How much does local and time-periodic dynamics resemble a random unitary? 
In the present work we address this question by using the Clifford formalism from quantum computation.
We analyse a Floquet model with disorder, characterised by a family of local, time-periodic, random quantum circuits in one spatial dimension.
We observe that the evolution operator enjoys an extra symmetry at times that are a half-integer multiple of the period.
With this we prove that after the scrambling time, namely when any initial perturbation has propagated throughout the system, the evolution operator cannot be distinguished from a (Haar) random unitary when all qubits are measured with Pauli operators.
This indistinguishability decreases as time goes on, which is in high contrast to the more studied case of (time-dependent) random circuits.
We also prove that the evolution of Pauli operators displays a form of mixing. 
These results require the dimension of the local subsystem to be large. In the opposite regime our system displays a novel form of localisation, 
produced by the appearance of effective one-sided walls, which prevent perturbations from crossing the wall in one direction but not the other.
\end{abstract}

\maketitle
\tableofcontents

\section{Introduction}

The distinction between chaotic and integrable quantum dynamics \cite{LesHouches_1989} plays a central role in many areas of physics, like the study of equilibration \cite{Short_2012}, thermalisation \cite{Rigol_2008}, and related topics like the eigenstate thermalisation hypothesis \cite{D'Alessio_review_2016, Deutsch_2013}, quantum scars \cite{Turner_2018, Bernevig_2021}, and the generalised Gibbs ensemble \cite{Essler_review_2016}.
This distinction is also important in the characterisation of many-body localisation \cite{Imbrie_2016}, the holographic correspondence between gravity and conformal field theory \cite{Maldacena_1999}, and in arguments concerning the black-hole information paradox \cite{Hayden_2007, Maldacena_Susskind_2013}.
Despite all this, the precise definitions of quantum chaos and integrability are still being debated \cite{Prosen_2007, Caux_2011, Yuzbashyan_2016}.
However, it is well established that the dynamics of quantum chaotic systems shares important features with random unitaries \cite{Haake_book}. These are the unitaries obtained with high probability when sampling from the unitary group of the total Hilbert space of the many-body system according to the uniform distribution (Haar measure \cite{Simon_book_repr}).

In order to find signatures of quantum chaos in physically relevant systems, it is a common practice to identify in them aspects of random unitaries. Some of these are: the presence of eigenvalue repulsion in the Hamiltonian \cite{Prosen_2018_1, Mondal_2019}, fast decay of out-of-time order correlators \cite{Maldacena_Shenker_Stanford, Kitaev_2019, Chalker_PRX_2019}, entanglement spreading \cite{Prosen_2019}, operator entanglement \cite{Bertini_2019}, entanglement spectrum \cite{Chamon_2014, Zhou_2019}, and Loschmidt echo \cite{Yan_2019}.
In this work we take a more operational approach and analyse setups in which the evolution operator of a system is physically indistinguishable from a random unitary.
We quantify this indistinguishability with a variant of the quantum-information notion of unitary $2$-design \cite{Eisert_2007}.
A set of unitaries $\mathcal U \subset {\rm SU}(2^n)$ forms a $2$-design if, despite having access to $2$ copies of a given unitary $U$, we cannot discriminate between the case where $U$ is sampled from $\mathcal U$ or from ${\rm SU}(2^n)$.
Because of this, unitaries sampled from $\mathcal U$ are called pseudo-random.
In our weaker variant of $2$-design we restrict the class of measurements available for this discrimination process to multi-qubit Pauli operators.
To define our set $\mathcal U$ we consider a model with (spatial) disorder, where each element of $\mathcal U$ is the evolution operator $W(t)$ at a fixed time $t$ generated by a particular configuration of the disorder (see Figure~\ref{fig:CircuitFigure} for $W(2)$).


\begin{figure}
    \hspace{-3mm}
    \includegraphics[width=152mm, height=51mm]
    {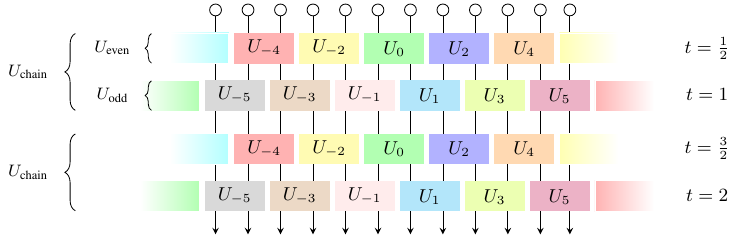}
    \caption{{\bf Time-periodic local dynamics.} This figure illustrates the physical model analysed in this work. The circles on top represent  lattice sites, each consisting of $N$ qubits. Coloured blocks represent two-site unitaries, and different colours stand for independently and identically distributed Clifford unitaries, representing the spatial disorder. After the first two half time-steps the dynamics repeats itself.}
    \label{fig:CircuitFigure}
\end{figure}


In this work we consider a spin chain with $L$ sites and periodic boundary conditions, where each site contains $N$ modes or qubits.
The first dynamical period consists of two half-steps. 
In the first half-step each even site interacts with its right neighbour with a random Clifford unitary (for the definition of the Clifford group see Appendix~\ref{app:CliffordPhaseSpace} or \cite{Koenig_2014}) and in the second half-step each odd site interacts with its right neighbour with a random Clifford unitary.
These $L$ Clifford unitaries are independent and uniformly sampled from the $2N$-qubit Clifford group.
The subsequent periods of the dynamics are repetitions of the first period, as illustrated in Figure \ref{fig:CircuitFigure}.
If we denote by $U_x$ the above-mentioned unitary action on sites $x$ and $x+1$ (modulo $L$ due to periodic boundary conditions) then the evolution operator after an {\em integer} time $t$ is 
\begin{align}\label{eq:integer}
  W(t) & = \big[ (U_1 \otimes U_3 \otimes  \cdots \otimes U_{L-1}) (U_0 \otimes U_2 \otimes \cdots \otimes U_{L-2}) \big]^{t} 
  \\ \nonumber & =
  (U_{\text{odd}}U_{\text{even}})^{t} = (U_{\text{chain}})^{t}\ ,
\end{align}
and after a {\em half-integer} time $t$ is 
\begin{align}\label{eq:halfinteger}
  W(t) 
  = U_{\text{even}}\, (U_{\text{chain}})^{t-1/2} 
 \ .
\end{align}
This evolution operator can also be generated by a time-periodic Hamiltonian $H(t)$ with nearest-neighbour interactions
\begin{equation}
  W(t) = \mathcal T e^{-\I \int_0^t d\tau H(\tau)}\ ,
\end{equation}
where $\mathcal T$ is the time-ordering operator.
This type of dynamics is called Floquet. Floquet dynamics in relation with the phenomenon of quantum chaos has been studied, among others, by Prosen and coauthors \cite{Prosen_2018_1, Prosen_2018_2, Prosen_2019, Bertini_2019}, with the review \cite{Prosen_2007}. A general review on quantum Floquet systems is reference \cite{Bukov_review_2015}. In the quantum information community the term QCA, quantum cellular automaton, commonly denotes such periodic systems \cite{Schlingemann_2008}, a review is \cite{Farrelly_2019}. The authors of \cite{Sunderhauf_2018} considered a QCA, with the same structure as us but with Haar-distributed unitaries gates instead of Cliffords.
(Another Floquet-Clifford model has been studied in \cite{Chandran_2015}.)
It is important to stress that this time-periodic model is very different from the much more studied time-dependent ``random circuits" \cite{Harrow_2009, Harrow_2018, Brandao_2016, Brandao_2016a, Hunter_Jones_2019, Nakata_2017, Brandao_2019} depicted in Figure~\ref{fig:RandomCircuitFigure}. Time-periodic circuits are more difficult to analyse but more relevant to physics; since they enjoy a (discrete) time translation symmetry.
\begin{figure}
    \centering
    \includegraphics[width=142mm]
    {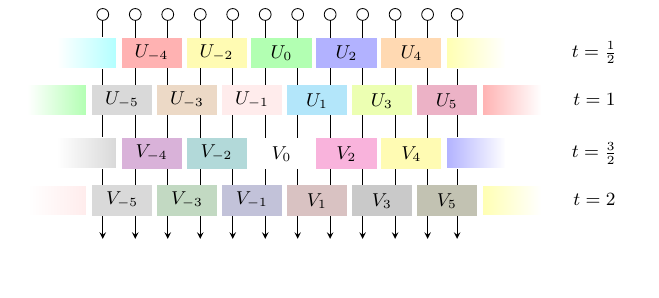}
    \vspace{-12mm}
    \caption{{\bf Time-dependent local dynamics.} In contrast to Figure~\ref{fig:CircuitFigure} the pictured circuit is not periodic in time: different time steps are independently sampled.}
    \label{fig:RandomCircuitFigure}
\end{figure}

We show that the ensemble of evolution operators at half-integer time \eqref{eq:halfinteger} has a larger symmetry than that at integer times \eqref{eq:integer}.
This allows us to prove approximate Pauli mixing \cite{Webb_2016}: each Pauli operator evolving with the random dynamics \eqref{eq:halfinteger} reaches any other Pauli operator inside its light cone with approximately equal probability.
In the integer-time case this only holds for a restricted class of initial operators, which includes local ones.
We also prove that at any half-integer time after the scrambling time $t_{\rm scr}$, the ensemble of evolution operators \eqref{eq:halfinteger} cannot be operationally distinguished from Haar-random unitaries (in the sense specified above).
We define the scrambling time $t_{\rm scr}$ as the smallest time allowing for any local perturbation to reach the entire system (in our model $t_{\rm scr} = L/4$).
In all these results, the degree of approximation increases with $N$ and decreases with time $t$.


Besides many-body physics, our results are relevant to the field of quantum information. The authors of  \cite{Emerson_2003} design a protocol to generate pseudo-random unitaries. Many quantum information tasks make use of unitary designs (entanglement distillation~\cite{Bennett_1996a}, quantum error correction \cite{Abeyesinghe_2009, Brandao_2016a}, randomised benchmarking \cite{Magesan_2011}, quantum process tomography \cite{Scott_2008}, quantum state decoupling~\cite{Brown_2015} and data-hiding \cite{DiVincenzo_2002}). In most current implementations of quantum information processing qubits are measured in a fixed basis, a particular case of our Pauli measurements. Hence we expect that our variant of 2-designs restricted to Pauli measurements will be useful in some of these applications, in particular on architectures where a time-periodic drive is more feasible than a time-dependent drive. It is worth mentioning that Google's quantum supremacy demonstration \cite{Arute_2019} is based on the statistics of multi-qubit Pauli measurements after pseudo-random unitary dynamics; and that their random circuit consists of time-dependent single-qubit gates and time-periodic two-qubit gates.

In the model under consideration, the number of modes per site $N$ is a free parameter that controls the behaviour of the system.
In the large-$N$ regime ($N \gg \log L$) we obtain the above-mentioned indistinguishability between the evolution operator \eqref{eq:halfinteger} and a Haar-random unitary, which increases with $N$. 
(We recall that large-$N$ is the relevant regime in holographic quantum gravity.) 
In the opposite regime ($N \ll \log L$) our model displays a novel form of localisation 
produced by the appearance of effective one-sided walls that prevent perturbations from crossing the wall in one direction but not the other.   
Interestingly, this localised phase seems to challenge the existing classification. On one hand, our model is not a system of free or interacting particles, so it does not fit in the framework of Anderson localisation. On the other hand, the evolution of each local operator is strictly confined to a finite region, so it does not behave as many-body localised. See \cite{Imbrie_2016} for the description of the differences among Anderson and many-body localisation and \cite{Nandkishore_2015} for a recent physicists' review on localisation phenomena.

This model also challenges the classification of integrable and chaotic quantum systems.
On one hand, it has a phase-space description of the dynamics like that of quasi-free bosons and fermions, and it can be efficiently simulated on a classical computer \cite{Aaronson_2004, Gottesman_1998}. See reference \cite{Koenig_2014} for an algorithmic classification of the elements of the Clifford group and appendix \ref{app:CliffordPhaseSpace} for the description of the Clifford group phase space. 
On the other hand, this model does not have anything close to local (or low-weight) integrals of motions, and it behaves like Haar-random dynamics in a way that quasi-free systems do not.
Therefore, we believe that Clifford dynamics is valuable for mapping the landscape of many-body phenomena.
It is also important to recall that we live in the age of synthetic quantum matter, see the review \cite{Bloch_2012}, and models similar to ours have actually been implemented on quantum simulators, like in the recent Google's experiment \cite{Arute_2019}.


In the following section we present our results on mixing (sections \ref{sec:ergo} and \ref{sec:ergo1/2}), pseudo-random unitaries (Section \ref{sec:2Design}) and strong localisation (Section \ref{sec:StrongLocalisation}).
In order to do so we introduce a few mathematical notions beforehand (Section \ref{sec:prel}). The proofs of the theorems presented in the Section \ref{results} are presented in sections \ref{app:LocalDynamics_PauliMixing}, \ref{app:PauliMixing_SemiInteger}, \ref{app:PauliMixing_Integer}, \ref{app:Ergodicity_AritraryInitial}, \ref{app:ApproxDesignSemiTime}, \ref{app:CliffordLocalisation}, together with several other lemmas. 
In Section \ref{sec:scrambling time} we discuss the physical significance of the scrambling time.
In Section \ref{sec:int/chao} we discuss the difficulties with classifying our model as integrable or chaotic.
We compare time-periodic and time-independent circuits in Section \ref{sec:t dep}.
In Section \ref{sec:conc} we provide the conclusions of our work.
The appendices include an introduction to the Clifford formalism, appendix \ref{app:CliffordPhaseSpace}, and some additional lemmas, appendix \ref{app:sympl_order} and \ref{app:AdditionalLemmas}.

\section{Results}\label{results}

\subsection{Preliminaries}\label{sec:prel}

An $n$-qubit {\em Pauli operator} is a tensor product of Pauli sigma matrices and one-qubit identities times a global phase $\lambda \in \{1,\I, -1, -\I\}$.
In what follows we ignore this global phase $\lambda$, so each Pauli operator is represented by a binary vector $\u = (q_1, p_1, q_2, p_2, \ldots, q_n, p_n ) \in \{0,1\}^{2n}$ as 
\begin{align}
  \label{eqb:Pauli element}
  \sigma_\u
  =
  \bigotimes_{i=1}^n 
  (\sigma_x^{q_i} \sigma_z^{p_i})\ .
\end{align}
Ignoring the global phase $\lambda$ allows us to write the product of Pauli operators as the simple rule $\sigma_\u \sigma_{\u'} = \lambda\, \sigma_{\u +\u'}$, where addition in the vector space $\{0,1\}^{2n}$ is modulo 2.
This defines the Pauli group, which is the discrete analog of the Weyl group, or the displacement operators used in quantum optics.

The $n$-qubit {\em Clifford} group $\C_n \subseteq {\rm SU}(2^n)$ is the set of unitaries $U$ which map each Pauli operator onto another Pauli operator $U\sigma_\u U^\dagger = \lambda\, \sigma_{\u'}$.
Each Clifford unitary $U$ can be represented by a $2n\times 2n$ symplectic matrix $S$ with entries in $\{0,1\}$ such that its action on Pauli operators can be calculated in phase space 
\begin{equation}
  U\sigma_\u U^\dagger 
  = \lambda\, \sigma_{S\u}\ ,
\end{equation}
where the matrix product $S\u$ is defined modulo 2. 
We call the binary vectors $\u \in \{0,1\}^{2n}$ the {\em phase space} representation of the Pauli operator $\sigma_\u$ because of its analogy with quasi-free bosons, where dynamics is also linear and symplectic.
%
%
A detailed introduction to the discrete phase space and Clifford and Pauli groups is provided in Appendix~\ref{app:CliffordPhaseSpace}, see also the references \cite{Koenig_2014,Aaronson_2004,Gottesman_1998}.
Note that Clifford unitaries are easy to implement in several quantum computation and simulation architectures.
 
Our model is an $L$-site lattice with even $L$ and periodic boundary conditions. 
The corresponding phase space can be written as
\begin{equation}
  \label{eqb:V chain}
  \mathcal V_{\rm chain} 
  =
  \bigoplus_{x\in \mathbb Z_L} \mathcal V_x
  \ ,
\end{equation}
where $\mathcal V_x \cong \Z2^{2N}$ is the phase space of site $x\in \mathbb Z_L$, which represents $N$ qubits.
A local Pauli operator $\sigma_\u$ at site $x$ is represented by a phase-space vector contained in the corresponding subspace $\u \in \mathcal V_{x} \subseteq \mathcal V_{\rm chain}$.
The identity operator corresponds to the zero vector.
(See Section \ref{app:DesciptionOfModel} for a collated description of the model and its phase space description). In the following we will denote $ S(t) $ the symplectic matrix acting on the space $ \mathcal V_{\rm chain} $ associated with the evolution operator $ W(t) $ as defined by equations \eqref{eq:integer} and \eqref{eq:halfinteger}.

\subsection{Approximate Pauli mixing}\label{sec:ergo}

A set of unitaries $\mathcal U$ is Pauli mixing if a uniformly sampled unitary $U\in \mathcal U$ maps any non-identity Pauli operator $\sigma_\u$ to any other $\sigma_{\u'} = U\sigma_\u U^\dagger$ with uniform distribution \cite{Webb_2016}.
Let us see that our model displays an approximate form of this property.


Each sequence of two-site Clifford unitaries $U_0, \ldots, U_{L-1}$ defines an evolution operator $W(t)$ via equations (\ref{eq:integer}-\ref{eq:halfinteger}), which maps each Pauli operator $\sigma_\u$ to another Pauli operator $\sigma_{\u'} = \lambda W(t) \sigma_\u W(t)^\dagger$.
This deterministic map $\u \mapsto \u'$ becomes probabilistic when we let $U_0, \ldots, U_{L-1}$ be random.
In this case, the probability that $\u$ evolves onto $\u'$ after a time $t$ is
\begin{equation}\label{Prob:final}
  P_t (\u'|\u) = \mathop{\mathbb E}_{\{U_x\}}
  \left| 2^{-NL} \tr(\sigma_{\u'} W(t) \sigma_{\u} W(t)^\dagger) \right|\ ,
\end{equation}
where we use the orthogonality of Paulies $\tr(\sigma_{\u'}\sigma_{\u}) = 2^{NL} \delta_{\u' \u}$.
The locality of the dynamics (see Figure~\ref{fig:CircuitFigure}) implies that only operators inside the light cone of the initial operator $\sigma_\u$ have non-zero probability \eqref{Prob:final}.
For example, if the initial operator $\sigma_\u$ is supported at the origin $x=0$, then after a time $t$, the evolved operator $\sigma_{\u'}$ must be fully supported inside the light cone $-2t+1 \leq x\leq 2t$. This means that the corresponding phase space vector $\u'$ is in the causal subspace 
\begin{equation}\label{eq:cc}
  \u' \in 
  \bigoplus_{x \in [-2t+1, 2t]} \mathcal V_x\ .
\end{equation}
The time $t$ at which the causal subspace becomes the whole system is the scrambling time $t_{\rm scr} = L/4$.
Section~\ref{sec:scrambling time} discusses the physical significance of this time scale.

When the distribution~\eqref{Prob:final} is approximately uniform inside the light cone we say that the random dynamics displays approximate Pauli mixing.
Let $Q_{t}(\u')$ denote the uniform distribution over all non-zero vectors $\u'$ in the causal subspace \eqref{eq:cc}, therefore after the scrambling time $t\geq t_{\rm scr}$, $Q_{t}(\u')$ is the uniform distribution over all non-zero vectors in the total phase space $\mathcal V_{\rm chain}$.
The following theorem from Section \ref{app:PauliMixing_Integer} proves approximate Pauli mixing for initially local operators.

\vspace{6mm}
\noindent
{\bf Theorem} \ref{lemma:lightcone} (Approximate Pauli mixing). If the initial Pauli operator $\sigma_\u$ is supported at site $x=0$ then the probability distribution~\eqref{Prob:final} for its evolution $\sigma_{\u'}$ is close to uniform inside the light cone, that is, 
\begin{equation}\label{res:ergod_intro}
  \sum_{\u'} 
  \left| P_t(\u'|\u) - Q_t(\u') \right|
  \ \leq\ 130\times t^2\, 2^{-N}\ ,
\end{equation}
for any integer or half-integer time $t \in [1/2 , 2 t_{\rm scr} ]$.
An analogous statement holds for any other initial location $x\neq 0$.

\vspace{6mm}

\noindent
The above bound is useful in the large-$N$ limit ($N\gg \log t$).
In the opposite regime ($N\ll \log L$) mixing cannot take place, since the system displays a strong form of localisation, in which local operators are mapped onto quasi-local operators.
This phenomenon is illustrated in Figure~\ref{fig:localisation} and detailed in Section~\ref{sec:StrongLocalisation}.

The error \eqref{res:ergod_intro} increases with time due to time correlations and dynamical recurrences (see Section~\ref{sec:scrambling time}).
Hence, as time goes on the character of the system is less mixing, which is the opposite of what happens in time-dependent dynamics (see Section~\ref{sec:t dep}).
Also note that at integer times $t$ our model can be considered to be time-independent (instead of time-periodic) with discrete time.

The above mixing result only applies to local initial operators. Next, we present a different result that applies to a large class of non-local initial operators.
However, due to the complexity of the problem, we only analyse their evolution inside a region $\mathcal R = \{1, \ldots, \Ls\} \subset \mathbb Z_L$.

\vspace{6mm}
\noindent
{\bf Theorem} \ref{lem:IntegerTime_SubsystemDistance}. 
Consider an initial vector $\u^0 \in \V_{\rm chain}$ with non-zero support in all lattice sites ($\u_x^0 \neq\zero$ for all $x\in \mathbb Z_L$).
Consider the evolved vector $\u^t = S(t) \u^0$ inside a region $x\in \{1, \ldots, \Ls\} \subseteq \mathbb Z_L$ where $\Ls$ is even and the time is $t \leq \frac{L - \Ls}{4}$.
If $\u^t_{[1,Ls]}$ is the projection of $\u^t$ in the subspace $\bigoplus_{x=1}^{\Ls} \V_x$ then
\begin{equation}
\sum_{\v \in \Z2^{2N\Ls}} \left| \pr\{\v=\u^t_{[1,Ls]}\} -\frac 1 {2^{2 N \Ls}}\right|
  \leq 
32\, t\, 2^{-N} ( 2 \Ls  +3^{\frac{\Ls}{2} + 1} ) +  4 L 2^{-2N} \ .
\end{equation}


\vspace{6mm}

\subsection{Pauli mixing at half-integer time}\label{sec:ergo1/2}

The evolution operator of our model $W(t)$ has an extra symmetry at half-integer time $t$ (see Section \ref{app:Twirling_PauliInvariance}).
This allows us to prove a mixing result that is stronger than those of the previous section. 
Specifically, the following theorem (from Section \ref{app:PauliMixing_SemiInteger}) applies to any initial Pauli operator instead of only local ones.

\vspace{6mm}
\noindent
{\bf Theorem} \ref{lem:tHalfDistance}. 
Let $\sigma_{\u'} = \lambda W(t) \sigma_\u W(t)^\dagger$ be the evolution of any initial Pauli operator $\sigma_\u \neq \unity$.
At any half-integer time $t$ larger than the scrambling time, in the interval $t \in [ t_{\rm scr} , 2  t_{\rm scr} ]$ the probability distribution~\eqref{Prob:final} for the evolved operator $\sigma_{\u'}$ is close to uniform, that is,
\begin{equation}\label{res:ergo}
  \sum_{\u'} 
  \left| P_t(\u'|\u) - Q_t(\u') \right|
  \ \leq\ 33 \times t\, L\, 2^{-N}\ .
\end{equation}

\vspace{6mm}
\noindent
The fact that mixing is more prominent at half-integer multiples of the period is not restricted to Clifford dynamics, since it applies to a large class of periodic random quantum circuit or Floquet dynamics with disorder.
In particular, it holds in any circuit where the two-site random interaction $U_x$ includes one-site random gates $V_x$ that are a 1-design.
That is, when the random variable $U_x$ follows the same statistics than the random variable $U'_x = U_x (V_x \otimes V_{x+1})$. 
This fact could be useful for implementing pseudo-random unitaries in quantum circuits with a periodic driving. 

\subsection{Pseudo-random unitaries}
\label{sec:2Design}

In this section we prove a consequence of the previous result: the evolution operator $W(t)$ at half-integer times $t$ is hard to physically distinguish from a Haar-random unitary $U \in {\rm SU}(2^{NL})$ when the available measurements are Pauli operators.
More precisely, imagine that it is given a unitary transformation $V$ which has been sampled from either the set of evolution operators $\{W(t)\}$ or the full unitary group ${\rm SU}(2^{NL})$.
The task is to choose a state $\rho$, process it with the given transformation $\rho \mapsto V\rho V^{\dagger}$, measure the result with a Pauli operator $\sigma_\u$, and guess whether $V$ has been sampled from the set of evolution operators $\{W(t)\}$ or from the full unitary group ${\rm SU}(2^{NL})$.
In order to sharpen this discrimination procedure, two uses of the transformation $V$ are permitted, which allows for feeding each of them with half of an entangled state $\rho$ (describing two copies of the system).
The following result tells us that, in the large-$N$ limit, the optimal guessing probability for the above task is almost as good as a random guess. (Recall that a random guess gives $p_{\rm guess} = 1/2$). 
The proof is given in Section \ref{app:ApproxDesignSemiTime}.

\vspace{6mm}
\noindent
{\bf Theorem} \ref{thm:approx2Design}. 
Consider the task of discriminating between two copies of $W(t)$ and two copies of a Haar-random unitary $U$ with measurements restricted to Pauli operators, when $t \in [ t_{\rm scr} , 2  t_{\rm scr} ]$ is half-integer. The success probability for correctly guessing the given pair of unitaries satisfies

\begin{align} 
  \nonumber
  p_{\rm guess}
  =\ &\frac 1 2 + \frac 1 4 
  \ \max_{\rho, \u, \v}\ 
  \tr\!\left(\!\sigma_\u \otimes \sigma_\v\! \left[
  \mathop{\mathbb E}_{W(t)}
  W(t)^{\otimes 2} \rho\, W(t)^{\otimes 2 \dagger}
  -
  \int_{{\rm SU}(d)} \hspace{-7mm} dU\,
  U^{\otimes 2} \rho\, U^{\otimes 2 \dagger}
  \right]\right)
  \\ \label{eq:P meas body_intro}
  \leq\ & 
  1/2 +9\, t L 2^{-N}  \ .
\end{align}

\vspace{6mm}

\noindent
The proof that the optimal guessing probability is given by formula \eqref{eq:P meas body_intro} can be found in \cite{nielsen_chuang_2010}.
If in Theorem \ref{thm:approx2Design}, measurements were not restricted then $W(t)$ would be an $(8\, t L 2^{-N})$-approximate unitary 2-design.
The precise definition of approximate 2-design allows for using an ancillary system in the discrimination process \cite{Eisert_2007}.
However, we have not included this ancillary system in Theorem \ref{thm:approx2Design} because it does not provide any advantage. 


\subsection{Strong localisation ($N\ll \log L$)}
\label{sec:StrongLocalisation}

The model under consideration has the property that certain combinations of gates in consecutive sites (e.g.~$U_x, U_{x+1},\ldots, U_{x+l}$) generate right- or left-sided walls. These are defined as follows: a right-sided wall at site $x$ stops the growth towards the right of any operator that arrives at $x$ from the left, but it does not necessarily stop the growth towards the left of any operator that arrives at $x$ from the right.
The analogous thing happens for left-sided walls (see Figure~\ref{fig:localisation}).

These gate configurations have a non-zero probability, hence, they will appear in a sufficiently long chain with a typical circuit. Below we provide bounds to this probability. The inverse of this probability is the average distance between walls, which can be understood as the \emph{localisation length scale}, and it quantifies the width of the lightcones displayed in Figure~\ref{fig:localisation}.

\begin{figure}
    \hspace{12mm}
    \includegraphics[width=120mm]
    {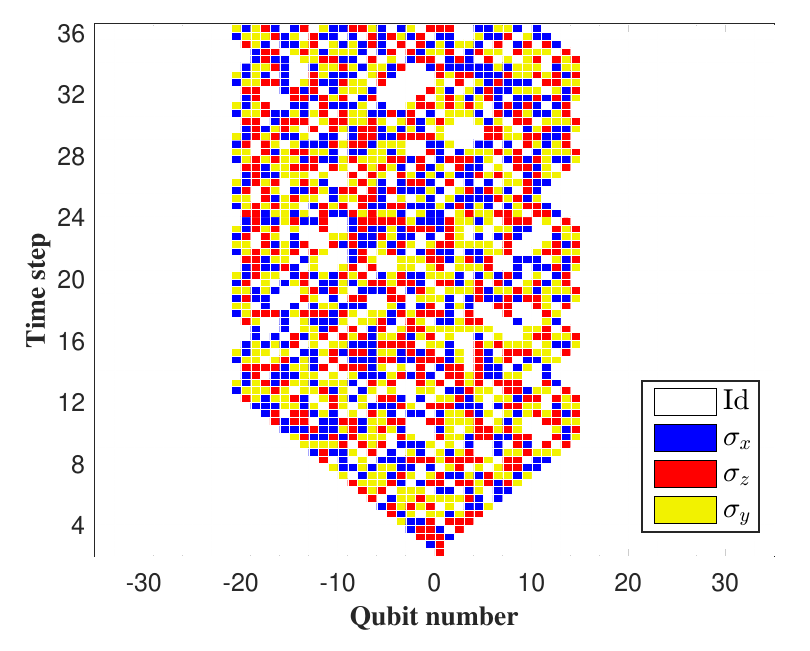}
    \vspace{-2mm}
    \caption{{\bf Strong localisation.} This figure displays the Heisenberg evolution of the initial operator $\sigma_z$ at site $x=1$.
    Each lattice site consists of one qubit ($N=1$) with first-neighbour interactions. 
    After a phase of linear growth the lateral wings collide with left- and right-sided walls with penetration length $l=1$, that confine the evolution for all times. This confinement affects all (not necessarily local or Pauli) operators between the two walls. Inside the confined region evolution seems to be mixing.}
    \label{fig:localisation}
\end{figure}

Each one-sided wall has some penetration length $l\geq 1$ into the forbidden region.
Suppose that a realisation of $U_{\rm chain}$ contains a right-sided wall at site $x=0$ with penetration length $l$. Then any operator with support on the sites $x\leq 0$ (and identity on $x>0$) is mapped by $(U_{\rm chain})^t$ to an operator with support on $x\leq l$ contained in a specific subspace within the interval $x\in [1,l]$ such that entering into region $x>l$ is impossible for all $t \geq 1$.
(The restriction to this subspace within the forbidden region can be seen in Figure \ref{fig:localisation} (with $l=1$) by the fact that the right-most points are either yellow followed by red, or white followed by white. And the left-most points are either blue followed by yellow, or white followed by white.)
An initial operator with support on the interval $x\in [1,l]$ which does not have the specific structure mentioned above can pass through and reach the side $x>l$.

Now let us characterize the pairs of gates $U_0, U_1$ (which act on sites $\{0,1\}$ and $\{1,2\}$ respectively) that generate a right-sided wall at $x=0$ with penetration length $l\leq 1$. 
Let $S_0, S_1$ be the phase-space representation of $U_0, U_1$. 
Next we use the fact that in phase space subsystems decompose with the direct sum (not the tensor product) rule, which allows to decompose $S_0, S_1$ in $2N$-dimensional blocks
\begin{equation}
  \label{eqb:S_x}
  S_x 
  =
  \left( \begin{array}{cc}
     A_x & B_x \\
     C_x & D_x 
  \end{array}\right) \ .
\end{equation}
The flow of information caused by $ S_x$ is easily seen by the action of $ S_x $ on the vector $ (\u_x, \u_{x+1})^T$:
\begin{equation}\label{eq:flow}
 \left( \begin{array}{cc}
     A_x & B_x \\
     C_x & D_x 
  \end{array}\right) \
  \left( \begin{array}{cc}
     \u_x \\
     \u_{x+1} 
  \end{array}\right) \ = 
  \left( \begin{array}{cc}
     A_x \u_x + B_x \u_{x+1} \\
     C_x \u_x + D_x \u_{x+1} 
  \end{array}\right) \
\end{equation}
Block $A_0$ ($D_0$) represents the local dynamics at site $x=0$ ($x=1$) in the first half step.
Block $C_0$ represents the flow from $x=0$ to $x=1$ in the first half step, and block $C_1$ represents the flow from $x=1$ to $x=2$ in the second half step. 


\begin{figure}
    \hspace{16mm}
    \includegraphics[width=80mm]
    {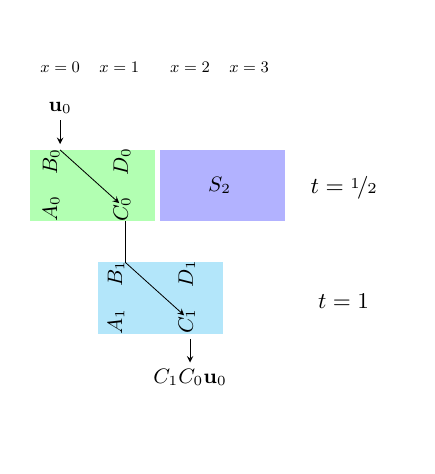}
    \vspace{0mm}
    \caption{The flow of information starting from $x=0$ at $t=0$ shows that with $C_1C_0=0$ there is no information reaching $ x=2 $ at $ t=1 $.}
    \label{fig:loc_1}
\end{figure}

\begin{figure}
    \hspace{16mm}
    \includegraphics[width=110mm]
    {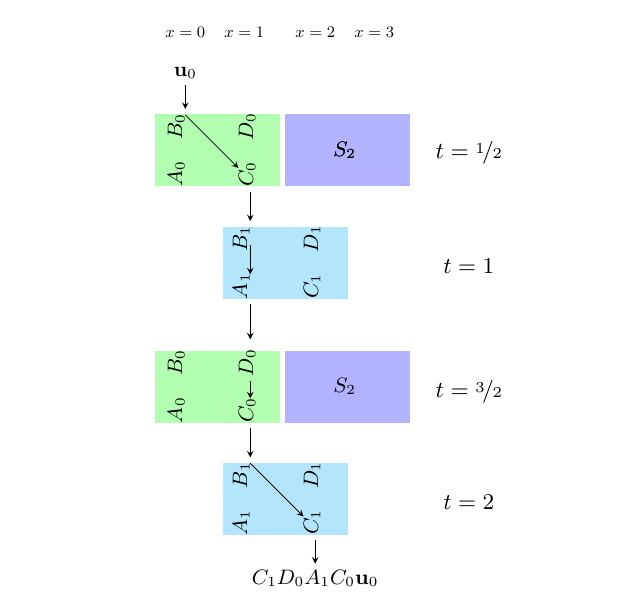}
    \vspace{0mm}
    \caption{As illustrated in Figure \ref{fig:loc_1}, the condition  $C_1C_0=0$ stops the information flow from $ x=0 $ at $ t=0$ to $ x=2 $ at $ t=1$. To prevent this from happening also at $ t=2$ we demand that $C_1 D_0 A_1 C_0 =0 $.}
    \label{fig:loc_2}
\end{figure}

Imposing that nothing arrives at $x=2$ after the first whole step amounts to $C_1 C_0 = 0$. This is illustrated in Figure \ref{fig:loc_1}.
Imposing that nothing arrives at $x=2$ after the first two whole steps amounts to 
\begin{equation}\label{twocond}
  C_1 D_0 A_1 C_0 =0 
  \ \mbox{ and }\  
  C_1 C_0 = 0\ .
\end{equation}
This is illustrated in Figure \ref{fig:loc_2}.
Finally, imposing that nothing arrives at $x=2$ after any number $(t+1)$ of whole steps amounts to 
\begin{equation}\label{t cond}
 C_1 (D_0 A_1)^t C_0 =0\, , 
\end{equation}
for all integers $t\geq 0$.
However, it is proven in Lemma \ref{less_integers} (Section \ref{app:CliffordLocalisation}) that this infinite family of conditions \eqref{t cond} is implied by the cases $t=0,1,\ldots, (2^{4N}-1)$. 
And for the simplest case $N=1$, Theorem \ref{lem:N1loc} shows that all conditions \eqref{t cond} follow from the two conditions \eqref{twocond}.

Equations \eqref{twocond} and \eqref{t cond} can be understood as characterising a pattern of destructive interference due to disorder which causes localisation. 

The following pair of Clifford unitearies $U_0, U_1$
\begin{equation}
U_0 = \frac{1}{\sqrt{2}}
\begin{pmatrix}
i & 0 & 0 & -i \\
0& i & -i & 0  \\
1 & 0 &  0 & 1 \\
0 & 1 & 1 & 0  \\
\end{pmatrix}, \ \ \ \ 
U_{1} =
\frac{1}{2}
\begin{pmatrix}
1 & -1 & 1 & 1 \\
-1 & 1 & 1 & 1 \\
1 & 1 &1 & -1 \\
1 & 1 & - 1 & 1 \\
\end{pmatrix} ,
\end{equation}
has phase-space representation 
\begin{equation}
S_0 =
\begin{pmatrix}
1 & 1 & 0 & 1 \\
0& 1 & 0 & 1  \\
1 & 0 &  1 & 0 \\
0 & 0 & 0 & 1  \\
\end{pmatrix}, \ \ \ \ 
S_{1} =
\begin{pmatrix}
1 & 0 & 0 & 1 \\
0 & 1 & 0 & 0 \\
0 & 1 &1 & 0 \\
0 & 0 & 0 & 1 \\
\end{pmatrix} .
\end{equation}
It can be checked that this pair of matrices satisfies conditions \eqref{twocond}, which implies \eqref{t cond}.

The following theorem provides the exact value of the probability for the appearance of a one-sided wall with penetration length $l=1$, in the case $N=1$.

\vspace{6mm}
\noindent
{\bf Theorem} \ref{lem:N1loc}.
For $N=1$ the conditions 
\begin{equation}
 C_{x+1} \left( D_x A_{x+1} \right)^k  C_x = 0 ,  
\end{equation}
for $k \in \{0,1,2,\ldots \}$ are implied by the two conditions 
\begin{equation} %
C_{x+1}  C_x = 0\quad \text{  and }\quad C_{x+1}  D_x A_{x+1}  C_x = 0 \ .
\end{equation}
Furthermore the probability of this is given exactly by
\begin{equation} \label{loc_ex_N=1}
\pr \{ C_{x+1} C_x = 0 ,  C_{x+1} D_x A_{x+1} C_x = 0 \}  = 0.12 \  \  , 
\end{equation}
which includes trivial localisation.

\vspace{6mm}
\noindent
The probability given in equation \eqref{loc_ex_N=1} is obtained numerically. 
It is worth mentioning that this model also displays walls with zero penetration length ($l=0$), which are necessarily two-sided.
These walls happen when a two-site gate $U_x$ is of product form $U_x = V_x \otimes V_{x+1}$.
This prevents the interaction between the two sides of the gate, and hence, it produces a trivial type of localisation.
The following theorem (proven in Section \ref{app:CliffordLocalisation}) shows that the probability of these trivial walls is very small.

\vspace{6mm}
\noindent
{\bf Theorem} \ref{eq:ProbE=0}
The probability that a Clifford unitary $U\in \C_{2N}$ is of product form is 
\begin{equation}
\frac 1 2\, 2^{-4N^2} \leq  \pr \{ U\ {\rm is \ product} \} \leq 2^{-4N^2} \ .
\end{equation}

\vspace{6mm}
\noindent 
We expect that $(l=0)$-walls are much less likely, except in the case $  N=1 $, than $(l\geq 1)$-walls. This would allow for a regime of $(L,N)$ where the system displays non-trivial localisation.

\subsection{Absence of localisation ($N \gg \log L$)}

The following theorem provides an upper bound for the probability that one-sided walls appear at a particular location. 
This upper bound implies that when $N \gg \log L$ a typical circuit has no localisation.

\vspace{6mm}
\noindent
{\bf Theorem} \ref{lem:RightBlockSuffCond}.
The conditions
\begin{equation} \label{loc_level_t}
 C_{x+1} \left( D_x A_{x+1} \right)^k  C_x = 0 
 \ \ \mbox{ for all }\ \  
 k \in \{0,1,2, \ldots \}\ ,    
\end{equation}
are sufficient to prevent all right-wards propagation past position $x$ at any time. The probability that this family of constrains holds is upper-bounded by
\begin{align} 
  & \pr \{ C_{x+1} \left( D_x A_{x+1} \right)^k C_x = 0 ,\ \forall k \in \mathbb{N}  \}  
  \nonumber \\ 
  \leq\ &\pr \{ C_{x+1} C_x = 0 \} \ \leq\  \frac{2N + 1}{(1-2^{-2N})^{2N}}\, 2^{2N - 2N^2}.  \label{loc bound}
\end{align}
By symmetry, left-sided walls have the same probabilities.

\vspace{6mm}
\noindent
If the system is finite ($L<\infty$), a sufficiently large $N$ will eliminate the presence of localisation in most realisations of the dynamics $U_{\rm chain}$. This fact is crucial in the mixing results of the previous sections.
Our previous results showing the mixing property in the regime $N\gg \log L$, suggest that in this regime the probability that the whole system has a wall of any type vanishes.

\section{Description of the model}
\label{app:DesciptionOfModel}

In this section we further specify the model analysed in this work.

\subsection{Locality, time-periodicity and disorder}


Consider a spin chain with an even number $L$ of sites and periodic boundary conditions.
Each site is labeled by $x\in \mathbb Z_L$ and contains $N$ qubits (Clifford modes), so the Hilbert space of each site has dimension $2^N$.
The dynamics of the chain is discrete in time, and hence, it is characterised by a unitary $U_{\rm chain}$, not a Hamiltonian.
Locality is imposed by the fact that $U_{\rm chain}$ is generated by first-neighbour interactions in the following way
\begin{equation}
  \label{eq:cal U}
  U_{\rm chain} 
  = 
  \left( \mbox{$\bigotimes_{x {\rm\, odd}}$} U_x \right)
  \left( \mbox{$\bigotimes_{x {\rm\, even}}$} U_x \right)
\end{equation}
where the unitary $U_x$ only acts on sites $x$ and $x+1$ ($\bmod L$ is understood).
The expression \eqref{eq:cal U} tells us that each time step decomposes in two half steps: in the first half each even site interacts with its right neighbor, and in the second half each even site interacts with its left neighbor.
This is illustrated in Figure \ref{fig:CircuitFigure}.

We define the evolution operator at integer and half-integer times $t\in \mathbb Z/2$ in the following way
\begin{equation}\label{def_W}
  W(t) = \left\{ \begin{array}{ll}
    (U_{\rm chain})^t &\ \mbox{integer $t$}\\
    ( \mbox{$\bigotimes_{x {\rm\, even}}$} U_x)(U_{\rm chain})^{t-1/2} &\ \mbox{half-integer $t$}
  \end{array}\right. \ .
\end{equation}
We understand that $t$ is half-integer when $t-1/2\in \mathbb Z$.

Translation invariance amounts to imposing that all $U_x$ with even $x$ are identical, and all $U_x$ with odd $x$ are identical too.
However, in this work we are interested in disordered systems, where the translation invariance is broken.
In fact, here we break the translation invariance in the strongest possible form, since each two-site unitary $U_x$ is independently sampled from the uniform distribution over the Clifford group.

\subsection{Phase-space description}

The phase space of the whole chain is
\begin{equation}
  \label{eq:V chain}
  \mathcal V_{\rm chain} 
  =
  \mbox{$\bigoplus_{x}$} \mathcal V_x
  \cong
  \Z2^{2NL}\ ,
\end{equation}
where $\mathcal V_x \cong \Z2^{2N}$ is the phase space of site $x$.
The phase-space representation of $U_x$ is the symplectic matrix $S_x \in \S_{2N}$, where $S_x$ acts on the subspace $\mathcal V_x \oplus \mathcal V_{x+1}$. Using this direct-sum decomposition we can write
\begin{equation}
  \label{eq:S_x}
  S_x 
  =
  \left( \begin{array}{cc}
     A_x & B_x \\
     C_x & D_x 
  \end{array}\right) \ ,
\end{equation}
where $A_x, B_x, C_x, D_x$ are $2N \times 2N$ matrices, with $ A_x:\mathcal V_{x} \rightarrow \mathcal V_x $, $ B_x:\mathcal V_{x+1} \rightarrow \mathcal V_x $, $ C_x:\mathcal V_{x} \rightarrow \mathcal V_{x+1} $, $ D_x:\mathcal V_{x+1} \rightarrow \mathcal V_{x+1} $.
The phase-space representation of $U_{\rm chain}$ given in \eqref{eq:cal U} is
\begin{equation}
  \label{eq:S chain}
  S_{\rm chain} 
  =
  \left( \mbox{$\bigoplus_{x\, {\rm odd}}$} S_x \right)
  \left( \mbox{$\bigoplus_{x\, {\rm even}}$} S_x \right)\ .
\end{equation}
Note that the tensor product becomes a direct sum, in analogy with the quantum optics formalism.
Using the single-site decomposition \eqref{eq:V chain} and \eqref{eq:S_x} we can write the two half steps in \eqref{eq:S chain} as 
\begin{align}
  \bigoplus_{x\, {\rm even}} \mathcal S_x 
  &=
  \left( \begin{array}{ccccccc}
    A_{0} & B_{0} &  &  &  &  &  \\
    C_{0} & D_{0} &  &  &  &  &  \\
     &  & A_2 & B_2 & & & \\
     &  & C_2 & D_2 & & & \\
     &  & & & \ddots & & \\
     &  & & & & A_{L-2} & B_{L-2}\\
     &  & & & & C_{L-2} & D_{L-2} \\
  \end{array}\right),
  \\
  \bigoplus_{x\, {\rm odd}} \mathcal S_x 
  &=
  \left( \begin{array}{ccccccc}
     D_{L-1} &  &  &  &  &  & C_{L-1} \\
     & A_1 & B_1 & & & &\\
     & C_1 & D_1 & & & &\\
     & & & \ddots & & &\\
     & & & & A_{L-3} & B_{L-3} &\\
     & & & & C_{L-3} & D_{L-3} &\\
     B_{L-1} &  &  &  &  &  & A_{L-1} \\
  \end{array}\right) ,
\end{align}
where the blank spaces represent blocks with zeros.
The ``phase space evolution operator" is then
\begin{equation}
  S(t) = \left\{ \begin{array}{ll}
    (S_{\rm chain})^t &\ \mbox{integer $t$}\\
    ( \mbox{$\bigoplus_{x {\rm\, even}}$} S_x)(S_{\rm chain})^{t-1/2} &\ \mbox{half-integer $t$}
  \end{array}\right. \ .
\end{equation}

\section{Random symplectic matrices}
\label{app:RandomCliffordDynamics}

In this section we discuss results relating to (uniformly) random symplectic matrices.
These results will then be used in the Section \ref{app:LocalDynamics_PauliMixing} to show that the random local circuit model we consider approximately satisfies the requirement of Pauli mixing: it maps any initial Pauli operator to the uniform distribution over all Pauli operators. 

An equivalent way to write the symplectic condition $S^T J S =J$ is that the columns of the matrix $S = (\u_1, \v_1, \u_2, \v_2, \ldots, \u_n, \v_n )$ satisfy
\begin{align}
  \label{simp const 1}
  \langle \u_i, \u_j\rangle = \langle \v_i, \v_j \rangle 
  &= 0\ ,
  \\ \label{simp const 2}
  \langle \u_i, \v_j \rangle &= \delta_{ij}\ .
\end{align}
We recall the notation $ \langle {\bf r} , {\bf s} \rangle \equiv {\bf r}^TJ {\bf s} $, for all $ {\bf r},{\bf s} \in \mathbb Z_2^{2n}$.
Using this, we can uniformly sample from $\S_n$ by sequentially generating the columns of $S$.

\begin{Lemma}\label{lemma:algorithm}
  The following algorithm allows to uniformly sample from the symplectic group $\S_n$.
\begin{enumerate}
  \item Generate $\u_1$ by picking any of the $(2^{2n}-1)$ non-zero vectors in $\mathbb Z_2^{2n}$.

  \item Generate $\v_1$ by picking any of the $2^{2n-1}$ vectors satisfying $\langle \u_1, \v_1 \rangle =1$.
  
  \item Generate $\u_2$ by picking any of the $(2^{2n-2}-1)$ non-zero vectors satisfying $\langle \u_1, \u_2 \rangle = \langle \v_1, \u_2 \rangle = 0$.

  \item Generate $\v_2$ by picking any of the $2^{2n-3}$ vectors satisfying $\langle \u_1, \v_2 \rangle = \langle \v_1, \v_2 \rangle = 0$ and $\langle \u_2, \v_2 \rangle =1$.
  
  \item Continue generating $\u_3, \v_3, \ldots, \u_n, \v_n$ in analogous fashion, completing the matrix $S = (\u_1, \v_1, \u_2, \v_2, \ldots, \u_n, \v_n )$.

\end{enumerate}
\end{Lemma}

\begin{proof}
We first look at the number of vectors $ (\u_1, \v_1, \u_2, \v_2, \ldots, \u_n, \v_n ) $, as stated above, that ensures $ S $ symplectic.

$\v_1$ has $ 2n $ components, since there is one constraint, the number of independent components is $ 2n-1 $, each component belongs to $ \mathbb Z_2 $, therefore the number of vectors $ \v_1 $ equals $2^{2n-1}$. Notice that since $ \u_1 $ is non vanishing, a vanishing $ \v_1 $ cannot solve $\langle \u_1, \v_1 \rangle =1$.

$\u_2$ must satisfy two constraints, the number of independent components is $ 2n-2 $, then there are $2^{2n-2}$ solutions, this includes also the case that $ \u_2 $ is vanishing, but since $ S $ must be full rank we need to exclude it, therefore there are $2^{2n-2}-1$ admissible $ \u_2 $ vectors. And so on.

We now proof that the distribution of simplectic matrices $ S $ generated with the algorithm is uniform. Let us see that the number of symplectic matrices whose first column is the non-zero vector $\u_1$ is independent of $\u_1$.

\begin{enumerate}
\item The number of vectors $\v_1$ satisfying $\langle \u_1, \v_1 \rangle =1$ is independent of which non-zero $\u_1$ we choose.
\item The number of vectors $\u_2$ satisfying $\langle \u_1, \u_2 \rangle =1$ and $\langle \u_2, \v_1 \rangle =0$ is independent of the pair $\u_1,\v_1$ (being both non-zero and $\langle \u_1, \v_1 \rangle =1$) that we choose. 
\item And analogously for $\v_2, \u_3$, and so on.
\end{enumerate}

This shows that the number of matrices having a fixed first column $\u_1$ is independent of $\u_1$. Therefore, all first columns $\u_1$ need to have the same probability. Using the steps 1,2,3 in a similar fashion we can analogously conclude that all second columns $\v_1$ need to have the same probability. And analogously, all vectors for column $ k $ (compatible with columns $1,2,...,k-1$) need to have the same probability. This shows the uniformity provided by the sampling algorithm of Lemma 6.

\end{proof}

\noindent
To obtain the above numbers, we use the fact that when $\langle \u, \v \rangle =1$ both, $\u$ and $\v$, are non-zero.
From these same numbers the next result follows.

\begin{Lemma}\label{lemma:order S}
  The order of the symplectic group is
\begin{equation}
  \label{eq:order S}
  |\S_n| = 
  (2^{2n} -1) 2^{2n-1} (2^{2n-2}-1) 2^{2n-3} \cdots (2^2 -1)2^1\ ,
\end{equation}
  and it satisfies
  \begin{equation} \label{eq:bound_order_S}
    a(n) \,2^{2n^2+n}
    \ \leq\  
    |\S_n| 
    \ \leq\ 
    b(n) \, 2^{2n^2+n}\
  \end{equation}
  with $0.64 < a(n) <b(n)<0.78$.
\end{Lemma}
The proof of \eqref{eq:order S} is a classic result to be found for example in \cite{Artin_1957}, it also directly follows from Lemma \ref{lemma:algorithm}. The proof of equation \eqref{eq:bound_order_S} is in the appendix \ref{app:sympl_order}.

Finally, the next Lemma shows that uniformly distributed symplectic matrices have random outputs.
\begin{Lemma}[Uniform output]
\label{lem:uniformity}
  If $S\in \S_n$ is uniformly distributed, then for any pair of non-zero vectors $\u, \u' \in \Z2^{2n}$ we have
\begin{equation}
  \label{uniform v1}
  \pr \{\u' = S\u\} 
  =
  (2^{2n}-1)^{-1}\ .
\end{equation}
\end{Lemma}

\begin{proof}
Let us first consider the case $\u = (1,0,\ldots, 0)^T$.
If we follow the algorithm of Lemma \ref{lemma:algorithm}, then the image of $(1,0,\ldots, 0)^T$ is uniformly distributed over the $(2^{2n}-1)$ non-zero vectors, and hence, it follows \eqref{uniform v1}. 
To show \eqref{uniform v1} for any given $\u$, take any $S_0 \in \S_n$ such that $S_0 \u = (1,0, \ldots, 0)^T$, and note that, if $S$ is uniformly distributed then so is $S S_0$. 
\end{proof}

\subsection{Rank of sub-matrices of $S$}

\begin{Lemma} \label{single rank decay}
Any given $S\in \S_{2n}$ can be written in block form 
\begin{equation}
  \label{eq:S decomp}
  S
  =
  \left( \begin{array}{cc}
     A & B \\
     C & D 
  \end{array}\right) ,
\end{equation}
according to the local decomposition $\Z2^{4n} = \Z2^{2n} \oplus \Z2^{2n}$. If $S$ is uniformly distributed this then induces a distribution on the sub-matrices $A, B, C, D$.
For each of them ($E= A, B, C, D$) the induced distribution satisfies
\begin{equation} \label{eq:upper_bound_rank_C}
  \pr \big\{ {\rm rank}\, E \leq 2n-k \big\} 
  \ \leq\  
  \min\{2^k,4\} 
  \frac {2^{- k^2}} {(1-2^{-2n})^k}
  \approx
  4 \times 2^{- k^2} \ .
\end{equation}
\end{Lemma}

\begin{proof}
We proceed by studying the rank of $C$ and later generalizing the results to $A, B, D$. 
Equation \eqref{eq:upper_bound_rank_C} is trivial for $ k=0 $, so in what follows we assume $ k \ge 1 $.
Let us start by counting the number of matrices $S\in \S_{2n}$ with a sub-matrix $C$ satisfying $C\u = \zero$ for a given (arbitrary) non-zero vector $\u \in \Z2^{2n}$.
Let $r$ denote the position of the last ``1" in $\u$, so that it can be written as: 
\begin{equation}
  \label{eq:form u}
  \u = \big(
  \underbrace{\u^1, \ldots, \u^{r-1}}_{r-1},
  1, \underbrace{0, \ldots, 0}_{2n-r} 
  \big)^T\ , 
\end{equation}
where $\u^1, \ldots, \u^{r-1} \in \{0,1\}$.
Then, the constraint $C\u = \zero$ can be written as
\begin{equation}   \label{eq:Cu=0}
\begin{cases}
  C_{i,1} = 0,  \hspace{0.5cm} \textrm{if} \, r = 1,  \hspace{0.5cm} \textrm{with} \, 1 \le i \le 2n\\
  C_{i,r} = \sum_{j=1}^{r-1} C_{i,j}\, \u^j, \hspace{0.5cm} \textrm{if} \, r > 1,  \hspace{0.5cm} \textrm{with} \, 1 \le i \le 2n
\end{cases}
\end{equation}
where $C_{i,j}$ are the components of $C$. \eqref{eq:Cu=0} reads as a constraint on the $ r$-th column of the matrix $ C $.

Next, we follow the algorithm introduced in Lemma \ref{lemma:algorithm} for generating a matrix $S\in \S_{2n}$ column by column, from left to right, and in addition to the symplectic constraints we include \eqref{eq:Cu=0}.
Constraint \eqref{eq:Cu=0} can be imposed by ignoring it during the generation of columns $1, \ldots, r-1$, completely fixing the rows $2n< i\leq 4n$ of the $r$ column, that corresponds to the $ r$-th column of the matrix $ C $, and again ignoring it during the generation of columns $r+1, \ldots, 4n$.
By counting as in Lemma \ref{lemma:order S} we obtain that the number of matrices $S\in \S_{2n}$ satisfying $C\u = \zero$ follows 
\begin{align}    
  \label{eq:number}
 & |\{ S\in \S_{2n} : C\u = \zero \}| \\
 & \le 
 \begin{cases}
 (2^{4n} - 1) (2^{4n-1}) \cdots (2^{4n-(r-2)} -1) 2^{2n -(r-1)} (2^{4n-r} -1) \cdots 2^1, \hspace{3mm} \textrm{r even} \\
 (2^{4n} - 1) (2^{4n-1}) \cdots 2^{4n-(r-2)}  2^{2n -(r-1)} 2^{4n-r}  \cdots 2^1, \hspace{3mm} \textrm{r odd}
 \end{cases}
 \nonumber
\end{align}.


Equation \eqref{eq:number} is an inequality because, for some values of the first $r-1$ columns of $S$ and the $r$-th column of $C$, it is impossible to complete the $r$-th column of $A$ satisfying the symplectic constraints (\ref{simp const 1}-\ref{simp const 2}). 

The probability that a random $S$ satisfies $C\u = \zero$ is 
\begin{equation}
    \pr \{ C \u = \zero\} 
    =  
    \frac{|\{ S\in \S_{2n} : C\u = \zero \}|}
    {|\S_{2n}|} \ .
\end{equation}
By noting that all factors in \eqref{eq:number} are the same as in \eqref{eq:order S} except for the factor at position $r$, we obtain
\begin{align}
    \nonumber
    \pr \{ C \u = \zero\} 
    &\leq
    \frac{2^{2n -(r-1)}}{2^{4n -(r-1)} - \alpha'}
    \leq
    \frac{2^{2n -(r-1)}}{2^{4n -(r-1)} -1}
    \\ \label{fixed u} &=
    \frac{2^{-2n}}{1 - 2^{(r-1)-4n}}
    \leq
    \frac{2^{-2n}}{1 - 2^{-2n}} \ ,
\end{align}
where $\alpha'=1$ if $r$ is odd and $\alpha'=0$ otherwise.
The last inequality above follows from $r\leq 2n$. The bound \eqref{fixed u} is correct also for $ r=1 $.
The fact that bound \eqref{fixed u} is independent of $r$ is crucial for the rest of the proof.

Next, we generalize bound \eqref{fixed u} to the case where $C \u_i = 0$ for $k$ given linearly-independent vectors $\u_i \in \{\u_1, \dots ,\u_k\}$.
To do this, we take the $2n \times k$ matrix $[\u_1, \dots ,\u_k]$ and perform Gauss-Jordan elimination, operating on the columns, to obtain a matrix $[\v_1, \dots ,\v_k]$ having column-echelon form. $ \{ C \u_1 = \zero, \ldots, C \u_k = \zero\} $ is equivalent to $ \{ C \v_1 = \zero, \ldots, C \v_k = \zero\} $, in fact only two operations are performed on the set $ \{ \u_1, \ldots, \u_k \} $ to obtain the set $ \{ \v_1, \ldots, \v_k \} $: changing the order of the vectors $ \{\u_1, \ldots, \u_k \} $, replacing a vector $ \u_j $ with the sum of $ \u_j $ with another vector $ \u_l $.
If we denote by $r_i$ the position of the last ``1" of $\v_i$, then column-echelon form amounts to $r_1 < r_2 < \cdots < r_k$.
Now we proceed as above to generate each column of $S$ satisfying the symplectic and the $C\v_i =\zero$ constraints.
This gives
\begin{align}
    \pr \{ C \u_1 = \zero, \ldots, C \u_k = \zero\} 
    \leq
    \frac{2^{2n -(r_1-1)}}{2^{4n -(r_1-1)} - \alpha'_1}
    \frac{2^{2n -(r_2-1)}}{2^{4n -(r_2-1)} - \alpha'_2}
    \cdots
    \frac{2^{2n -(r_k-1)}}{2^{4n -(r_k-1)} - \alpha'_k}
    \ ,
\end{align}
where $\alpha'_i \in \{0,1\}$.
Similarly as in \eqref{fixed u} we obtain
\begin{align}
    \label{fixed uk}
    \pr \{ C \u_1 = \zero, \ldots, C \u_k = \zero\} 
    \leq
    \frac{2^{-2nk}} {(1-2^{-2n})^k}
    \ .
\end{align}
If we multiply the above bound by the number $\N_k^{2n}$ of $k$-dimensional subspaces of $\Z2^{2n}$ (see appendix \ref{app:AdditionalLemmas}), then we obtain
\begin{align} \label{eq:upper_b}
    \nonumber
    \pr \{ \text{rank} (C) \leq 2n - k \} 
    &= 
    \N_k^{2n} \ 
    \pr \{ C \u_1 = \zero, \ldots, C \u_k = \zero\}
    \\ \nonumber &\leq
    \min\{2^k,4\}\, \frac{2^{2nk}}{2^{k^{2}}}\, 
    \frac{2^{-2nk}}{(1-2^{-2n})^k}  \ ,
    \\ &=
    \min\{2^k,4\}\, \frac {2^{-k^{2}}} {(1-2^{-2n})^k}  \ ,
\end{align}
where in the last inequality we used Lemma \ref{lemma:upperbound subspace}. 
Using Lemma \ref{S block permutation} (appendix \ref{app:AdditionalLemmas}), the above argument applies to any of the four sub-matrices $A, B, C, D$. 
The proof of equation \eqref{eq:upper_bound_rank_C} is then completed.
\end{proof}

\subsection{Rank of product of sub-matrices}

\begin{Lemma}
\label{le:diagonal}
  Let the random matrices $S_1, S_2, \ldots, S_r \in \S_{2n}$ be independent and uniformly distributed, which induces a distribution for the sub-matrices
\begin{equation}
  \label{eq:S_i decomp}
  S_i
  =
  \left( \begin{array}{cc}
     A_i & B_i \\
     C_i & D_i 
  \end{array}\right) .
\end{equation}
For any choice $E_i \in \{A_i, B_i, C_i, D_i\}$ for each $i\in \{1, \ldots, r\}$, we have
\begin{align}
  \label{eq: bound rank CCCCC}
  \pr\big\{ {\rm rank} (E_r \cdots E_1)\leq 2n-k\} 
  &\leq 
  \ 
  \frac {2^k} {(1-2^{-2n})^k} 
   \binom {k+r-1} {k} 
  2^{-\frac 1 2 k^2}
    \ .
\end{align}
\end{Lemma}

\begin{proof}
Before analyzing the rank of the product of $r$ independent random matrices $C_r \cdots C_1$, we start by a much simpler problem. 
Analyzing the rank of the product $CF$ where $C$ follows the usual $C$-distribution and $F$ is a fixed $2n \times 2n$ matrix with ${\rm rank}(F)= 2n-k_1$.
Noting that the input space of $C$ has dimension $2n-k_1$, from \eqref{eq:upper_b}, with $k_2 \equiv k-k_1 \geq 0$,   we obtain
\begin{align}
    \nonumber
    \pr \big\{ {\rm rank} (C F) \leq 2n-k \big\} 
    &\leq \N_{k_2}^{2n-k_1} 
    \pr\{C\u_1 =\zero, \ldots, C\u_{k_2}=\zero\}
    \\ \nonumber &\leq
    \min\{2^{k_2},4\}\, 
    \frac{2^{(2n-k_1)k_2}}{2^{k_2^2}}\, 
    \frac{2^{-2n k_2}}{(1-2^{-2n})^{k_2}}
    \\ \label{eq:CF rank} &\leq   
    \frac{2^{k_2 -k_1 k_2 -k_2^2}}{(1-2^{-2n})^{k_2}}
    \ ,
\end{align}

Proceeding in a similar fashion, we can analyze the product of two independent $C$-matrices.
To do so, we multiply two factors \eqref{eq:CF rank} and sum over all possible intermediate kernel sizes $k_1$, obtaing
\begin{align}
  \label{C22}
  \pr \big\{ {\rm rank} (C_2 C_1) \leq 2n-k \big\} 
  &\leq 
  \sum_{k_1=0}^k
  \frac{2^{k_2 -k_2 k_1 -k_2^2}}{(1-2^{-2n})^{k_2}}
  \ 
  \frac{2^{k_1 -k_1^2}}{(1-2^{-2n})^{k_1}}
  \\ &= 
  \sum_{k_1=0}^k
  \frac{2^{k -k_2 k_1 -k_1^2 -k_2^2}}{(1-2^{-2n})^k}
  \nonumber ,
\end{align}
where again $k_2= k-k_1$.

Equation \eqref{C22} works as follows: the matrix $ F $ in \eqref{eq:CF rank} has fix rank equal to $ 2n-k_1 $. That's the dimension of the input space of $ C $. In \eqref{C22} the input space of $ C_1 $ is the full space $ \Z2^{2n} $ that has dimension $ 2n $, therefore the factor $ \frac{2^{k_1 -k_1^2}}{(1-2^{-2n})^{k_1}} $ in \eqref{C22} equals the upper bound in \eqref{eq:CF rank} that is $ \frac{2^{k_2 -k_1 k_2 -k_2^2}}{(1-2^{-2n})^{k_2}} $  with $ k_1 = 0 $  and then with $ k_2 $ replaced by $ k_1 $. 
Moreover the input space of $ C_2 $ has dimension $ 2n-k_1$ that is like in equation \eqref{eq:CF rank}, that explains the first factor in \eqref{C22}.

Analogously, we can bound the rank of a product of $r$ independent random $C$-matrices as
\begin{align}
  \nonumber
  \pr\big\{ {\rm rank} (C_r \cdots C_1) \leq 2n-k \big\}  
  &\leq 
  \sum_{\{k_i\}} \prod_{i=1}^r
  \frac{2^{k_i -k_i \sum_{j=1}^i k_j}}{(1-2^{-2n})^{k_i}}
  \\ \label{eq: sum over paths} &= 
  \frac {2^k} {(1-2^{-2n})^k}
  \sum_{\{k_i\}} 
  2^{-\sum_{i=1}^r k_i \sum_{j=1}^i k_j}
  \ ,
\end{align}
where the sum $\sum_{\{k_i\}}$ runs over all sets of $r$ non-negative integers $\{k_1, \ldots, k_r\}$ such that $\sum_{i=1}^r k_i =k$.
These are all ways of sharing out $k$ units among $r$ distinguishable parts.
The number of all these sets equals:
\begin{align}
  \label{eq:numpaths}
  \sum_{\{k_i\}} 1
  =
  \binom {k+r-1} {r-1} 
  =
  \binom {k+r-1} {k} 
  \ .
\end{align}


Finally, for any set $\{k_1, \ldots, k_r\}$ we have
\begin{align}
  \nonumber
  k^2 
  &= 
  \sum_{i=1}^r \sum_{j=1}^r k_i k_j
  \leq 
  \sum_{i=1}^r \sum_{j=1}^i k_i k_j
  +
  \sum_{i=1}^r \sum_{j=i}^r k_i k_j
  \\ \label{eq:pathindep} &= 
  2\sum_{i=1}^r \sum_{j=1}^i k_i k_j
  \ .
\end{align}
Substituting \eqref{eq:numpaths} and \eqref{eq:pathindep} back in \eqref{eq: sum over paths} we obtain
\begin{align}
  \pr\big\{ {\rm rank} (C_r \cdots C_1) \leq 2n-k \big\}  
  \ \leq \ 
  \frac {2^k} {(1-2^{-2n})^k} 
   \binom {k+r-1} {k} 
  2^{-\frac 1 2 k^2}
  \ .
\end{align}
Once again, by using Lemma \ref{S block permutation}, this proof applies to all products of  sub-matrices $E\in \{A,B,C,D\}$.
\end{proof}

\begin{Lemma}\label{le:CCCCu=0}
  If the random variables $S_1, S_2, \ldots, S_r \in \S_{2n}$ and $\u \in \Z2^{2n}$ are independent and uniformly distributed it follows that
\begin{equation}
  \pr \big\{ E_r \cdots E_1 \u = \zero \big\} 
  \leq 
  8\, r\, 2^{-n}\ .
\end{equation}
being $ E_j \in \{A_j,B_j,C_j,D_j\} $ the subblocks of the symplectic matrices $S_1, ..., S_r$.
\end{Lemma}

\begin{proof}
If $M$ is a fixed $2n \times 2n$ matrix with ${\rm rank} M = 2n-k$ and $\u \in \Z2^{2n}$ is uniformly distributed, then
\begin{equation}
  \pr \big\{ M \u = \zero \big\} 
  =
  \frac {2^k} {2^{2n}}\ .
\end{equation}
Also, if ${\rm rank} M > 2n-k$ then
\begin{equation}\label{eq:prob zero}
  \pr \big\{ M \u = \zero \big\} 
  \leq
  \frac {2^{k-1}} {2^{2n}}\ .
\end{equation}
This inequality is useful for the following bound
\begin{align}
  \nonumber
  & \hspace{-10mm}
  \pr \big\{ C_r \cdots C_1 \u = \zero \big\} 
  \\ \nonumber & = 
  \pr \big\{ C_r \cdots C_1 \u = \zero \mbox{ and rank} (C_r \cdots C_1) >2n-k \big\} 
  \\ \nonumber & +
  \pr \big\{ C_r \cdots C_1 \u = \zero \mbox{ and rank} (C_r \cdots C_1) \leq 2n-k \big\} 
  \\ \nonumber & \leq
  \pr \big\{ C_r \cdots C_1 \u = \zero\, \big|\,  \mbox{rank} (C_r \cdots C_1) >2n-k \big\} 
  \\ \nonumber & +
  \pr\big\{ {\rm rank} (C_r \cdots C_1) \leq 2n-k \big\} 
  \\ &\leq
  2^{k-1-2n} + 
  \frac {2^k} {(1-2^{-2n})^k} 
   (1+ r)^k
  2^{-\frac 1 2 k^2}\ ,
\end{align}
where the last inequality uses \eqref{eq:prob zero} and Lemma \ref{le:diagonal} (and additional lemma \ref{lem:BinomialUpper} in the appendix \ref{app:AdditionalLemmas}).

Using
\begin{align} \label{root_of_e} \nonumber
 \frac{1}{(1-2^{-2n})^k} & \le \frac{1}{(1-2^{-2n})^{2n}} = \left(1+\frac{1}{2^{2n}-1}\right)^{2n} = \left(1+\frac{1}{2\left(2^{2n-1}-\frac{1}{2}\right)}\right)^{2n} \\
 & \le  \left(1+\frac{1}{4n}\right)^{2n} \le \sqrt{e} < 2 \ , 
\end{align}
we obtain
\begin{align}
  \label{prob CCCu=0}
  \pr \big\{ C_r \cdots C_1 \u = \zero \big\} 
  &\leq 
  2^{k-2n} + 
  2 \left(4r\right)^k
  2^{-\frac 1 2 k^2}
  =
  \epsilon\ ,
\end{align}
where the last equality defines $\epsilon$.
Note that the left-hand side above is independent of $k$. Hence, for each value of $k$ we have a different upper bound. We are interested in the tightest one of them. Therefore, we need to find a value of $k\in [1,2n]$ that makes the upper bound \eqref{prob CCCu=0} have a  small enough value.
This can be done by equating each of the two terms to $\epsilon/2$ as
\begin{equation}
  \label{eq:2 exponents}
  2^{k-2n} 
  = 
  2\left(4 r\right)^k
  2^{-\frac 1 2 k^2}
  =
  \frac \epsilon 2 \ .
\end{equation}
Isolating $k$ from the first and second terms gives
\begin{align}
  k &= 2n -\log_2\frac 2 \epsilon\ ,
  \\
  k &= \log_2 4r 
  +\sqrt{\log_2^2 4r +\log_2\frac 2 \epsilon + 1}
  \ ,
\end{align}
where we only keep the positive solution.
Equating the above two identities for $k$ we obtain
\begin{align}
  \nonumber
  n
  &= \frac{1}{2} \left(
  \log_2 4r +\log_2\frac 2 \epsilon
  +\sqrt{\log_2^2 4r + \log_2\frac 2 \epsilon + 1} \right)
  \\ \nonumber &\leq  \frac{1}{2}\left( 
  \log_2 4r +\log_2\frac 2 \epsilon
  +\log_2 4r +\sqrt{\log_2\frac 2 \epsilon + 1}
  \right)
  \\ &\leq  
  \log_2 4r +\log_2\frac 2 \epsilon
  \ ,
\end{align}
which implies
\begin{equation}
  \epsilon \leq 8\, r\, 2^{-n}\ .
\end{equation}
Substituting this into \eqref{prob CCCu=0} we finish the proof of this lemma.
\end{proof}




\section{Local dynamics is Pauli mixing}
\label{app:LocalDynamics_PauliMixing}
In this section, using the results from the Section \ref{app:RandomCliffordDynamics}, we will prove that in the regime $N \gg \log L$ the random dynamics of the model that we are considering maps any Pauli operator to any other Pauli operator with approximately uniform probability.

The time evolution of an initial vector $\u^0 \in \V_{\rm chain}$ at time $t$ is denoted by $\u^t = S(t) \u^0$.
If the initial vector is supported only at the origin $\u^0 \in \V_0$ then, as time $t$ increases, the evolved vector $\u^t$ is supported on the lightcone 
\begin{equation}\label{eq:lightcone}
  x \in 
  \{-(2t-1), -(2t-2), \ldots, 2t\} 
  \subseteq \mathbb Z_L\ .
\end{equation}

This leads to the definition of scrambling time: the length of the chain, $L$, is taken to be an integer multiple of 4, the system goes from $-L/2$ to $L/2$ with periodic boundary conditions. The scrambling time is the smallest time such that a perturbation supported at $x=0$ at $t=0$ evolves spreading its support to $\{-L/2+1,...,L/2\}$, therefore: 
\begin{equation}
 t_{\rm scr} \equiv \frac{L}{4}.
\end{equation}
The definition equally applies to the evolution of a vector $\u^t = S(t) \u^0$ as above.

Finally, we denote the projection of $\u$ on the local subspace $\V_x$ by $\u_x$.

\begin{Lemma}\label{lem:local operator}
  Consider a vector $\u^0 $ supported at the origin $ \V_0$ and its time evolution $\u^t$ for any $t\in \{1/2, 1, 3/2, \ldots, 2 t_{\rm scr}\}$. 
  The projection of $\u^t$ at the rightmost site of the lightcone $x=2t$ follows the probability distribution
\begin{equation}
  \label{eq_P(ux)}
  P(\u^t_{2 t}) = \left\{
  \begin{array}{cc}
    \frac {1-q_t} {2^{2N}-1} & \mbox{ if } \u^t_{2 t} \neq \zero
    \\
    q_t & \mbox{ if } \u^t_{2 t} = \zero
  \end{array}
  \right.\ ,
\end{equation}
where $q_t\leq 2\, t\, 2^{-2N}$.
The projection onto the second rightmost site $\u^t_{2 t-1}$ also obeys distribution \eqref{eq_P(ux)}.
\end{Lemma}

\begin{proof}
After half a time step the evolved vector $\u^{1/2}$ is supported on sites $x\in \{0, 1\}$ and it is determined by
\begin{equation}
  \u^{1/2}_0 \oplus \u^{1/2}_1
  = 
  S_0 (\u^0_0 \oplus \zero)\ .
\end{equation}
Lemma~\ref{lem:uniformity} tells us that the vector $\u^{1/2}_0 \oplus \u^{1/2}_1$ is uniformly distributed over all non-zero vectors in $\V_0 \oplus \V_1$.
This implies that the vector $\u^{1/2}_0$ (and the same for $\u^{1/2}_1$) satisfies
\begin{equation}
  \label{eq:probzero1}
  \pr\{\u^{1/2}_0 = \zero \} 
  = \frac {2^{2N}-1} {2^{4N}-1}
  \leq 2^{-2N} \ ,
\end{equation}
and has probability distribution of the form \eqref{eq_P(ux)} with $t=1/2$.

In the next time step we have
\begin{equation}
  \u^{1}_1 \oplus \u^{1}_2
  = 
  S_1 (\u^{1/2}_1 \oplus \zero)\ .
\end{equation}
Hence, if $\u^{1/2}_1 =\zero$ then $\u^{1}_1 = \u^{1}_2 =\zero$.
Also, applying again Lemma~\ref{lem:uniformity} we see that, if $\u^{1/2}_1 \neq \zero$, then $\u^{1}_1 \oplus \u^{1}_2$ is uniformly distributed over all non-zero values.
Putting these things together we conclude that $\u^{1}_1$ (and the same for $\u^{1}_2$) satisfies
\begin{align}
  \nonumber
  \pr\{\u^1_1 = \zero\} 
  &= 
  \pr\{\u_1^{1/2} = \zero\} 
  + 
  \pr\{\u_1^{1/2} \neq\zero\}\, 
  \pr\{\u^1_x = \zero | \u^{1/2}_1 \neq \zero\}
  \\ \nonumber &\leq 
  \pr\{\u_1^{1/2} = \zero\} 
  + 
  \pr\{\u_1^{1/2} \neq\zero\}\, 2^{-2N}
  \\ \label{eq:probzero2} &\leq 
  2\times 2^{-2N}
  \ ,
\end{align}
and has probability distribution of the form \eqref{eq_P(ux)} with $t=1$.

We can proceed as above, applying Lemma~\ref{lem:uniformity} to each evolution step
\begin{equation}
  \label{eq:recursive}
  \u^t_{2t-1} \oplus \u^t_{2t}
  = 
  S_{2t-1} (\u^{t-1/2}_{2t-1} \oplus \zero)\ ,
\end{equation}
for $t=1/2, 1, 3/2, 2, \ldots$
This gives us the recursive equation
\begin{align}
  \nonumber
  \pr\{\u_{2t}^t =\zero\}
  &=
  \pr\{\u_{2t-1}^{t-1/2} =\zero\} 
  +
  \pr\{\u_{2t-1}^{t-1/2} \neq\zero\}\,
  \pr\{\u_{2t}^t =\zero | \u_{2t-1}^{t-1/2} \neq\zero\}
  \\ \label{eq:rec rel}&\leq
  2t\times 2^{-2N}\ .
\end{align}
And the same for $\u^t_{2t-1}$.
Also, Lemma~\ref{lem:uniformity} implies that $\u^t_{2t-1}$ and $\u^t_{2t}$ follow the probability distribution \eqref{eq_P(ux)} for all $t=1/2, 1, 3/2, 2, \ldots, 2 t_{\rm scr}$.

For $t> 2t_{\rm scr}$ the recursion relation \eqref{eq:recursive} includes repeated matrices $S_x$. Hence the argument is no longer valid.
\end{proof}

\begin{Lemma}\label{lem:nonzero pure}
  If the initial vector $\u^0 \in \V_{\rm chain}$ is supported on all lattice sites ($\u_x^0 \neq\zero$ for all $x$) then the projection of its evolution $\u^t$ onto any site $x\in \mathbb Z_L$ satisfies
\begin{align}
  \label{eq:nonzero pure}
  \pr\!\left\{\u^t_x \neq \zero \right\}
  \geq 
  1- 16\, t\, 2^{-N} ,
\end{align}
for all $t\in \{1/2, 1, 3/2, \ldots, 2 t_{\rm scr}\}$.
\end{Lemma}

\begin{proof}
To prove this lemma we proceed similarly as in Lemma~\ref{lem:local operator}.
However, here, the recursive equation~\eqref{eq:recursive} need not have a $\zero$-input in the right system
\begin{equation}
  \label{eq:recursive 2}
  \u^t_{2t-1} \oplus \u^t_{2t}
  = 
  S_{2t-1} (\u^{t-1/2}_{2t-1} \oplus \u^{t-1/2}_{2t})\ .
\end{equation}
This difference in the premises does not change  conclusion \eqref{eq:probzero1}, due to the fact that bound \eqref{uniform v1} is independent of $\u_1^0$ being zero or not.
This gives \eqref{eq_P(ux)} for $t=1/2$. 
Also, using 
\begin{align}
  \nonumber
  \pr\{\u^1_2 = \zero\} 
  &= 
  \pr\{\u_1^{1/2} \oplus \u_2^{1/2} = \zero\} 
  + 
  \pr\{\u_1^{1/2} \oplus \u_2^{1/2}\neq\zero \mbox{ and } \u^1_2 = \zero \}
  \\ \nonumber &\leq 
  \pr\{\u_1^{1/2} = \zero\} 
  + 
  \pr\{\u^1_x = \zero \, \big|\,  \u_1^{1/2} \oplus \u_2^{1/2} \neq\zero\}
  \\ \label{eq:probzero3} &\leq 
   2^{-2N+1}
  \ ,
\end{align}
we obtain the same probability distribution as in \eqref{eq_P(ux)} for $t=1$, but under different premises. 
However, here there is a very delicate point.
As can be seen in Figure~\ref{fig:evolution}, the vector $\u^1_2$ is partly determined by $S_2$, and hence, it is not independent from $S_2$.
Crucially, the bound \eqref{eq:probzero2} for $\u^1_2$ holds regardless of the right input $\u^{1/2}_2$, and hence,  is independent of $S_2$.
This fact can be summarized with the following bound
\begin{equation}
  \label{eq:P(u2)}
  P(\u^1_2| S_2) = \left\{
  \begin{array}{cc}
    \frac {1-q_1} {2^{2N}-1} & \mbox{ if } \u^1_2 \neq \zero
    \\
    q_1 & \mbox{ if } \u^1_2 = \zero
  \end{array}
  \right.\ ,
\end{equation}
for any $S_2$, where $q_1\leq 2\, 2^{-2N}$.
That is, the correlation between $\u^1_2$ and $S_2$ can only happen through small variations of $q_1$.

\begin{figure}
    \hspace{2mm}
    \includegraphics[width = 16cm]{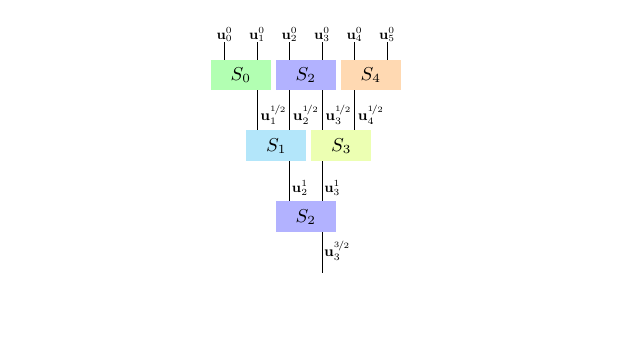}
    \vspace{-7mm}
    \caption{This figure shows that the causal past of $\u^1_2$ is partly determined by $S_2$. 
    Hence, at $t=3/2$, the input $\u^1_2$ of $S_2$ is not independent of $S_2$.
    This makes the exact probability distribution of $\u_3^{3/2}$ very complicated.
To overcome this problem we exploit the fact that $S_1$ only appears once in the past of $\u_3^{3/2}$. 
This allows to map the randomness of $S_1$ to $\u_3^{3/2}$ for most of the values of the other gates $(S_1, S_2, S_3, S_4)$. More concretely, we can apply Lemma~\ref{le:CCCCu=0} to the case $r=1$, $E_1 = C_2$ and $\u = \u_2^{1}$, resulting in that $\u^{3/2}_3$ is approximately uniform. 
    }
    \label{fig:evolution}
\end{figure}

For $t> 1$, the inputs in \eqref{eq:recursive 2} are not independent of the matrix $S_{2t-1}$, as illustrated in Figure~\ref{fig:evolution}, and hence, Lemma~\ref{lem:uniformity} cannot be applied.
If we restrict equation \eqref{eq:recursive 2} to the rightmost output ($x=2t$) then we obtain
\begin{align}
  \nonumber
  \u^t_{2t}
  &= 
  C_{2t-1}\u^{t-1/2}_{2t-1} + D_{2t-1}\u^{t-1/2}_{2t}
  \\ \label{eq:recursive 3}
  &= 
  C_{2t-1}\u^{t-1/2}_{2t-1} + \v^{t-1/2}\ ,
\end{align}
where the vector $\v^{t-1/2} = D_{2t-1}\u^{t-1/2}_{2t} \in \Z2^{2N}$ is not independent of $C_{2t-1}$.
Expanding this recursive relation we obtain
\begin{align}
  \nonumber
  \u^t_{2t}
  &= 
  C_{2t-1} C_{2t-2} \u^{t-1}_{2t-2} 
  +C_{2t-1} \v^{t-1} +\v^{t-1/2}
  \\ \label{ut wt} &= 
  C_{2t-1} \cdots C_2 \u^{1}_2 
  +{\bf w}^t\ ,
\end{align}
where the random vector
\begin{equation}
\label{eq:random wt}
  {\bf w}^t
  =
  C_{2t-1} \cdots C_{3} \v^{1}
  +\cdots 
  +C_{2t-1} C_{2t-2} \v^{t-3/2}
  +C_{2t-1}\v^{t-1}
  +\v^{t-1/2}
\end{equation}
is not independent of the matrices $C_{2t-1}, \ldots, C_2$.
Crucially, the bound \eqref{eq:P(u2)} for the distribution of $\u_2^1$ is independent of all these matrices.

Let us introduce the uniformly distributed random variable $\u \in\Z2^{2N}$, which is independent of all gates $S_x$.
According to \eqref{eq:P(u2)}, the random variable $\u_2^1$ is close to uniform, hence, it has small statistical distance with $\u$,
\begin{align}
  \nonumber
  {\rm d}(\u_2^1, \u)
  &=\sum_{\u_2^1} \left| P(\u_2^1) -2^{-2N} \right| 
  \\ \nonumber &= \left| q_1 -2^{-2N} \right| 
  +\left( 2^{2N}-1\right) \left|\frac {1-q_1} {2^{2N}-1} -2^{-2N} \right|
  \\ 
  \nonumber
  &=
  2\left| q_1 -2^{-2N} \right|
  \leq 2\, 2^{-2N}\ .
\end{align}
For any event $\mathcal E \subseteq \Z2^{2N}$ we have that
\begin{align}
  \nonumber
  \pr\{\u_2^1 \in \mathcal E\}
  &=
  \sum_{\u_2^2 \in \mathcal E} P(\u_2^1)
  \leq
  \sum_{\u_2^2 \in \mathcal E} 
  \left(2^{-2N} 
  +\left|P(\u_2^1) -2^{-2N}\right| \right)
  \\ \nonumber &\leq 
  \pr\{\u \in \mathcal E\}
  + {\rm d}(\u_2^1, \u)
  \\ \label{eq:E18} &\leq 2\, 2^{-2N}
  +\pr\{\u \in \mathcal E\}\ .
\end{align}
Next, we apply this bound to the particular event $\mathcal E$ defined by $\u^t_{2t} = \textbf 0$, which can be written as $C_{2t-1} \cdots C_2 \u_2^1 = {\bf w}^t$ by using \eqref{ut wt}.
Putting all this together we obtain
\begin{align}
  \label{eq:prob u=0}
  \pr\{\u^t_{2t} =\zero\}
  \leq 2\, 2^{-2N} +
  \pr\{ C_{2t-1} \cdots C_2 \u = {\bf w}^t \}\ ,
\end{align}
where the random variable $\u \in \Z2^{2N}$ is uniformly distributed and independent of ${\bf w}^t$ and $C_i$ for all $i\in \{2t-1, \ldots, 2\}$.
This has the advantage that now we can invoke Lemma~\ref{le:CCCCu=0}.
Lets start by rewriting
\begin{equation}
  \pr\{ C_{2t-1} \cdots C_2 \u =\mathbf w^t \}
  =
  \mathop{\mathbb E}_{C_i, {\bf w}^t} 
  \mathop{\mathbb E}_{\, \u}
  \delta \!\left( C_{2t-1} \cdots C_2 \u  - {\bf w}^t, \textbf{0} \right)\ , \nonumber  
\end{equation}
and consider the average $\mathop{\mathbb E}_\u \delta\!\left( C_{2t-1} \cdots C_2 \u - {\bf w}^t, \zero \right)$ for a fixed value of the variables ${\bf w}^t$ and $C_i$.
If the vector ${\bf w}^t$ is not in the range of the matrix  $(C_{2t-1} \cdots C_2)$ then the average is zero.
If the vector ${\bf w}^t$ is in the range of the matrix  $(C_{2t-1} \cdots C_2)$ then there is a vector $\tilde {\bf w}$ such that ${\bf w}^t = (C_{2t-1} \cdots C_2) \tilde {\bf w}$.
Then we can write the average as
\begin{align}
  \mathop{\mathbb E}_{\, \u}
  \delta\!\left( C_{2t-1} \cdots C_2 \u 
  -{\bf w}^t , \textbf{0} \right)
  &=
  \mathop{\mathbb E}_{\, \u}
  \delta\!\left( C_{2t-1} \cdots C_2 (\u +\tilde {\bf w}), \zero \right) \nonumber
  \\ &=
  \mathop{\mathbb E}_{\, \u}
  \delta\!\left( C_{2t-1} \cdots C_2 \u , \zero \right)\ , \nonumber
\end{align}
where the last equality follows from the fact that the random variable $\u +\tilde {\bf w}$ is uniform and independent of $C_i$, likewise $\u$.
Combining together the two cases for ${\bf w}^t$ we can write
\begin{align}
  \nonumber
  \pr\{ C_{2t-1} \cdots C_2 \u ={\bf w}^t \}
  &\leq 
  \mathop{\mathbb E}_{C_i} 
  \mathop{\mathbb E}_{\, \u}
  \delta\!\left( C_{2t-1} \cdots C_2 \u 
  ,\zero \right)
  \\ \nonumber &= 
  \pr\{ C_{2t-1} \cdots C_2 \u =\zero \}
  \\ \label{eq:cccu=w} &\leq 
  8 (2t-2) 2^{-N}\ ,
\end{align}
where the last step follows from Lemma~\ref{le:CCCCu=0}.
Substituting this back into \eqref{eq:prob u=0} we obtain
\begin{align}
  \label{eq:prob u=0 3}
  \pr\{\u^t_{2t} =\zero\}
  \leq 
  2\, 2^{-2N} +16\, (t-1) 2^{-N}
  \leq 
  16\, t\, 2^{-N}\ .
\end{align}

If we repeat all the steps of this proof since \eqref{eq:recursive 3} substituting $\u^t_{2t}$ for $\u^t_{2t-1}$, then we arrive at  
\begin{align*}
  \pr\{\u^t_{2t-1} =\zero\}
  \leq 2\, 2^{-2N} +
  \pr\{ A_{2t-1} C_{2t-2} \cdots C_2 \u = {\bf w}^t \}\ ,
\end{align*}
instead of \eqref{eq:prob u=0}.
But Lemma~\ref{le:CCCCu=0} also applies in this case, giving the bound 
\begin{align*}
  \pr\{ A_{2t-1} C_{2t-2}\cdots C_2 \u =\zero \}
  \leq 
  8 (2t-2) 2^{-N}\ ,
\end{align*}
which implies 
\begin{align*}
  \pr\{\u^t_{2t-1} =\zero\}
  \leq 
  16\, t\, 2^{-N}\ .
\end{align*}
Also, since the premises of this lemma are invariant under translations in the chain $\mathbb Z_L$, then the conclusions hold for all $x \in \mathbb Z_L$.
\end{proof}

In order to prove the next theorem it is important to note the following remark. 
The bound \eqref{eq:prob u=0 3} requires that either $\u^0_0 \neq \zero$ or $\u^1_0 \neq \zero$, but does not require $\u^0_x \neq \zero$ for $x>1$.

\begin{Lemma}\label{lemma:allnonzero}
  After the scrambling time $t\in [t_{\rm scr}, 2 t_{\rm scr}]$, with $ t $ integer or half-integer, the evolved vector $\u^t = S(t) \u^0$ is non-zero at each lattice site with probability
\begin{align}
  \label{eq:allnonzero}
  \pr\big\{\u^t_x \neq \zero ,\forall\, x \in \mathbb Z _L\big\}
  \geq 
  1- 16\, t\, L\, 2^{-N} ,
\end{align}
for any initial non-zero vector $\u^0 \in \V_{\rm chain}$.
\end{Lemma}

\begin{proof}
Let $\mathcal F(\u^0) \subseteq \mathbb Z_L \times \mathbb N$ be the set of spacetime points consisting of the causal future of the sites $x' \in \mathbb Z_L$ where the initial vector $\u^0$ has support ($\u^0_{x'} \neq \zero$).
For example, if the initial vector is supported in the origin of the chain $\u^0\in \mathcal V_0$ then the causal future is given by the light cone \eqref{eq:lightcone}.

The main objective in this proof is to bound the probability of $\u_x^t \neq \zero$ for any fixed site $x\in \mathbb Z_L$ and time $t\in [t_{\rm scr}, 2 t_{\rm scr}]$.
For the sake of simplicity, let us start by considering the case of $x$ odd and $t$ integer. In this case, the left-most spacetime points in the causal past of $(x,t)$ that are also contained in $\mathcal F(\u^0)$ are
\begin{equation}\label{eq:first seq}
  (x-1,t-1/2), \ldots, (x-n, t-n/2), \ldots, (x_{\rm e}, t_{\rm e})\ .
\end{equation}
We have that either $t_{\rm e} =0$ or $t_{\rm e} > 0$.
In the first case ($t_{\rm e} =0$) we have that $\u^0$ has support on $x_{\rm e}$ or $x_{\rm e}+1$. 
And we can prove
\begin{equation}\label{eq:basic bound}
  \pr\{ \u_x^t = \zero\} \leq 16\, t\, 2^{-N}\ ,  
\end{equation}
by applying the same procedure as in Lemma~\ref{lem:nonzero pure}. 
Note that the possibility that $\u^0_{x'} =\zero$ for $x'> x_{\rm e}+1$ does not affect the argument (see last paragraph in the proof of Lemma~\ref{lem:nonzero pure}).

In the second case ($t_{\rm e} > 0$), the sequence \eqref{eq:first seq} can be continued by including the following points from $\mathcal F(\u^0)$,
\begin{equation}\label{eq:second seq}
  (x_{\rm e}, t_{\rm e}-1/2),\ldots, (x_{\rm e}+n, t_{\rm e}-1/2-n/2),\ldots, (x_0-1,1/2),
  \left\{\begin{array}{ll}
    (x_0-1,0)
    \\
    (x_0,0)
  \end{array}\right.\ ,  
\end{equation}
where the last element is chosen so that it belongs to $\mathcal F(\u^0)$.
If $\u^0$ has support on both $(x_0-1,0)$ and $(x_0,0)$ then the choice is arbitrary.
Here, for the sake of concreteness, we assume that $\u_{x_0}^0 \neq \zero$ and take $(x_0,0)$ as the last point of the sequence.
The subindex e stands for ``elbow", because it labels the point where the sequence \eqref{eq:first seq} changes direction to \eqref{eq:second seq} (see Figure~\ref{fig:elbow}).

\begin{figure}
  \includegraphics[width = 14cm]{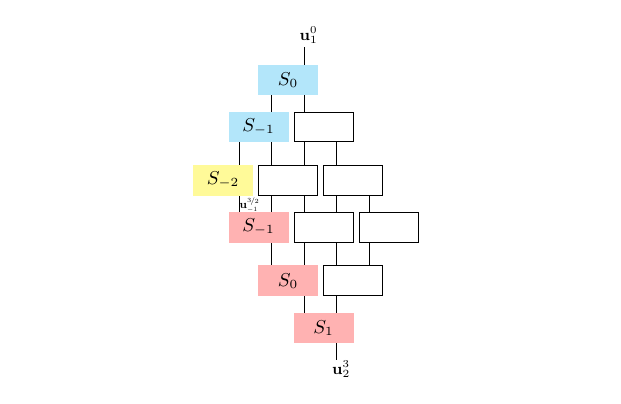}
  \vspace{-1cm}
  \caption{This figure represents the evolution of an initially local operator $\u^0 \in \V_1$ at site $x=1$. 
  The figure only displays gates $S_x$ that are in the intersection of the causal future of the initial location $x=1$ and the causal past of the chosen point $\u_2^3$. 
  The probability of $\u_2^3 =\zero$ is bounded by analysing the sequence of coloured gates, which has an ``elbow" at location $(x_{\rm e}, t_{\rm e}) = (-1,3/2)$.
  The analysis of blue gates uses Lemma~\ref{lem:local operator}, and that of red gates uses Lemma~\ref{lem:nonzero pure}.
  The key feature of the bound is that the yellow gate $S_{-2}$ only appears once.}
  \label{fig:elbow}
\end{figure}

Now we can write our chosen vector $\u_x^t$ as
\begin{align}
  \nonumber
  & \u_{x_{\rm e}}^{t_{\rm e}-1/2} 
  = 
  B_{x_{\rm e}} \cdots B_{x_0-2} B_{x_0-1} 
  \u_{x_0}^0\ , 
  \\ \nonumber
  & \u_{x_{\rm e}}^{t_{\rm e}} 
  = 
  D_{x_{\rm e}-1} 
  \u_{x_{\rm e}}^{t_{\rm e}-1/2}\ ,
  \\ \nonumber
  & \u_x^t 
  = 
  C_{x-1} \cdots C_{x_{\rm e}+1} C_{x_{\rm e}}
  \u_{x_{\rm e}}^{t_{\rm e}} + \w \ , 
\end{align}
where the random vector $\w$ is correlated with $B_{x_{\rm e}}, \ldots, B_{x_0-1}$ and $C_{x-1}, \ldots, C_{x_{\rm e}}$ but not with $D_{x_{\rm e}-1}$.
Vector $\w$ is analogous to $\w^t$, defined in \eqref{eq:random wt}.
Note also that the random matrices $B_{x_{\rm e}}, \ldots, B_{x_0-1}$ are not independent from $C_{x-1}, \ldots, C_{x_{\rm e}}$, but that $D_{x_{\rm e}-1}$ is independent from all the rest.
(Figure~\ref{fig:elbow} contains an example where the gates associated to $B_{x_{\rm e}}, \ldots, B_{x_0-1}$ are in blue, those of $C_{x-1}, \ldots, C_{x_{\rm e}}$ in red, and that of $D_{x_{\rm e}-1}$ in yellow.)

Now we can start constructing our bound as
\begin{align}
  \nonumber
  \pr\{\u^t_{x} =\zero\}
  &=
  \pr\{\u^t_{x} =\zero \mbox{ and } \u_{x_{\rm e}}^{t_{\rm e}-1/2}=\zero\}
  + \pr\{\u^t_{x} =\zero \mbox{ and } \u_{x_{\rm e}}^{t_{\rm e}-1/2}\neq \zero\} 
  \\ \label{eq:s1} &\leq
  \pr\{\u_{x_{\rm e}}^{t_{\rm e}-1/2}=\zero\}
  + \pr\{\u^t_{x} =\zero \mbox{ and } \u_{x_{\rm e}}^{t_{\rm e}-1/2}\neq \zero\} 
\end{align}
The first term can be bounded with the recursive equation \eqref{eq:rec rel} as
\begin{equation}
  \pr\{\u_{x_{\rm e}}^{t_{\rm e}-1/2}=\zero\}
  \leq 2 (t_{\rm e}-1/2) 2^{-2N}\ . \nonumber
\end{equation}
The second term can be bounded by using the independence of $D_{x_{\rm e}-1}$, the fact that $\u_{x_{\rm e}}^{t_{\rm e}-1/2}$ is not zero, and proceeding in a manner similar to \eqref{eq:E18} and \eqref{eq:prob u=0}.
Therefore, we again introduce the uniformly distributed random vector $\u \in\Z2^{2N}$, which is independent of all gates $S_{x_{\rm e}}, S_{x_{\rm e}+1}, \ldots, S_{x-1}$.
The statistical distance between $\u_{x_{\rm e}}^{t_{\rm e}}$ and $\u$, conditioned on $\u_{x_{\rm e}}^{t_{\rm e}-1/2}\neq \zero$, is 
\begin{align}
  \nonumber
  {\rm d}(\u_{x_{\rm e}}^{t_{\rm e}},\u) 
  =&\ 
  \sum_{\u_{x_{\rm e}}^{t_{\rm e}}}
  \left| P(\u_{x_{\rm e}}^{t_{\rm e}} |
  \u_{x_{\rm e}}^{t_{\rm e}-1/2}\neq \zero)
  -2^{-2N}\right|
  \\ \nonumber = 
  &\left(2^{2N}-1\right) 
  \left|\frac {2^{2N}}{2^{4N}-1} -2^{-2N}\right| 
  + \left|\frac {2^{2N}-1}{2^{4N}-1} -2^{-2N}\right| 
  \\ \label{eq:sdist} \leq &\ 
  2^{1-4N}\ ,
\end{align}
where we have used Lemma \ref{lem:uniformity}.
Proceeding in a manner similar to \eqref{eq:E18} and \eqref{eq:prob u=0} we obtain
\begin{align}
  \nonumber
  &\pr\{\u^t_{x} =\zero \mbox{ and } 
  \u_{x_{\rm e}}^{t_{\rm e}-1/2}\neq \zero\}
  \\ \nonumber =\ &
  \pr\{ C_{x-1} \cdots C_{x_{\rm e}} \u_{x_{\rm e}}^{t_{\rm e}} = \w \mbox{ and } 
  \u_{x_{\rm e}}^{t_{\rm e}-1/2}\neq \zero\}
  \\ \nonumber \leq\ &
  {\rm d}(\u_{x_{\rm e}}^{t_{\rm e}},\u) 
  +
  \pr\{ C_{x-1} \cdots C_{x_{\rm e}}\u =\w \mbox{ and } 
  \u_{x_{\rm e}}^{t_{\rm e}-1/2}\neq \zero\}
  \\ \nonumber \leq\ &
  2^{1-4N} +
  \pr\{ C_{x-1} \cdots C_{x_{\rm e}}\u =\w \}\ .
\end{align}
The bound \eqref{eq:cccu=w} exploits the fact that that $\w$ and $\u$ are independent, giving
\begin{equation} 
  \pr\{ C_{x-1} \cdots C_{x_{\rm e}}\u =\w \}
  \leq 8\, (x-x_{\rm e})\, 2^{-N}
  = 16\, (t-t_{\rm e})\, 2^{-N}\ . \nonumber
\end{equation}
Putting all things together we obtain 
\begin{align}\label{eq:s2}
  \pr\{\u^t_{x} =\zero\} 
  \leq
  2(t_{\rm e}-1/2) 2^{-2N} + 2^{1-4N} +
  16\, (t-t_{\rm e})\, 2^{-N} 
  \leq 16\, t\, 2^{-N} \ .
\end{align}
Finally, we use the union bound to conclude that 
\begin{align}
\nonumber
  \pr\{\exists\, x\in \mathbb Z _L : \u^t_x =\zero\big\}
  \leq 
  8\, t\, L\, 2^{-N}\ ,
\end{align}
which is equivalent to the statement \eqref{eq:allnonzero}.
\end{proof}

\subsection{Twirling technique and Pauli invariance}
\label{app:Twirling_PauliInvariance}

Figure~\ref{fig:TwirlingFigure} illustrates the fact that, at integer time $t$, the probability distribution of $W(t)$ is invariant under the transformation
\begin{equation}\label{eq:W_int}
  W(t) \ \mapsto\ 
  \left( \mbox{$\bigotimes_x$} V'_x \right)^\dagger W(t)
  \left( \mbox{$\bigotimes_x$} V'_x \right) 
   \ ,
\end{equation} 
for any string of local Clifford unitaries $V'_1, \ldots, V'_L \in \C_{N}$.
This property translates to distribution \eqref{Prob:final} as 
\begin{equation}\label{eq:Psym}
  P_t \left( [ \mbox{$\bigoplus_x$} S_x^{-1} ]\u' \big| [\mbox{$\bigoplus_x$} S_x ] \u \right) 
  = P_t(\u'|\u)
\end{equation}
for any list of local symplectic matrices $S_1, \ldots, S_L$.
In order to prove Theorem \ref{thm:approx2Design} we exploit the fact that, at half-integer time $t$, the evolution operator displays a higher degree of symmetry. 
The probability distribution of $W(t)$ is invariant under the transformation
\begin{equation}\label{eq:W_h_int}
  W(t) \ \mapsto\ 
  \left( \mbox{$\bigotimes_x$} V_x \right) W(t)
  \left( \mbox{$\bigotimes_x$} V'_x \right)
   \ ,
\end{equation} 
for any string of local Clifford unitaries $V_1, V'_1, \ldots, V_L, V'_L \in \C_N$.
This translates onto $P_t(\u'|\u)$ in a way analogous to \eqref{eq:Psym}.

\begin{figure}\label{fig:TwirlingFigure}
    \hspace{-3mm}
    \includegraphics[width=152mm]{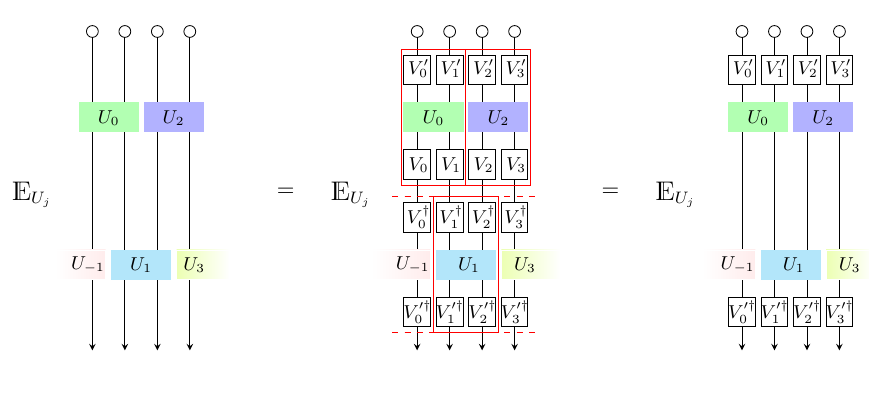}
    \vspace{-14mm}
    \caption{{\bf Twirling technique.} This figure illustrates the fact that the probability distribution of the evolution operator is invariant under local transformations. 
    On the left we have a section of the circuit of Figure~\ref{fig:CircuitFigure}. 
    On the middle we use the fact that, for any pair of local Clifford unitaries $V_x, V_{x+1}$, the random two-site Clifford unitary $(V_x \otimes V_{x+1}) U_x (V'_x \otimes V'_{x+1})$ has the same probability distribution than $U_x$.
    On the right we see that all local unitaries get cancelled except for those of the initial and final times. Note that if the time $ t $ is integer or half-integer the invariance property of $ W(t) $ is different according to equations \eqref{eq:W_int} and \eqref{eq:W_h_int}.}  
\end{figure}

In this section, we will present what is referred to as the twirling technique in the work \cite{Sunderhauf_2018} and discuss how it applies to the random Clifford circuit model we consider. 

We recollect that the definition of the evolution operator after an {\em integer} time $t$ is: 
\begin{align}\nonumber
  W(t) & \equiv \big[ (U_1 \otimes U_3 \otimes  \cdots \otimes U_{L-1}) (U_0 \otimes U_2 \otimes \cdots \otimes U_{L-2}) \big]^{t} 
  \\ \nonumber & =
  (U_{\text{odd}}U_{\text{even}})^{t} = (U_{\text{chain}})^{t}\ ,
\end{align}
and after a {\em half-integer} time $t$ is: 
\begin{align}\nonumber
  W(t) \equiv U_{\text{even}}\, (U_{\text{chain}})^{t-1/2} 
 \ .
\end{align}

\begin{Lemma}
\label{lem:TwirlingTechnique_Appendix}
Consider a set of $2L$ single-site Clifford unitaries $V_x, V'_x \in \C_{N}$, these unitaries are fixed.
At integer time $t$, the random evolution operator $W(t)$, as defined above, has the same probability distribution as
\begin{equation}
\left( \bigotimes_{x=0}^{L-1} V_x^ {\prime \dagger} \right) W(t) \left( \bigotimes_{x=0}^{L-1} V'_x \right) \ .
\end{equation}
Similarly, at half-integer time $t$ the evolution operator $W(t)$ has the same probability distribution as 
\begin{equation}
\left( \bigotimes_{x=0}^{L-1} V_x \right) W(t) \left( \bigotimes_{x=0}^{L-1} V'_x \right) \ . 
\end{equation}
\end{Lemma}

\begin{proof}
First, we note that any uniformly distributed two-site Clifford unitary $U_x \in \C_{2N}$ has the same probability distribution as  the unitary $(V_x \otimes V_{x+1} ) U_x ( V'_x \otimes V'_{x+1} )$ for any arbitrary choice of $V_x, V_{x+1}, V'_x, V'_{x+1} \in \C_N $; this is denoted as single-site Haar invariance. 
Hence, we introduce the primed notation for the random two-site Clifford unitary $U_x$
\begin{align}
U'_x & \equiv  (V_x \otimes V_{x+1} ) U_x (V'_x \otimes V'_{x+1})   \text{ for even } x \in \mathbb{Z}_L\ , \\ 
U'_x & \equiv (V'_x \otimes V'_{x+1})^{-1}  U_x (V_x \otimes V_{x+1} )^{-1}  \text{for odd } x \in \mathbb{Z}_L \ , 
\end{align}
where $V_x, V_{x+1}, V'_x, V'_{x+1} \in \C_N $ are any arbitrary choice of single-site Clifford unitary. 
Consequently, the primed version of the global dynamics for integer $t$ becomes
\begin{equation}
W'(t)  = \left(\bigotimes_{x=0}^{L-1} V_x ^{\prime \dagger} \right)W(t) \left(\bigotimes_{x=0}^{L-1} V'_x \right)  \ , 
\end{equation}
and for half-integer $t$
\begin{equation}
W'(t)  = \left(\bigotimes_{x=0}^{L-1} V_x \right) W(t ) \left(\bigotimes_{x=0}^{L-1} V'_x \right) \ . 
\end{equation}
The single-site Haar invariance of the probability distributions of the primed and not-primed evolution operators are identical, this proves the result.
\end{proof}

Next, we will define Pauli invariance and state when it applies to our model.
\begin{Definition}
\label{Def:PauliInvariance_Appendix}
An $n$-qubit random unitary $U\in$ SU($2^n$) with probability distribution $P(U)$ is Pauli invariant if $P(U\sigma) = P(U)$ for all $\sigma \in \mathcal P_n$ and $U\in$ SU($2^n$).
\end{Definition}

\begin{Lemma}
\label{lem:PauliInvariantCircuit_Appendix}
At half-integer time $t$, the random evolution operator $W(t)$ is Pauli invariant.
\end{Lemma}
\begin{proof}
The proof of this lemma follows from lemma \ref{lem:TwirlingTechnique_Appendix}.
When t is half-integer, $W(t)$ and $( \mbox{$\bigotimes_{x=0}^{L-1}$} V_x ) W(t) ( \mbox{$\bigotimes_{x=0}^{L-1}$} V'_x )$ have identical probability distributions, where $V_x,V'_x \in \C_n$.
Since $\mathcal P_n \subset \C_n$, we can choose $( \mbox{$\bigotimes_{x=0}^{L-1}$} V_x )$ to be any element of the Pauli group. 
Hence, $W(t)$ is Pauli invariant.
\end{proof}

\subsection{Half-integer times}
\label{app:PauliMixing_SemiInteger}

\begin{Lemma}
\label{lem:tSemiIntegerTwirling}
At half-integer $t\geq t_{\text{scr}}$ the probability distribution of the evolved vector $\u^{t} = S(t) \u^0$ conditioned on it being non-zero at every site is uniform:
\begin{equation}\label{E26}
\pr \{ \u^{t}=\v  | \u_x^{t} \neq \zero , \forall x \in \mathbb{Z}_L \} = \frac{1}{(2^{2N} - 1)^{L} } \ ,
\end{equation}
for all vectors $\v$ that are non-zero at every site $\v_x \neq \zero , \forall x \in \mathbb{Z}_L$. 
\end{Lemma}

\begin{proof}
The proof of this lemma follows from the twirling technique discussed in Section \ref{app:Twirling_PauliInvariance} lemma \ref{lem:TwirlingTechnique_Appendix}.
The probability distribution of the evolved vector $\u^{t} = S(t) \u^0 $ is identical to 
\begin{equation}
 \u^{t} =  \left( \mbox{$\bigoplus_{x=0}^{L-1}$} X_x \right) S(t) \left( \mbox{$\bigoplus_{x=0}^{L-1}$} Y_x \right) \u^0 , \nonumber
\end{equation}
where $X_x, Y_x \in \mathcal{S}_{N}$ are arbitrary single-site matrices. 
Hence, since the choice of each $X_x$ is arbitrary, each $X_x$ is independent and uniformly distributed over all single-site symplectic matrices. 
Therefore, imposing the condition that the evolved vector is non-zero on every site, then, since the twirling matrices $X_x$ are independent and uniform, the probability distribution of the evolved vector at each site is independent and uniformly distributed over all non-zero vectors. The application of lemma \ref{lem:uniformity} eventually provides the conditional probability \eqref{E26}.
\end{proof}

\begin{Theorem} \label{lem:tHalfDistance}
Let $\sigma_{\u'} = \lambda W(t) \sigma_\u W(t)^\dagger$ be the evolution of any initial Pauli operator $\sigma_\u \neq \unity$.
At any half-integer time $t$ larger than the scrambling time, in the interval $t \in [ t_{\rm scr} , 2  t_{\rm scr} ]$ the probability distribution~\eqref{Prob:final} for the evolved operator $\sigma_{\u'}$ is close to uniform, namely
\begin{equation}\label{res:ergo}
  \sum_{\u'} 
  \left| P_t(\u'|\u) - Q_t(\u') \right|
  \ \leq\ 33 \times t\, L\, 2^{-N}\ .
\end{equation}
\end{Theorem}

{\it Remark}. The following is an equivalent statement to Theorem \ref{lem:tHalfDistance} formulated in phase space, we then present a proof.

\vspace{6mm}
\noindent
{\bf Theorem} \ref{lem:tHalfDistance} (Alternative form). 
For any initial non-zero vector $\u^0 \in \mathcal V_{\rm chain}$, the probability distribution of the time evolved vector $ \u^t = S(t) \u^0 $, at any half-integer time in the interval $t\in [t_{\text{scr}}, 2t_{\text{scr}}]$, is approximately uniformly distributed over all non-zero vectors of the total system, and bounded by 
\begin{align}
  \nonumber
  \sum_{\v} \left| \pr\{ \u^t = \v \} -\frac{1}{2^{2NL} - 1}\right|
  \leq 
  32\, t L 2^{-N} + L 2^{-2N}
  \ . 
\end{align}

\vspace{6mm}

\begin{proof}
Below we make use of: 
\begin{align} 
 \textrm{prob}(A)&=\textrm{prob}(A \wedge B) + \textrm{prob}(A \wedge \bar{B}) \nonumber \\
 & =\textrm{prob}(A|B)  \textrm{prob}(B) + \textrm{prob}(A|\bar{B}) \textrm{prob}(\bar{B}) \nonumber
\end{align}
where $ A $ and $ B $ are events in a probability space.

Defining $q \equiv \pr \{\u^t_x \neq \zero , \forall x \in \mathbb{Z}_L\}$, 
 we rewrite $\pr\{\u^t=\v\}$ as follows  
\begin{align}
\nonumber
 \pr\{\u^t=\v\} & =  q \, \pr \{ \u^t =\v| \u^t_x \neq \zero , \forall x \in \mathbb{Z}_L  \}
 \\ \nonumber
  & + (1-q) ( \pr \{\u^t=\v | \exists y \in \mathbb{Z}_L : \u^t_y = \zero \}).  
\end{align}
Adding and subtracting $ q \frac{1}{2^{2NL} - 1} $ in the sum and then applying the triangular inequality, we find that:
\begin{align}
\nonumber
  \sum_{\v} \left| \pr\{\u^t=\v\} -\frac{1}{2^{2NL} - 1} \right|
  & \leq  
    \sum_{\v} \left|q\,  \pr \{ \u^t=\v | \u^t_x \neq \zero \, \forall x \} -\frac{1}{2^{2NL} - 1} \right| \\ \nonumber
  & \hspace{-13mm} + 
  (1-q)  \sum_{\v} \left|  \pr \{\u^t=\v | \exists y \in \mathbb{Z}_L : \u^t_y = \zero \} -\frac{1}{2^{2NL} - 1} \right| \ .
\end{align}
We can upper bound the first term with $q \leq 1$ and apply lemma \ref{lem:tSemiIntegerTwirling} to find that
\begin{equation}
\nonumber
q \sum_{\v} \left|  \pr \{ \u^t=\v | \u^t_x \neq \zero \, \forall x \} -\frac{1}{2^{2NL} - 1}\right|  \leq L 2^{-2N} \ . 
\end{equation}
To bound the second term we notice that the maximum value of the sum is 2, in fact:
\begin{align}\nonumber
& (1-q)  \sum_{\v} \left|  \pr \{\u^t=\v | \exists y \in \mathbb{Z}_L : \u^t_y = \zero \} -\frac{1}{2^{2NL} - 1} \right| \\ 
& \le (1-q)  \sum_{\v}  \left( \pr \{\u^t=\v | \exists y \in \mathbb{Z}_L : \u^t_y = \zero \} + \frac{1}{2^{2NL} - 1} \right) = 2(1-q) \nonumber 
\end{align}
and use the result of lemma \ref{lemma:allnonzero} to find that
\begin{equation}
(1-q) \leq 16 t L 2^{-N} \ . 
\end{equation}
This gives the stated result.
\end{proof}

\subsection{Integer times}
\label{app:PauliMixing_Integer}

In this section, we will consider only initial vectors which are supported (i.e.~non-zero) on a single site, $\u^0 \in \V_0 \subseteq \V_{\rm chain},$ and their time evolution at integer times only.
The validity of the following lemma isn't restricted to integer times or even quantum circuits.

\begin{Lemma}
\label{lemma:IntegerTwirl_4Paramaters}
Let $\u$ be a fixed non-zero element of $\Z2^{2N}$.
Let the probability distribution $P(\v)$ over $\v \in \Z2^{2N}$ have the property that $P(S\v) = P(\v)$ for any $S\in \mathcal S_N$ such that $S \u = \u$.
Then it must be of the form
\begin{equation}
  \label{eq:prob q}
  P(\v ) 
  = 
  \left\{ \begin{array}{ll}
    q_1 & \mbox{ if $\v= \zero$}
    \\
    q_2 & \mbox{ if $\v=\u$}
    \\
    q_3 & \mbox{ if $\langle\v,\u\rangle =0$ and $\v\neq \zero, \u$}
    \\    
    q_4 & \mbox{ if $\langle\v,\u\rangle =1$}
\end{array} \right.\ ,
\end{equation}
where the positive numbers $q_i$ are constrained by the normalization of $P(\v)$.
\end{Lemma}

\begin{proof}
We initially consider that $\u = (1, 0 , \ldots , 0)^T$, and the subgroup of $\S_N$ that leaves $\u$ unchanged.
If $\v = \zero \text{ or } \u$, then the action of this subgroup has no effect, and hence we require a parameter for each in the distribution, $q_1$ and $q_2$ respectively. 
This is not the case for all other choices of $\v$, since the action of the subgroup will transform $\v$ into some other vector in $\Z2^{2N}$. 
This transformation is constrained by the symplectic form:
\begin{equation}
    \langle \v , \u \rangle =  \langle S \v , S \u \rangle = \langle S \v , \u \rangle \ ,
\end{equation}
and hence the subgroup is composed of two subgroups, which transform $\v$ into another vector in $\Z2^{2N}$ that has the same value for the symplectic form.
Furthermore, the two subgroups are such they can map any vector to any other vector with the same value for the symplectic form. 

This can be seen by considering the case where $\v = (0,1, \ldots , 0)^T$, so $\SympForm{\u}{\v}=1$.
The subgroup that keeps $\u= (1,0, ... , 0)^T$ unchanged consists of all the elements of $\S_N$ with $\u$ as the first column of the matrix.
Hence, by lemma \ref{lemma:algorithm}, we can select the second column of the matrix to be any vector which has symplectic form of 1 with the first column, which is $\u$.
Thus, we can map $\v$ to any other vector with symplectic form one with $\u$, which is also unchanged.
Then, by noting that the product of symplectic matrices is a symplectic matrix, the subgroup can map any vector with symplectic form of one with $\u$ to any other. 
Similarly, this argument applies to the other case where the symplectic form has a value of zero.

Then since $P(S\v) = P(\v)$, all vectors that give the same value for $\langle \v, \u \rangle$ have the same probability.
Thus, we get the probability distribution in \eqref{eq:prob q}.

Finally, we note that since via a symplectic transformation $\u$ can be mapped to any other vector in $\Z2^{2N}$, and that the product of two symplectic matrices is symplectic, this result applies for any $\u \in \Z2^{2N}$.
\end{proof}

\begin{Lemma}
\label{Lemma:Integer_SingleSitePhase}
For an initial vector $ \u^0 \in \V_0$ supported on location $ x=0 $, the probability that the value of the symplectic form between the evolved vector $\u^t = S_{ \text{chain} } ^t \u^0$, with integer $ t $, and the initial vector, $\SympForm{\u^t}{\u^0} = \SympForm{\u_0^t}{\u_0^0}$ is equal to $ s $, has an $ s $-independent upper bound given by:
\begin{equation}
\pr \{ \SympForm{\u_0^t}{\u_0^0} = s \} \leq \frac{1}{2} + 8\, t\, 2^{-N} .
\end{equation}
Furthermore, this result is independent from the location of the support of $ \u^0 $ provided that it is a single site.
\end{Lemma}

\begin{proof}
To prove this lemma we proceed similarly as in Lemma~\ref{lemma:allnonzero}. 
That is, we consider a sequence of gates in the causal past of $\u_0^t$ with an elbow shape (see example in Figure~\ref{fig:elbow}).
More concretely, we write $\u_0^t$ as
\begin{align}
  \u_{1-t}^{t/2} &= 
  B_{1-t}\cdots B_{-2}B_{-1}A_0 \u_0^0\ ,
  \\
  \u_{1-t}^{t/2+1/2} &= D_{-t} \u_{1-t}^{t/2}\ ,
\\
  \u_0^t &= 
  C_{-1} \cdots C_{2-t} C_{1-t} \u_{1-t}^{t/2+1/2} + \w \ ,
\end{align}
where, crucially, the random vector $\w$ is independent of the random matrix $D_{-t}$.
This vector $\w$ is defined in a way similar to \eqref{eq:random wt}. 

Next, we follow a sequence of steps similar to those from \eqref{eq:s1} to \eqref{eq:s2}.
First we write
\begin{align}
  \nonumber
  & \pr\!\left\{ \langle\u_0^0, \u_0^t\rangle =s 
  \right\}
  \\ \nonumber =\ &
  \pr\!\left\{ \langle\u_0^0, \u_0^t\rangle =s     
  \mbox{ and } \u_{1-t}^{t/2}= \zero\right\}
  +
  \pr\!\left\{ \langle\u_0^0, \u_0^t\rangle =s 
  \mbox{ and } \u_{1-t}^{t/2}\neq \zero\right\}
  \\ \label{eq:E50} \leq\ &
  \pr\!\left\{ \u_{1-t}^{t/2}= \zero\right\}
  +
  \pr\!\left\{ \langle\u_0^0, \u_0^t\rangle =s 
  \mbox{ and } \u_{1-t}^{t/2}\neq \zero\right\}\ .
\end{align}
Second, we bound the first term by using the recursive relation \eqref{eq:rec rel} as
\begin{align}
  \pr\!\left\{ \u_{1-t}^{t/2}= \zero\right\}
  \leq 
  2\, t\, 2^{-2N}\ .
\end{align}
Third, we introduce the uniformly distributed random vectors $\u,\u' \in \Z2^{2N}$, which are independent of the gates $S_{1-t}, S_{2-t}, \ldots, S_{-2}$, and write
\begin{align}
  \nonumber
  &\pr\!\left\{ \langle\u_0^0, \u_0^t\rangle =s 
  \mbox{ and } \u_{1-t}^{t/2}\neq \zero\right\}
  \\ \nonumber =\ &  
  \pr\!\left\{ \u_0^{0\,T} J C_{-1}\cdots C_{1-t} \u_{1-t}^{t/2+1/2} =s +\langle\u_0^0, \w\rangle
  \mbox{ and } \u_{1-t}^{t/2}\neq \zero\right\}
  \\ \nonumber \leq\ &  
  {\rm d}\!\left(\u_{1-t}^{t/2+1/2}, \u\right)
  +\pr\!\left\{ \u_0^{0\,T} J C_{-1}\cdots C_{1-t} \u =s +\langle\u_0^0, \w\rangle
  \mbox{ and } \u_{1-t}^{t/2}\neq \zero\right\}
  \\ \nonumber \leq\ &  
  {\rm d}\!\left(\u_{1-t}^{t/2+1/2}, \u\right)
  +\pr\!\left\{ \u_0^{0\,T} J C_{-1}\cdots C_{1-t} \u =s +\langle\u_0^0, \w\rangle\right\}
  \\ \label{eq:thirdterm} \leq\ &  
  {\rm d}\!\left(\u_{1-t}^{t/2+1/2}, \u\right)
  + {\rm d}\!\left(JC_{-1}^T J\u_0^0, \u'\right)
  +\pr\!\left\{ \u'^T J C_{-2}\cdots C_{1-t} \u =s +\langle\u_0^0, \w\rangle\right\}\ .
\end{align}
Fourth, using \eqref{eq:sdist} we can write the bounds
\begin{align}
  \nonumber
  {\rm d}\!\left(\u_{1-t}^{t/2+1/2}, \u\right)
  \leq\ & 2^{1-4N}\ ,
  \\
  {\rm d}\!\left(JC_{-1}^T J\u_0^0, \u'\right)
  \leq\ & 2^{1-4N}\ .
\end{align}
Fifth, in order to bound the third term in \eqref{eq:thirdterm} we note that, for any non-zero ${\bf a}\in \Z2^{2N}$ we have $  \pr\{\u'^T {\bf a} = s\} =1/2$, for both $s=0,1$, therefore
\begin{align}
  \nonumber
  &\pr\!\left\{ \u'^T J C_{-2}\cdots C_{1-t} \u =s +\langle\u_0^0, \w\rangle\right\}
  \\ \nonumber=\ &
  \pr\!\left\{ \u'^T J C_{-2}\cdots C_{1-t} \u =s +\langle\u_0^0, \w\rangle \mbox{ and } C_{-2}\cdots C_{1-t} \u\neq \zero \right\}
  \\ \nonumber+\ &
  \pr\!\left\{ \u'^T J C_{-2}\cdots C_{1-t} \u =s +\langle\u_0^0, \w\rangle \mbox{ and } C_{-2}\cdots C_{1-t} \u= \zero \right\}\ .
\end{align}
Next we bound the first term by using the fact that the uniformly distributed vector $\u'$ is independent of ${\bf a} := J C_{-2} \cdots C_{1-t} \u$ and $\langle\u_0^0, \w\rangle$, as
\begin{align}
  \nonumber
  \pr\!\left\{ \u'^T {\bf a} =s +\langle\u_0^0, \w\rangle \mbox{ and } {\bf a} \neq \zero \right\}
  \leq 
  \pr\!\left\{ \u'^T {\bf a} =s +\langle\u_0^0, \w\rangle \big|\, {\bf a} \neq \zero \right\}
  =\frac 1 2 \ .
\end{align}
The second term can be easily bounded as
\begin{align}
  \nonumber
  &\pr\!\left\{ \u'^T J C_{-2}\cdots C_{1-t} \u =s +\langle\u_0^0, \w\rangle \mbox{ and } C_{-2}\cdots C_{1-t} \u= \zero \right\}
  \\ \nonumber \leq\ & 
  \pr\!\left\{ C_{-2}\cdots C_{1-t} \u= \zero \right\}
  \ \leq\ 8 (t-2) 2^{-N}\ ,
\end{align}
where the last inequality follows from Lemma~\ref{le:CCCCu=0}.
Combining the above two bounds we obtain
\begin{align}
  \nonumber
  \pr\!\left\{ \u'^T J C_{-2}\cdots C_{1-t} \u =s +\langle\u_0^0, \w\rangle\right\}  
  \ \leq\
  \frac 1 2 + 8 (t-2) 2^{-N}\ .
\end{align}
Sixth, putting everything together back from \eqref{eq:E50} we arrive at
\begin{align}
  \nonumber
  \pr\!\left\{ \langle\u_0^0, \u_0^t\rangle =s 
  \right\}
  \leq\ &
  2\, t\, 2^{-2N} + 4\, 2^{-4N} + \frac 1 2 + 8 (t-2) 2^{-N}
  \\ \nonumber \leq\ & \frac 1 2 + 8\, t\, 2^{-N},
\end{align}
as we wanted to show.
\end{proof}

\begin{Lemma}
\label{app:Integer_ApproxUniformNonzero}
For an initial vector $\u^0 \in \V_0$ supported at $ x=0 $, the probability distribution of the evolved vector $\u^t = S_{ \text{chain} } ^t \u^0$ at integer times, conditioned on the evolved vector being non-zero at every site and different from the initial single-site non-zero vector, $\u_x^t \neq 0  \ \forall x \in \mathbb{Z}_L \text{ and } \u_0^t \neq \u_0^0 $, after the scrambling time $t_{\text{scr}}$ is of the form
\begin{equation}
\nonumber
\pr \{ \u^t | \u_x^t \neq 0 \  \forall x \in \mathbb{Z}_L ,\u_0^t \neq \u_0^0 \} \leq \frac{1}{(2^{2N} - 1)^{L-1}}  \left\{
 \begin{array}{cc}
   \frac{8 t 2^{-N} + 1/2}{2^{2N-1} - 2} & \text{ if } \SympForm{\u_0^t}{\u_0^0} = 0  \\
   \frac{8 t 2^{-N} + 1/2}{2^{2N-1} } &  \text{ if } \SympForm{\u_0^t}{\u_0^0} = 1
  \end{array}  \right. \ ,
\end{equation}
and before the scrambling time for all sites within the causal light-cone the probability distribution is of the form
\begin{align}
\nonumber
&\pr \{ \u^t | \u_x^t \neq 0 \  \forall x \in [- 2t + 1 , 2t ]  ,\u_0^t \neq \u_0^0 \} 
  \\ &\hspace{20mm} \leq 
  \frac{1}{(2^{2N} - 1)^{4t-1}}  \left\{
 \begin{array}{cc}
   \frac{8 t 2^{-N} + 1/2}{2^{2N-1} - 2} & \text{ if } \SympForm{\u_0^t}{\u_0^0} = 0  \\
   \frac{8 t 2^{-N} + 1/2}{2^{2N-1} } &  \text{ if } \SympForm{\u_0^t}{\u_0^0} = 1
  \end{array}  \right. \ .
\end{align}
Furthermore, this result holds for any choice of the single-site at which the initial vector is non-zero.
\end{Lemma}

\begin{proof}
The proof of this lemma uses the twirling technique discussed in Section \ref{app:Twirling_PauliInvariance} lemma \ref{lem:TwirlingTechnique_Appendix}.
The probability distribution of the evolved vector $\u^t = (S_{\text{chain}})^t \u^0$ at integer times is identical to 
\begin{align}
 \u^t &=  \left( \mbox{$\bigoplus_{x=0}^{L-1}$} X_x \right) S_{\text{chain}}^t \left( \mbox{$\bigoplus_{x=0}^{L-1}$} X_x^{-1} \right)  \u^0 \nonumber \\
 &=\left( \mbox{$\bigoplus_{x=0}^{L-1}$} X_x \right) S_{\text{chain}}^t \left( X_0^{-1}\u_0^0 \, \mbox{$\bigoplus_{x=1}^{L-1}$} \zero \right)  \label{E42}
\end{align}
where $X_x \in \mathcal{S}_{N}$ are arbitrary single-site symplectic matrices. Equation \eqref{E42} follows from the fact that $ \u^0 $ has been assumed supported at $ x=0 $, therefore $ \left( \mbox{$\bigoplus_{x=0}^{L-1}$} X_x^{-1} \right)  \u^0 $ is supported at $ x=0 $ as well.
If we restrict $X_0$ to the elements of $\mathcal{S}_N$ that satisfy $X_0 \u_0^0 = \u_0^0$, then the probability distribution of $\u^t$ is identical to $\left( \bigoplus_{x=0}^{L-1}  X_x \right) \u^t$.
Since the choice of symplectic matrices $\bigoplus_{x=0}^{L-1} X_x$ to twirl is arbitrary, we can take each single-site matrix to be independent and uniformly distributed over all single-site symplectic matrices, except for $X_0$ which is uniformly distributed over the restricted set satisfying $X_0 \u_0^0 = \u_0^0$.
Then, we condition on the evolved vector being non-zero at all sites $x$ and different for the initial single-site non-zero vector, $\u_x^t \neq 0  \ \forall x \in \mathbb{Z}_L \text{ and } \u_0^t \neq \u_0^0 $. 
Therefore under this condition, the evolved vector at each site is independent and uniformly distributed over all non-zero vectors (lemma \ref{lem:uniformity}) apart from the initial vector $ \u_0^0 $.
On the vector space $\V_0$ we invoke lemma \ref{lemma:IntegerTwirl_4Paramaters}, and hence the evolved vector $\u_0^t$ ($\neq \zero, \u_0^0$) at $ x=0 $ is uniformly distributed over all the vectors with the same symplectic form 
with $\u_0^0$, $\SympForm{\u_0^t}{\u_0^0}$. 
Hence, using lemma \ref{Lemma:Integer_SingleSitePhase}, which gives an upper bound for the probability of $\SympForm{\u_0^t}{\u_0^0} \in \{0 , 1\}$, we get the stated result.
\end{proof}

The following theorem establishes approximate Pauli mixing: 
the probability that $\u$ evolves onto $\u'$ after a time $t$ given by:
\begin{equation}
  P_t (\u'|\u) = \mathop{\mathbb E}_{\{U_x\}}
  \left| 2^{-NL} \tr(\sigma_{\u'} W(t) \sigma_{\u} W(t)^\dagger) \right|\ ,
\end{equation}
is close to the uniform distribution over all non-zero vectors $\u'$ in the causal subspace \eqref{eq:cc} denoted by $Q_{t}(\u')$. After the scrambling time $t\geq t_{\rm scr}$, $Q_{t}(\u')$ is the uniform distribution over all non-zero vectors in the total phase space $\mathcal V_{\rm chain}$.

\begin{Theorem} \label{lemma:lightcone} (Approximate Pauli mixing) If the initial Pauli operator $\sigma_\u$ is supported at site $x=0$ then the probability distribution~\eqref{Prob:final} for its evolution $\sigma_{\u'}$ is close to uniform inside the light cone
\begin{equation}\label{res:ergod}
  \sum_{\u'} 
  \left| P_t(\u'|\u) - Q_t(\u') \right|
  \ \leq\ 130\times t^2\, 2^{-N}\ ,
\end{equation}
for any integer or half-integer time $t \in [1/2 , 2 t_{\rm scr} ]$.
An analogous statement holds for any other initial location $x\neq 0$.
\end{Theorem}

{\it Remark.} The following is an alternative enunciation of Theorem \ref{lemma:lightcone} formulated in phase space rather than Hilbert space. A proof of this alternative form then follows.

\vspace{6mm}
\noindent
{\bf Theorem} \ref{lemma:lightcone} (Alternative form).
For an initial vector supported at $ x=0 $, the evolved vector $\u^t = S_{ \text{chain} } ^t \u^0$, at integer times, is approximately uniformly distributed over all non-zero vectors within the light-cone. 
For any $ t \in [1,t_{\rm scr}] $ and $x \in [-2t + 1 , 2t]$ we have:
\begin{align}
  \sum_{\v \in \mathbb{Z}_2^{8Nt}} \left| \pr\{\u^t =\v \} -\frac{1}{2^{8Nt} - 1 }\right|
  \leq  
  \ 32 t (4t + 1 ) 2^{-N} + 4 t 2^{-2N}
\end{align}
For any $ t \in [t_{\rm scr}, 2t_{\rm scr}] $ it holds:
\begin{align}
 \nonumber
  \sum_{\v \in \mathbb{Z}_2^{2NL}} \left| \pr\{\u^t =\v \} -\frac{1}{2^{2NL} - 1 }\right|
  \leq 32 t (L+1) 2^{-N} + L 2^{-2N}
\end{align}

\vspace{6mm}

\begin{proof}
Let us consider the case $ t \ge t_{\rm scr} $ first.
Similarly to the proof of Theorem \ref{lem:tHalfDistance}  we employ 
\begin{align} 
 \textrm{prob}(A)&=\textrm{prob}(A \wedge B) + \textrm{prob}(A \wedge \bar{B}) \nonumber \\
 & =\textrm{prob}(A|B)  \textrm{prob}(B) + \textrm{prob}(A|\bar{B}) \textrm{prob}(\bar{B}) \nonumber
\end{align}
where $ A $ and $ B $ are events in a probability space. With $q \equiv \pr \{ \u^t_x \neq \zero \, \forall x \in \mathbb{Z}_L  \wedge \u_0^t \neq \u_0^0 \}$, 
$\pr\{ \u^t=\v \}$ is then rewritten in the following way
\begin{align}
\nonumber
 \pr\{\u^t=\v\} & =  q \, \pr \{ \u^t=\v | \u^t_x \neq \zero \, \forall x \in \mathbb{Z}_L \wedge \u_0^t \neq \u_0^0 \}
 \\ \nonumber
  & + (1-q) ( \pr \{ \u^t=\v | \, \exists x \in \mathbb{Z}_L \, \textrm{such that} \, \u^t_x  = \zero \vee \u_0^t = \u_0^0 \} ) \ , 
\end{align}

Summing and subtracting $ q  \frac{1}{2^{2NL}-1} $ into the sum over $ \v $ and using the triangular inequality we find that
\begin{align}
\nonumber
  & \sum_{\v} \left| \pr\{ \u^t=\v \} -\frac{1}{2^{2NL} - 1}\right| \\ \nonumber
  & \leq  
   q \sum_{\v} \left| \pr \{ \u^t=\v | \u^t_x \neq \zero \, \forall x \in \mathbb{Z}_L  , \u_0^t \neq \u_0^0 \} -\frac{1}{2^{2NL} - 1}\right| \\ \nonumber
  & + 
  (1-q)  \sum_{\v} \left|   \pr \{ \u^t=\v | \, \exists x \in \mathbb{Z}_L \, \textrm{such that} \, \u^t_x  = \zero \vee \u_0^t = \u_0^0 \}  -\frac{1}{2^{2NL} - 1}\right| \ .
\end{align}
We can bound the first term using $ q \leq 1$ and apply lemma  \ref{app:Integer_ApproxUniformNonzero} to find that
\begin{equation}
\nonumber
q \sum_{\v} \left| \pr \{\u^t=\v | \u^t_x \neq \zero \, \forall x \in \mathbb{Z}_L  , \u_0^t \neq \u_0^0 \} -\frac{1}{2^{2NL} - 1} \right|  \leq 16 t 2^{-N} + (L + 1)  2^{-2N}  \ . 
\end{equation}
To evaluate the second term above, we upper bound the sum with its maximum value of 2 and use the result of lemma \ref{lemma:allnonzero} to find that
\begin{equation}
\nonumber
  (1-q)\sum_{\v} \left|    \pr \{ \u^t=\v |  \, \exists x \in \mathbb{Z}_L \, \textrm{such that} \, \u^t_x  = \zero \vee \u_0^t = \u_0^0 \} -\frac{1}{2^{2NL} - 1} \right|  \leq 32 t (L + 1) 2^{-N} \ . 
\end{equation}
Combining, this gives the stated result for integer times after the scrambling time.

To derive the results for integer times before the scrambling time, we note that the derivation is identical with the substitution $L \rightarrow 4 t $, which agree when $t = t_{\text{scr}}$ (and after this time).
\end{proof}

\subsection{Approximate mixing with arbitrary initial state}
\label{app:Ergodicity_AritraryInitial}

Consider a subsystem of the chain comprising $\Ls$ consecutive sites, where $\Ls$ is even. Without loss of generality we choose this subsystem to be $\{1,2, \ldots, \Ls\} \subseteq \mathbb Z_L$.
We analyse the state of this subsystem at times
\begin{equation}\label{eq:prescram}
  t \leq \frac {L-\Ls} 4 \ .
\end{equation}
This condition ensures that the left backwards wave front of $\u_1^t$ and the right backwards wave front of $\u_{\Ls}^t$ do not collide.
Without this condition, the analysis becomes very complicated.

\begin{Lemma}\label{lem:IntegerTime_SubsystemPhases}
Consider an initial vector $\u^0 \in \V_{\rm chain}$  supported on all lattice sites ($\u_x^0 \neq\zero$ for all $x\in \mathbb Z_L$), and its evolution at time $ t $, $\u^t$.
Define the random variable $s_x = \langle \mathbf \u^t_x, \u^0_x \rangle$ at each site of the region $x\in \{1, \ldots, \Ls\} \subseteq \mathbb Z_L$, where $\Ls$ is even.
Then we have
\begin{align} \label{E59}
  P(s_1, \ldots, s_{\Ls}) 
  & \leq 
  2^{-L_s} + 32\,t\, 3^{\frac{L_s}{2}+1}\, 2^{-N}\ , 
\end{align}
as long as $t \leq (L-\Ls)/4$.
\end{Lemma}

\begin{proof}
The value of the random vectors $\u^t_1, \ldots, \u^t_{\Ls}$ is only determined by the random matrices $S_{2-2t}, \ldots, S_{\Ls+2t-2}$. The rest of matrices $S_x$ are not contained in the causal past of the region under consideration $\{1,2, \ldots, \Ls\}$.
In order to simplify this proof, we will replace $S_{2-2t}, \ldots, S_{\Ls+2t-2}$ by a new set of random variables defined in what follows.

Let us label by $y \in \{1, \ldots, \Ls/2\}$ the pair of neighbouring sites $\{2y-1, 2y\} \subseteq \{1, \ldots, \Ls\}$.
For each pair $y$ we consider a given non-zero vector $\mathbf a_y \in \Z2^{4N}$ and define the random variables
\begin{align}
  \mathbf b_y 
  &= 
  S_{2y-1}^{-1} \mathbf a_y\ ,
  \\ \label{def:h_y}
  h_y 
  &= 
  \left\langle \mathbf a_y, \u^t_{2y-1} \oplus \u^t_{2y} 
  \right\rangle
  = 
  \left\langle \mathbf b_y, \u^{t-1/2}_{2y-1} \oplus \u^{t-1/2}_{2y} \right\rangle\ .
\end{align}
The left-most random contribution to $h_y$ is the matrix $S_{2y-2t}$, or equivalently the vector $\w_y$, defined through
\begin{equation}
  \tilde\w_y \oplus \w_y 
  = 
  S_{2y-2t} (\u^0_{2y-2t} \oplus \u^0_{2y-2t+1})\ .
\end{equation}
We note that $ \w_y \in \V_{2y-2t+1} $.
This contribution and others are illustrated in Figure~\ref{fig:past_region}. 
\begin{figure}
    \hspace{-6mm}
    \includegraphics[width = 160mm]{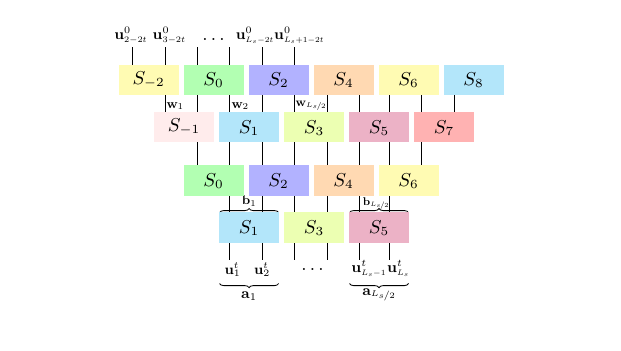}
    \vspace{-16mm}
    \caption{This figure represents the region $\{1,2, \ldots, 6\}$ at time $t=2$, and its causal past back to $t=0$. (Hence $\Ls =6$.) All the random matrices $S_{-2}, \ldots, S_8$ contribute to the value of the vectors $\u^t_1, \ldots, \u^t_6$. The left-most contribution to the vector $\u^t_1$ is the matrix $S_{-2}$, or equivalently the vector $\w_1$. The given vector $\mathbf a_y$ associated to the pair of neighbouring sites $y$, and its 1/2-step backwards time translations $\mathbf b_y$, are also represented.}
    \label{fig:past_region}
\end{figure}
The contribution of the vector $\w_y$ to $h_y$ (and $\u^{t-1/2}_{2y-1}$) is ``transmitted through" the matrices $S_{2y-2}$, $S_{2y-3}$, $\ldots$, $S_{2y-2t+2}$, $S_{2y-2t+1}$.
More precisely, $\w_y$ is mapped via the matrix product
\begin{equation}
  F_y
  = 
  C_{2y-2} C_{2y-3} \cdots C_{2y-2t+2} C_{2y-2t+1}\ ,
\end{equation}
where we have used decomposition \eqref{eq:S_x}.
We denote by $\v_y$ all contributions to $\u^{t-1/2}_{2y-1}$ that are not $F_y \w_y$,
\begin{equation}
  \v_y 
  = 
  (\u_{2y-1}^{t-1/2} +F_y \w_y) \oplus \u_{2y}^{t-1/2}.
\end{equation}
We remark that $ \v_y \in \V_{2y-1} \oplus \V_{2y} $.
The last random variable that we need to define is $g_y = \langle \mathbf b_y, \v_y \rangle$, which together with \eqref{def:h_y} allows us to write
\begin{equation}
  h_y 
  =
  \left\langle \mathbf b_y, F_y \w_y +\v_y \right\rangle
  =
  \left\langle \mathbf b_y, F_y \w_y \right\rangle + g_y\ .
\end{equation}
Note the slight abuse of notation in that we write $F_y \w_y$ instead of $F_y \w_y \oplus \mathbf 0$.

In summary, we have replaced the variables $S_{2-2t}, \ldots, S_{\Ls+2t-2}$ by the variables $\w_y, \mathbf b_y, F_y, g_y$ for $y=1, \ldots, \Ls/2$. (We are not using $\v_y, \tilde \w_y$ any more.)
These variables are not all independent, but they satisfy the following independence relations:
\begin{itemize}
  \item $\w_1, \mathbf b_1, \ldots, \w_{\Ls/2}, \mathbf b_{\Ls/2}$ are independent and uniform.

  
  \item $\w_y$ is independent of $g_{y'}$ for all $y'\geq y$.
  
  \item $F_y$ is independent of $\w_{y'}$ and $\mathbf b_{y''}$ for all $y' \leq y$ and $y'' \geq y$.
\end{itemize}
To continue with the proof it is convenient to introduce the following notation:
\begin{align}
  \u_{\geq y} &= (\u_y, \u_{y+1}, \ldots, \u_{\Ls/2})\ ,
  \\
  \u_{\leq y} &= (\u_1, \u_2, \ldots, \u_y)\ ,
\end{align}
and analogously for $>,<$ and the rest of variables $\mathbf b_y, F_y, g_y$.
This allows us to write the joint probability distribution of $h_1, \ldots, h_{\Ls/2}$ as
\begin{equation}\label{E62}
  P(h_{\geq 1})
  = \sum_{ \w_{\geq 1}, \mathbf b_{\geq 1}, F_{\geq 1}, g_{\geq 1} }
  P(\w_{\geq 1}, \mathbf b_{\geq 1}, F_{\geq 1}, g_{\geq 1})
  \prod_y \delta \left(h_y, \left\langle \mathbf b_y, F_y \w_y \right\rangle + g_y \right)\ .
\end{equation}
Equation \eqref{E62} follows directly from the definition of the Kronecker-delta.
Note that we can write the above distribution $P(\w_{\geq 1}, \mathbf b_{\geq 1}, F_{\geq 1}, g_{\geq 1})$ as 
\begin{align}
  \nonumber
  &     
  P(\mathbf w_{\geq 1}, \mathbf b_{\geq 1}, F_{\geq 1}, g_{\geq 1}) 
  =
  \sum_{S_0, S_1, \ldots, S_{L-1}} 
  P(S_0) P(S_1) \cdots P(S_{L-1})\times
  \\ \nonumber & \hspace{1cm} \times \prod_{y=1}^{L_s/2}
  \delta\!\left(\mathbf w_{y}, S_{2y-2t}[ \mathbf u^0_{2y-2t} \oplus \mathbf u^0_{2y-2t+1}] \right) \times
  \delta\!\left(\mathbf b_{y}, S^{-1}_{2y-1} \mathbf a_y \right) \times
  \\ & \hspace{1cm} \times 
\delta\!\left(F_{y}, C_{2y-2}\cdots C_{2y-2t+1} \right) \times
\delta \left( g_y, \left\langle \mathbf b_y , (\u_{2y-1}^{t-1/2} +F_y \w_y) \oplus \u_{2y}^{t-1/2}\right\rangle \right).
\end{align}
The following sum-rule is repeatedly exploited below, where $ {\bf 0}_{2N} $ denotes the $ 2N \times 2N $ matrix with all entries equal to $ 0 $, instead $ \bf 0 $ is the vector with $ 2N $ components equal to $ 0 $.
\begin{equation}
  \label{eq:sum_w}
  \sum_{\w_1} P(\w_1) 
  \delta \left( h_1, \left\langle \mathbf b_1, F_1 \w_1 \right\rangle + g_1 \right)
  = \left\{
  \begin{array}{ll}
    \delta \left( h_1,g_1 \right) & \mbox{ if } (F_1 \oplus {\bf 0}_{2N})^T J \mathbf b_1 = \mathbf 0\ ,
    \\
    1/2 & \mbox{ otherwise.}
  \end{array}\right.
\end{equation}
Equation \eqref{eq:sum_w} is obtained as follows. We first note that 
\begin{align} \nonumber
  \left\langle \mathbf b_1, F_1 \w_1 \oplus {\bf 0} \right\rangle = \left\langle \mathbf b_1, (F_1 \oplus {\bf 0}_{2N}) (\w_1 \oplus {\bf 0}) \right\rangle = \mathbf b_1^T J (F_1 \oplus {\bf 0}_{2N}) (\w_1 \oplus {\bf 0} ) \ .
\end{align}
If $ (F_1 \oplus {\bf 0}_{2N})^T J \mathbf b_1 = \mathbf 0 $, the first of equations \eqref{eq:sum_w} follows from the normalisation of the probability $ P(\w_1) $.
If $ (F_1 \oplus {\bf 0}_{2N})^T J \mathbf b_1 \neq \mathbf 0 $, since the values of $ \w_1 $ are distributed uniformly over all the vectors of $ \Z2^{2N} $, implying that $ P(\w_1) = \frac{1}{2^{2N}} $, then $ \left\langle \mathbf b_1, F_1 \w_1 \right\rangle $ takes half of the times the value $ 0 $ and half of the times the value $ 1 $. 
Recall that $ F_1 $ and $ \mathbf b_1 $ are fixed in \eqref{eq:sum_w}. The second equation of \eqref{eq:sum_w} then follows.

Using $\delta(h,h') \leq 1$ for all $h,h'$ and \eqref{eq:sum_w} we can write
\begin{align*}
  & P(h_{\geq 1}) = 
  \sum_{ \w_{\geq 1}, \mathbf b_{\geq 1}, F_{\geq 1}, g_{\geq 1} } 
  P(\w_1) P(\w_{\geq 2}, \mathbf b_{\geq 1}, F_{\geq 1}, g_{\geq 1})
  \prod_y \delta \left( h_y, \left\langle \mathbf b_y, F_y \w_y \right\rangle + g_y \right)
  \\
  &\leq 
  \sum_{ \w_{\geq 2}, \mathbf b_{\geq 1}, F_{\geq 1}, g_{\geq 1} } 
  \delta \left( (F_1 \oplus {\bf 0}_{2N})^T J \mathbf b_1 , \mathbf 0 \right) P(\w_{\geq 2}, \mathbf b_{\geq 1}, F_{\geq 1}, g_{\geq 1})
  \prod_{y \ge 2} \delta \left( h_y, \left\langle \mathbf b_y, F_y \w_y \right\rangle + g_y \right) + \nonumber \\
  & + \frac 1 2 \sum_{ \w_{\geq 2}, \mathbf b_{\geq 2}, F_{\geq 2}, g_{\geq 2} } 
  P(\w_{\geq 2}, \mathbf b_{\geq 2}, F_{\geq 2}, g_{\geq 2})
  \prod_{y\geq 2} \delta \left( h_y , \left\langle \mathbf b_y, F_y \w_y \right\rangle + g_y \right) \\
  &\leq 
  \pr \{(F_1 \oplus {\bf 0}_{2N})^T J \mathbf b_1 = \mathbf 0\} + \nonumber \\
  & + \frac 1 2 \sum_{ \w_{\geq 2}, \mathbf b_{\geq 2}, F_{\geq 2}, g_{\geq 2} } 
  P(\w_{\geq 2}, \mathbf b_{\geq 2}, F_{\geq 2}, g_{\geq 2})
  \prod_{y\geq 2} \delta \left( h_y , \left\langle \mathbf b_y, F_y \w_y \right\rangle + g_y \right)
\end{align*}
To bound the term associated with the case $ (F_1 \oplus {\bf 0}_{2N})^T J \mathbf b_1 \neq \mathbf 0 $  we extended the sum over $\mathbf b_1, F_1$ from the values satisfying $(F_1 \oplus {\bf 0}_{2N})^T J \mathbf b_1 \neq \mathbf 0$ to all values.
Since the variables $\mathbf b_1, F_1, g_1$ do not appear in any of the remaining $\delta$-functions, we can trace them out.
Subsequently we repeat the above process by summing over $\w_2$, using the analog of \eqref{eq:sum_w} for $y=2$, and summing over $\w_2, F_2, g_2$, obtaining
\begin{align}
  \nonumber
  P(h_{\geq 1})
  = 
  \epsilon + \frac 1 2 \left( \epsilon + \frac 1 2 
  \sum P(\w_{\geq 3}, \mathbf b_{\geq 3}, F_{\geq 3}, g_{\geq 3})
  \prod_{y\geq 3} \delta \left( h_y , \left\langle \mathbf b_y, F_y \w_y \right\rangle + g_y \right) \right)\ ,
\end{align}
where we define $\epsilon = \pr \{(F_1 \oplus {\bf 0}_{2N})^T J \mathbf b_1 = \mathbf 0\}$.
Continuing in this fashion yields
\begin{align}
\label{prob_bounds}
  \nonumber
  P(h_1, \ldots, h_{\Ls/2})
  &= 
  \epsilon \sum_{k=0}^{\Ls/2 -1} 2^{-k} + 2^{-\Ls/2} \ ,
  \\ &\leq  
  2 \epsilon + 2^{-\Ls/2}\ .
\end{align}

We now wish to turn this bound from a distribution of $h_y$ to the distribution of $s_x \equiv \SympForm{\u_x^t }{\u_x^0}$ (recalling that $x \in \{ 1, 2 , \ldots \Ls \}$ and $ y \in \{1, ..., L_s/2 \} $), that is to say we want to bound $P(s_1 , s_2, \ldots , s_{\Ls} )$. 

Let us consider first the simplest case, that is $ L_s = 2 $, $ h_1 = s_1 + s_2 $. The couple $ (s_1,s_2) $ has four possible realizations $ (0,0), (1,0), (0,1), (1,1) $.  
We have the following bounds:
\begin{align}
& P(h_1=0) = P(0,0) + P(1,1) \le \frac{1}{2} + 2\epsilon \label{e1}\\
& P(h_1=1) = P(0,1) + P(1,0) \le \frac{1}{2} + 2\epsilon \nonumber \\
& P(0,0) + P(1,1) \ge \frac{1}{2} - 2\epsilon \nonumber
\end{align}
The last bound combine normalization $ P(0,0) + P(1,1) + P(0,1) + P(1,0) = 1 $ and the second bound above. Similarly it holds $  P(0,1) + P(1,0) \ge \frac{1}{2} - 2\epsilon $.

The bound \eqref{prob_bounds} extends to the case where rather than the values of $ h_y \equiv s_{2y-1} + s_{2y} $ are fixed, the values of certain  $ h_y $ and of certain $ s_x $ are fixed.  
In the case $ L_s = 2 $ this amounts to three cases: $ \text{$  h_1 $ fixed} $, $ \text{$ s_1 $ fixed} $, $ \text{$ s_2 $ fixed} $.
It then follows
\begin{align} 
& P(s_1=0)=P(0,0)+P(0,1) \le \frac{1}{2} + 2\epsilon \label{e2} \\
& P(s_1=1)=P(1,0)+P(1,1) \le \frac{1}{2} + 2\epsilon \nonumber  \\
& P(s_2=0)=P(0,0)+P(1,0) \le \frac{1}{2} + 2\epsilon \nonumber  \\
& P(s_2=1)=P(0,1)+P(1,1) \le \frac{1}{2} + 2\epsilon \nonumber 
\end{align}
The lower bound can be obtained similarly to what we have done above, for example: 
\begin{align}
P(0,1)+P(1,1) \ge \frac{1}{2} - 2\epsilon \label{e3}
\end{align}
Summing \eqref{e1} with \eqref{e2} and subtracting \eqref{e3}, we obtain $ P(0,0) \le \frac{1}{4} + 3\epsilon $.
With a similar approach we obtain $ P(0,0) \ge \frac{1}{4} - 3\epsilon $. The procedure that we have described holds for all $ P(s_1,s_2) $, then:
\begin{equation} \nonumber
 \frac{1}{4} - 3\epsilon \le P(s_1,s_2) \le \frac{1}{4} + 3\epsilon .
\end{equation}
We introduce a matrix formalism to re-obtain the result above, this formalism will allow us to treat the case $ L_s > 2 $. 
The set of inequalities \eqref{e1} and \eqref{e2}, with the respective lower bounds, can be written as:

\begin{equation} \label{system_inequalities}
\begin{pmatrix}
 \frac{1}{2} - 2\epsilon \\
 \frac{1}{2} - 2\epsilon \\
 \frac{1}{2} - 2\epsilon \\
 \frac{1}{2} - 2\epsilon \\
 \frac{1}{2} - 2\epsilon \\
 \frac{1}{2} - 2\epsilon \\
\end{pmatrix}
\le
\begin{pmatrix}
 1 & 0 & 0 & 1 \\
 0 & 1 & 1 & 0 \\
 1 & 1 & 0 & 0 \\
 0 & 0 & 1 & 1 \\
 1 & 0 & 1 & 0 \\
 0 & 1 & 0 & 1 \\
\end{pmatrix}
\begin{pmatrix}
 P(0,0) \\
 P(0,1) \\
 P(1,0) \\
 P(1,1) \\
\end{pmatrix}
\le
\begin{pmatrix}
 \frac{1}{2} + 2\epsilon \\
 \frac{1}{2} + 2\epsilon \\
 \frac{1}{2} + 2\epsilon \\
 \frac{1}{2} + 2\epsilon \\
 \frac{1}{2} + 2\epsilon \\
 \frac{1}{2} + 2\epsilon \\
\end{pmatrix}
\end{equation}
We denote the $ 6 \times 4 $ matrix above with $ A $.
As far as regards the system of inequalities above we have already shown explicitly the solution of it, in particular we saw that to find the upper bound $ P(s_1,s_2) \le \frac{1}{4} + 3 \epsilon $ we need both upper and lower bounds in \eqref{e1} and \eqref{e2}. The same holds for the lower bound. It is easy to describe a way to obtain a solution of \eqref{system_inequalities} where we get exactly the term of $ O(1) $ and we overestimate the correction $ O(\epsilon) $. Since in \eqref{system_inequalities} the term of $ O(1) $ is the same both in the upper bound and the lower bound then the term of $ O(1) $ in \eqref{system_inequalities} is obtained replacing the inequalities with equality. The solution of the corresponding system is also easily obtained by inspection, in fact since every row of $ A $ has two entries equal to $ 1 $ a solution of the system is $ P(s_1,s_2)=\frac{1}{4} $ for all $  (s_1,s_2) $, since $ A $ is a full rank matrix this is also the only solution. 
To evaluate the error we consider the following equality where $ {\bf{a}}_j $ is the $ j$-th row of the matrix $ A $, and $ \bf{P} $ is the column vector of probabilities as in equation \eqref{system_inequalities}
\begin{equation} \label{P(0,0)}
 \frac{1}{2} ({\bf{a}}_3 + {\bf{a}}_5 - {\bf{a}}_2) \cdot {\bf{P}}=P(0,0)  
\end{equation}
The equation above can be generalized to every $ P(s_1,s_2) $. This means that in general we only need to sum or subtract three among the inequalities in \eqref{system_inequalities} to obtain any  $ P(s_1,s_2) $ therefore the maximal error in modulus that can arise is equal to $ 6 \epsilon $, then we can rewrite:
\begin{equation} \label{error}
 \frac{1}{4} - 6\epsilon \le P(s_1,s_2) \le \frac{1}{4} + 6\epsilon .
\end{equation}

It is easy to understand that each row of the matrix $ A $ carries a ``label'' as specified below, in fact, for example, the product of the first row of $ A $ with the vector that has entries given by $ P(s_1,s_2) $, outlined in equation \eqref{system_inequalities}, gives $ P(0,0)+P(1,1) \equiv P(h_1=0) $, therefore the first row carries the label $ h_1=0 $.

\begin{equation} \label{m_A}
 A \equiv
\begin{pmatrix}
 1 & 0 & 0 & 1 \\
 0 & 1 & 1 & 0 \\
 1 & 1 & 0 & 0 \\
 0 & 0 & 1 & 1 \\
 1 & 0 & 1 & 0 \\
 0 & 1 & 0 & 1 \\
\end{pmatrix}
\leftrightarrow
\begin{pmatrix}
 h_1=0 \\
 h_1=1 \\
 s_1=0 \\
 s_1=1 \\
 s_2=0 \\
 s_2=1 \\
\end{pmatrix} \ .
\end{equation}
The generalization to the case  $ L_s = 4 $ is given considering the Kronecker product (i.e. the tensor product in the standard basis) of the matrix $ A $ with itself. For example the first row of the matrix $ A \otimes A $ carries the label $ h_1=0 , h_2=0 $, the second row carries the label $ h_1=0 , h_2=1 $, the third $ h_1=0 , s_2=0 $ and so on.

In the case $ L_s=4 $ we want to bound $ P(s_1,s_2,s_3,s_4) $, to get them the idea is the same as that exploited in the case $ L_s =2 $ namely taking linear combinations of bounds on $ P(h_1,h_2) $, $ P(h_1,s_3) $, $ P(h_1,s_4) $ and so on.
We notice that equation \eqref{P(0,0)} generalizes to $ L_s = 4 $ as follows:
\begin{equation} \label{one_bound}
 \frac{1}{4}({\bf{a}}_3 + {\bf{a}}_5 - {\bf{a}}_2) \otimes ({\bf{a}}_3 + {\bf{a}}_5 - {\bf{a}}_2) \cdot {\bf{P}} = P(0,0,0,0).
\end{equation}
$\bf{a}_j$ denotes row $j$-th of matrix $A$, and the Kronecker product of two rows is a row with $ L_s^2 $ elements, $\bf{P}$ denotes the vector of all possible choices of $ P(s_1,s_2,s_3,s_4) $.
The equation \eqref{one_bound} involves nine bounds because there are nine terms in the tensor product $  ({\bf{a}}_3 + {\bf{a}}_5 - {\bf{a}}_2) \otimes ({\bf{a}}_3 + {\bf{a}}_5 - {\bf{a}}_2) $ , and so
\begin{equation}
 \frac{1}{16}-54 \varepsilon \leq P(0,0,0,0) \leq \frac{1}{16}+54 \varepsilon \ . \nonumber
\end{equation}
Note that the error $ 54 \varepsilon $ arises as the product of the error associated with each bound in \eqref{error} and the number of inequalities that is $ 9 $.
To generalize this to arbitrary $\Ls$, we just consider further tensor products of $A$, and hence
\begin{equation} 
 2^{-L_s} - 2 \; 3^{\frac{L_s}{2}+1} \varepsilon \leq P(s_1,...,s_{L_s}) \leq 2^{-L_s} + 2 \; 3^{\frac{L_s}{2}+1} \varepsilon \ . 
\end{equation}
Using lemma \ref{le:CCCCu=0}, with $r=2t$ and $n=N$ we have $ \epsilon < 16t2^{-N} $, this implies \eqref{E59}.
\end{proof}

\begin{Theorem}
\label{lem:IntegerTime_SubsystemDistance}
Consider an initial vector $\u^0 \in \V_{\rm chain}$ with non-zero support in all lattice sites ($\u_x^0 \neq\zero$ for all $x\in \mathbb Z_L$).
Consider the evolved vector $\u^t = S(t) \u^0$ inside a region $x\in \{1, \ldots, \Ls\} \subseteq \mathbb Z_L$ where $\Ls$ is even and the time is $t \leq \frac{L - \Ls}{4}$.
If $\u^t_{[1,Ls]}$ is the projection of $\u^t$ in the subspace $\bigoplus_{x=1}^{\Ls} \V_x$ then
\begin{equation}
\sum_{\v \in \Z2^{2N\Ls}} \left| \pr\{\v=\u^t_{[1,Ls]}\} -\frac 1 {2^{2 N \Ls}}\right|
  \leq 
32\, t\, 2^{-N} ( 2 \Ls  +3^{\frac{\Ls}{2} + 1} ) +  4 L 2^{-2N} \ .
\end{equation}
\end{Theorem}

\begin{proof}
First, we re-state $ \pr\{\v=\u^t\}$ in the following way
\begin{align}
\nonumber
 \pr\{\v=\u^t\} & =  q \, \pr \{\v = \u^t | \u^t_x \neq \zero , \u^0_x \, \forall x \in \mathbb Z _ {\Ls}\}
 \\ \nonumber
  & + (1-q) (1 -  \pr \{\v = \u^t | \u^t_x \neq \zero, \u^0_x \, \forall x \in \mathbb Z _ {\Ls} \} )  \ , 
\end{align}
where $x\in \{1, \ldots, \Ls\} \subseteq \mathbb Z_L$ with $\Ls$ is even, $q$ is the probability of distribution $\pr \{
 \u^t_x \neq \zero , \u^t_x \, \forall x \in \mathbb Z _ {\Ls}\}$, and similarly with the complement. 
Then using convexity we find that
\begin{align}
\nonumber
\sum_{\v \in \Z2^{2N\Ls}} \left| \pr\{\v=\u^t\} -\frac 1 {2^{2 N \Ls}}\right|
  \leq 
 q  \sum_{\v \in \Z2^{2N\Ls}} \left| \pr \{\v = \u^t | \u^t_x \neq \zero , \u^0_x \, \forall x  \} - \frac 1 {2^{2N \Ls}}\right|
&  \\ \nonumber 
  +
 (1-q) \sum_{\v \in \Z2^{2N\Ls}} \left| 1 -  \pr \{\v = \u^t | \u^t_x \neq \zero, \u^0_x \, \forall x   \} -\frac 1 {2^{2N \Ls }}\right|
 \ .
\end{align}
We can evaluate the first term using the upper bound $q \leq 1$ and use Lemma \ref{lemma:IntegerTwirl_4Paramaters} combined with Lemma \ref{lem:IntegerTime_SubsystemPhases} to find that 
\begin{equation}
  q \sum_{\v \in \Z2^{2N\Ls}} \left| \pr \{\v = \u^t | \u^t_x \neq \zero , \u^0_x \, \forall x \in \mathbb Z _ {\Ls}  \} - \frac 1 {2^{2N \Ls}}\right|
\leq
32\, t 3^{\frac{\Ls}{2} +1 } 2^{-N} + L 2^{2-2N} \ .
\end{equation}
To evaluate the second term, we can upper bound the sum by its maximum value, 2, and use the result of Lemma \ref{lemma:allnonzero} to upper bound $(1-q)$ to find that
\begin{equation}
 (1-q) \sum_{\v \in \Z2^{2N\Ls}} \left| 1 -  \pr \{\v = \u^t | \u^t_x \neq \zero, \u^0_x \, \forall x  \in \mathbb Z _ {\Ls}  \} -\frac 1 {2^{2N \Ls }}\right|
 \leq 
 64 \Ls t 2^{-N} \ .
\end{equation}
Combining these two terms we get the stated result. 
\end{proof}




\section{Approximate 2-design at half-integer time}
\label{app:ApproxDesignSemiTime}
In this section we will combine the results of the sections \ref{app:Twirling_PauliInvariance} and \ref{app:LocalDynamics_PauliMixing}, with the results of the reference \cite{Webb_2016}, to show that the random circuit model we consider is an approximate 2-design in a weak sense (Theorem \ref{thm:approx2Design}).

As discussed in the main body, in the reference \cite{Webb_2016} (specifically Appendix A) it is demonstrated that if a Clifford circuit satisfies both Pauli invariance (Section \ref{app:Twirling_PauliInvariance} Definition \ref{Def:PauliInvariance_Appendix}) and Pauli mixing (Section \ref{app:LocalDynamics_PauliMixing} Theorem \ref{lem:tHalfDistance}) then it is an exact 2-design. 
In the following theorem, we will demonstrate that when Pauli mixing is only approximate, as in our case, then the random Clifford circuit is instead an approximate 2-design when one has access to Pauli measurements alone.


\begin{Theorem} \label{thm:approx2Design}
  Consider the task of discriminating between two copies of $W(t)$ and two copies of a Haar-random unitary $U$ with measurements restricted to Pauli operators, when $t \in [ t_{\rm scr} , 2  t_{\rm scr} ]$ is half-integer.
  The success probability for correctly guessing the given pair of unitaries satisfies
\begin{align} 
  \nonumber
  p_{\rm guess}
  =\ &\frac 1 2 + \frac 1 4 
  \ \max_{\rho, \u, \v}\ 
  \tr\!\left(\!\sigma_\u \otimes \sigma_\v\! \left[
  \mathop{\mathbb E}_{W(t)}
  W(t)^{\otimes 2} \rho\, W(t)^{\otimes 2 \dagger}
  -
  \int_{{\rm SU}(d)} \hspace{-7mm} dU\,
  U^{\otimes 2} \rho\, U^{\otimes 2 \dagger}
  \right]\right)
  \\ \label{eq:P meas}
  \leq \ & 
  1/2 + 9\, t L 2^{-N}
  \ .
\end{align}
\end{Theorem}

\begin{proof}
Let us consider a general state describing two copies of the system
\begin{equation}
  \rho = \sum_{\u, \v} \alpha_{\u, \v}\, \sigma_\u \otimes \sigma_\v\ ,
\end{equation}  
where $\alpha_{0,0} = 2^{-2NL}$ by normalisation.
The coefficients $\alpha_{\u, \v}$ must satisfy the following
\begin{equation}\label{eq:alpha bound}
  \alpha_{\u,\v}\, 2^{2NL} 
  = 
  \tr(\rho\, \sigma_\u \otimes \sigma_\v)
  \in [-1,1]\ .
\end{equation}
Applying the average dynamics to $\rho$ we obtain
\begin{equation}\label{eq:WW}
  \mathop{\mathbb E}_{W(t)}
  W(t)^{\otimes 2} \rho\, W(t)^{\otimes 2 \dagger}
  =
  2^{-2NL} \unity \otimes \unity +
  \sum_{\u,\v\neq 0} \alpha_{\v,\v} 
  \pr \{\v=S(t)\u\} \sigma_\u \otimes \sigma_\u
  \ .
\end{equation}
The fact that terms $\alpha_{\u,\u'}$ and $\sigma_\u \otimes \sigma_{\u'}$ with $\u\neq \u'$ are not present in the above expression follows from the fact that $W(t)$ is Pauli-invariant (see appendix A of the reference \cite{Webb_2016}), which is proven in Lemma \ref{lem:PauliInvariantCircuit_Appendix}.
Recall that at half-integer $t$ we have the time-reversal symmetry
\begin{equation}\label{eq:time refl}
  \pr \{\v=S(t)\u\} = \pr \{\u=S(t)\v\}\ .
\end{equation}
Applying the Haar twirling on $\rho$ we obtain
\begin{equation}\label{eq:UU}
  \int_{{\rm SU}(d)} \hspace{-7mm} dU\,
  U^{\otimes 2} \rho\, U^{\otimes 2 \dagger}
  =
  2^{-2NL} \unity \otimes \unity +
  \sum_{\u,\v\neq 0} \alpha_{\v,\v}\, \gamma\,
  \sigma_\u \otimes \sigma_\u\ ,
\end{equation}
where $\gamma = (2^{2NL}-1)^{-1}$ is the uniform distribution over non-zero vectors in $\V_{\rm chain}$.
Substituting \eqref{eq:WW} and \eqref{eq:UU} into \eqref{eq:P meas} we obtain
\begin{align} 
  &\tr\left(\sigma_\u \otimes \sigma_\v \left[
  \mathop{\mathbb E}_{W(t)}
    W(t)^{\otimes 2} \rho\, W(t)^{\otimes 2 \dagger}
  -
  \int_{{\rm SU}(d)} \hspace{-7mm} dU\,
  U^{\otimes 2} \rho\, U^{\otimes 2 \dagger}
  \right]\right)
  \\ = &\ 
  \delta_{\u,\v} 
  \sum_{\w\neq 0} \alpha_{\w,\w} 
  \left(\pr \{\u=S(t)\w\} - \gamma\right) 2^{2NL}
  \\ \leq &\ 
  \delta_{\u,\v} 
  \sum_{\w\neq 0} 
  \left|\pr \{\u=S(t)\w\} - \gamma\right|
  \leq 33\, t L 2^{-N} \delta_{\u,\v} \ ,
\end{align}
where in the last two inequalities we use \eqref{eq:alpha bound}, \eqref{eq:time refl} and Theorem \ref{lem:tHalfDistance}. This implies that the guessing probability satisfies $p_{\rm guess} \leq 1/2 + 9\, t L 2^{-N}$, hence \eqref{eq:P meas}.
\end{proof}

The following result is not presented in the main text because it is difficult to interpret. It is important to not confuse the infinite norm between two states with the infinite norm between two maps. What we have here is the first. The second is the definition of quantum tensor-product expander.

\begin{Lemma}
The dynamics $W(t)$ defined in equation \eqref{def_W}, with $ t \ge t_{\rm scr} $ half-integer, is closed to an approximate 2-design with respect to the infinity norm, namely for any state $\rho$ it holds:
\begin{equation} \label{eq:phi meas}
  \left\|
  \mathop{\mathbb E}_{W(t)}
  W(t)^{\otimes 2} \rho\, W(t)^{\otimes 2 \dagger}
  -
  \int_{{\rm SU}(2^{NL})} \hspace{-7mm} dU\,
  U^{\otimes 2} \rho\, U^{\otimes 2 \dagger}
  \right\|_\infty
  \leq 
  33\, t L 2^{-N}.
\end{equation}
\end{Lemma}

\begin{proof}
Let $|\phi_0\rangle \equiv \left( \frac{|0,1\rangle - |1,0\rangle}{\sqrt{2}} \right)^{\otimes NL} $ denote the $NL$-fold tensor-product of the singlet state, where each singlet entangles each qubit of the first copy of the system and the corresponding qubit in the second copy of the system.
This implies that $(\sigma_\u \otimes \sigma_\u) |\phi_0\rangle = (-1)^{|\u|} |\phi_0\rangle$, where $|\u| \equiv \left( \sum_j u_j \right)\textrm{mod} 2$. $ \sigma_\u $ is defined as in \eqref{eq:Pauli element}:
\begin{align}
  \sigma_\u  =   \bigotimes_{i=1}^n   (\sigma_x^{q_i} \sigma_z^{p_i}) \in {\rm U}(2^n) \nonumber
\end{align}
with $\u = (q_1, p_1, q_2, p_2, \ldots, q_n, p_n ) \in \Z2^{2n}$.
Any Bell state (as described above) can be written as $|\phi_\v\rangle = (\unity \otimes \sigma_\v) |\phi_0\rangle$ for all $\v\in \V_{\rm chain}$.
Note that these form an orthonormal basis for the Hilbert space of two copies of the system $\langle \phi_\u|\phi_\v\rangle = \delta_{\u,\v}$.
Also, using the commutation relations \eqref{eq:commutatation relation} we obtain
\begin{align}
  \nonumber
  (\sigma_\u \otimes \sigma_\u) |\phi_\v\rangle
  &=
  (\sigma_\u \otimes \sigma_\u) 
  (\unity \otimes \sigma_\v) |\phi_0\rangle
  \\ \nonumber &=
  (-1)^{\langle \v, \u \rangle} 
  (\unity \otimes \sigma_\v) 
  (-1)^{|\u|} |\phi_0\rangle
  \\ \label{eq:eigen phi} &=
  (-1)^{\langle \v, \u \rangle +|\u|} 
  |\phi_\v\rangle\ .
\end{align}
This together with \eqref{eq:WW} and \eqref{eq:UU} implies that the argument inside the norm \eqref{eq:phi meas} is diagonal in the $|\phi_\v\rangle$ basis.
Therefore, the following bound for each element of the basis provides the bound for the $\infty$-norm:
\begin{align}
  &\langle\phi_\v|\left(
  \mathop{\mathbb E}_{W(t)}
  W(t)^{\otimes 2} \rho\, W(t)^{\otimes 2 \dagger}
  -
  \int_{{\rm SU}(d)} \hspace{-7mm} dU\,
  U^{\otimes 2} \rho\, U^{\otimes 2 \dagger}
  \right)|\phi_\v\rangle
  \\ =& 
  \sum_{\u,\w\neq 0} \alpha_{\w,\w} 
  \left(\pr \{\u=S(t)\w\} - \gamma\right)
  \langle\phi_\v| \sigma_\u \otimes \sigma_\u |\phi_\v\rangle
  \\ =& 
  \sum_{\u,\w\neq 0} \alpha_{\w,\w} 
  \left(\pr \{\u=S(t)\w\} - \gamma\right)
  (-1)^{\langle \v, \u \rangle +|\u|} 
  \\ \leq & 
  \sum_{\u,\w\neq 0} 2^{-2NL} 
  \left|\pr \{\u=S(t)\w\} - \gamma\right|
  \leq 33\, t L 2^{-N}\ .
\end{align}
\end{proof}


\section{Localisation with $ N \ll \log L $}
\label{app:CliffordLocalisation}

In this section, we consider the same spin chain with random local Clifford dynamics and again we will work in the phase space description, which was discussed in the Appendix \ref{app:CliffordPhaseSpace}. 
We will show that in the regime of $N \ll \log L$ the random dynamics, instead of displaying scrambling, results in the localisation of all operators in bounded region.

The most simple case that results in localisation is when one of the $L$ two-site gates $S_x$ has $C_x=0$, so there is no right-wards propagation, and hence by the time-periodic nature of the circuit prevents right-wards propagation for all subsequent times also. 
A bound on the probability of this happening is given in the following theorem.

\begin{Theorem}
\label{eq:ProbE=0}
Any given $S\in \S_{2n}$ can be written in block form 
\begin{equation}
  S =
  \left( \begin{array}{cc}
     A & B \\
     C & D 
  \end{array}\right) ,
\end{equation}
according to the decomposition $\Z2^{4n} = \Z2^{2n} \oplus \Z2^{2n}$, and if $S$ is uniformly distributed then this induces a distribution on the sub-matrices $A, B, C, D$.
For each of the sub-matrices ($E= A, B, C, D$) the induced distribution satisfies
\begin{equation}
  \frac{ 2^{-4N^2} }{2}   \leq \pr \{ E = 0 \} = \frac{ | \S_{n} |^2}{|\S_{2n}|}  \leq \ 2^{-4N^2} \ , 
\end{equation}
It also holds: $\pr \{ A = 0 | D = 0 \}  =  \pr \{ D = 0 | A = 0 \} = \pr \{ B = 0 | C = 0 \} =\pr \{ C = 0 | B = 0 \} = 1$.
\end{Theorem}

\begin{proof}
We first consider when $C = 0$. 
By Lemma \ref{eq:C=0iffB=0} in the appendix \ref{app:AdditionalLemmas}, this implies that $B = 0$. 
Therefore, $A$ and $D$ are both $2n \times 2n$ symplectic matrices, which can be counted independently. 
Following the counting algorithm in Lemma \ref{lemma:algorithm}, the number of choices of $S$ with $C=0$ is given exactly by
\begin{equation}
    |\{ S \in \S_{2n} : C = 0 \}| = |\S_n| |\S_n| = |\S_n|^2 \ .
\end{equation}
Finally, dividing by the total number of choices for S gives the probability. 
Using Lemma \ref{S block permutation} and \ref{eq:C=0iffB=0}, this argument applies to any of the four sub-matrices $A,B,C,D$.
The bounds are found using lemma \ref{eq:order S}.
\end{proof}

We refer to this as trivial localisation as it is equivalent a non-interacting matrix, and hence results in the spin chain being split into two independent parts.
In the rest of this section, we investigate other conditions for localisation which are not trivial and occur as a result of the dynamics.

 The following Lemma \ref{less_integers} shows that the number of powers $ k $ that need to satisfy equation \eqref{loc_level_t} is finite.

\begin{Lemma} \label{less_integers}
The conditions
\begin{equation} \label{l1}
  C_{x+1} \left( D_x A_{x+1} \right)^k  C_x = 0 
  \ \ \mbox{ for all }\ \  
  k \in \{0,1,2, \ldots, 2^{4N}-1 \}\ ,    
\end{equation}
imply
\begin{equation} \label{l1}
  C_{x+1} \left( D_x A_{x+1} \right)^k  C_x = 0 
  \ \ \mbox{ for all }\ \  
  k \in \{0,1,2, \ldots \}\ .    
\end{equation}
\end{Lemma}

\begin{proof}
Suppose that the square matrix $M$ has $n$ linearly independent powers 
\begin{equation}\label{eq:powers}
  M, M^2, M^3, \ldots, M^n,  
\end{equation}
and that $M^{n+1}$ is a linear combination of (\ref{eq:powers}). 
Let us prove that for any integer $m>n$ the matrix $M^m$ is also a linear combination of (\ref{eq:powers}).
First note that our premise $M^{n+1} = \sum_{k=0}^n a_k\, M^k$ implies that
\begin{equation}
  M^{n+2} = \sum_{k=0}^n a_k\, M^{k+1} 
  = \sum_{k=0}^{n-1} a_k\, M^{k+1} + a_n \sum_{k=0}^n a_k\, M^{k} 
\end{equation}
is also a linear combination of (\ref{eq:powers}).
Now we can proceed by induction.
For any $m>n$, suppose that the matrix $M^{m}$ is a linear combination of (\ref{eq:powers}), that is $M^{m} = \sum_{k=0}^n b_k\, M^{k}$. 
Then, proceeding as before, we have
\begin{equation}
  M^{m+1} = \sum_{k=0}^n b_k\, M^{k+1}
  = \sum_{k=0}^{n-1} b_k\, M^{k+1} + b_n \sum_{k=0}^n a_k\, M^{k}\ ,
\end{equation}
which proves our claim.

Finally, we apply this result to $M= D_x A_{x+1}$, and note that, since $M$ is a square matrix of dimension $2^{2N}$, it can have at most $2^{4N}$ linearly independent powers.
\end{proof}

\begin{Theorem}\label{lem:N1loc}
For $N=1$ the conditions 
\begin{equation}
 C_{x+1} \left( D_x A_{x+1} \right)^k  C_x = 0 ,  
\end{equation}
for $k \in \{0,1,2,\ldots \}$ are implied by the two conditions 
\begin{equation} \label{loc_N=1}
C_{x+1}  C_x = 0\quad \text{  and }\quad C_{x+1}  D_x A_{x+1}  C_x = 0 \ .
\end{equation}
Furthermore the probability of this is given exactly by
\begin{equation}
\pr \{ C_{x+1} C_x = 0 ,  C_{x+1} D_x A_{x+1} C_x = 0 \}  = 0.12 \  \  , 
\end{equation}
which includes trivial localisation.
\end{Theorem}

\begin{proof}
We are concerned with the case $N=1$, then $S_x$ and $S_{x+1}$ are $4 \times 4$ symplectic matrices and the sub blocks $A,B,C,D$ are $2\times 2$ matrices.
We first note that if $C_{x} = 0$ and/or $C_{x+1} = 0$, which is trivial localisation, then it is clear that the conditions for all $k$ are satisfied.
Hence, we now focus only on the cases where $C_{x} \neq 0$ and $C_{x+1} \neq 0$.
Moreover, we note that we will only focus on the cases where $\text{Rank}(C_x) = \text{Rank}(C_{x+1}) = 1$, since if either of $C_{x}$ or $C_{x+1}$ are full rank then to satisfy $C_{x+1} C_{x} = 0$ the other of the $C$ matrices must be the zero matrix.

When $ \text{Rank}(C_{x+1}) = 1$ then $C_{x+1}^{T} J C_{x+1} = 0$, this follows from the fact that the matrix $ C_{x+1} $ has only one distinct column that is non-zero.
By the symplectic conditions, equations \eqref{eq:system}, this implies that $A_{x+1}$ is a $2 \times 2$ symplectic matrix.
This argument also applies to $C_{x}$, and so $D_{x}$ is also a $ 2 \times 2$ symplectic matrix. 

Therefore, since the product of symplectic matrices is also a symplectic matrix, for $N=1$ neglecting the cases of trivial localisation ($C_{x} = 0$ and/or $C_{x+1} = 0$) the conditions for right localisation, \eqref{loc_N=1}, become
\begin{equation}
 C_{x+1} S^k  C_x = 0 ,  
\end{equation}
where $S$ is a generic $2 \times 2$ symplectic matrix.
For all $2\times2$ symplectic matrices, there exist $\alpha, \ \beta \in \mathbb{Z}_2$ such that:
\begin{equation}
S^2 = \alpha \mathbb{I} + \beta S \ ,
\end{equation}
which can be verified by a direct check. 
So, if $C_{x+1}  C_x = 0 \text{ and } C_{x+1}  S  C_x = 0$ hold then $ C_{x+1} S^k  C_x = 0$ for all $k > 1$. 

The exact result for the probability given above for the case of $N=1$ follows from directly counting, with the aid of a computer program, the number of symplectic matrices that satisfy \eqref{loc_N=1}.
\end{proof}

In the following Lemma \ref{counterexample} we provide an explicit example showing that the conditions \eqref{loc_N=1} sufficient to ensure localisation in the case $ N=1 $ are not enough to imply \eqref{loc_level_t}, therefore \eqref{loc_N=1} does not imply localisation for $ N > 1 $.

\begin{Lemma} \label{counterexample}
In the case $ N > 1 $ the set of equations \eqref{loc_level_t} are sufficient to ensure the presence of a hard wall. For qubits, $ N = 1 $, equations \eqref{loc_N=1} imply equations \eqref{loc_level_t}. We show that for $ N > 1 $, \eqref{loc_N=1} does not imply   \eqref{loc_level_t} by explicitly constructing an example for $ N = 2 $ that also   generalizes to all $ N > 1 $. In what follows to ease the notation we set $ x =0 $.
In the following $ J_{4N} $ the symplectic form of order $ 4N $.
The definition of symplectic matrix, $ S^TJ_{4N}S=J_{4N}$, when $ S $ is written in block form 
\begin{align}
S=
 \begin{pmatrix}
  A & B  \\
  C & D 
 \end{pmatrix}
\end{align}
reads:
\begin{equation} \label{eq:system}
\begin{cases} 
   A^T J_{2N} A  + C^T J_{2N} C  = J_{2N} \\
   A^T J_{2N} B + C^T J_{2N} D = 0 \\
   B^T J_{2N} B  + D^T J_{2N} D  = J_{2N} 
\end{cases}
\end{equation}
With $ J_{2N} $ the symplectic form of order $ 2N $. A solution of the system \eqref{eq:system} is given by:
\begin{equation} \label{C_eq}
\begin{cases} 
   C^T J_{2N} C  = C J_{2N} C^T = 0  \\
   A^T J_{2N} A = J_{2N} \\
   D^T J_{2N} D = J_{2N} \\
   B= AJ_{2N}C^{T}J_{2N}D 
\end{cases}
\end{equation}
This implies that $ A $ and $ D $ are symplectic, $ B $ is determined by $ A , C, D $.

Our goal is to build $ C_0 $, $ D_0 $, $ A_1 $ and $ C_1 $ such that: $ C_1C_0=0 $, $ C_1D_0A_1C_0=0 $ but $ C_1 \left( D_0A_1 \right)^2 C_0 \neq 0 $ showing that with $ N > 1 $ the proof given above for qubits fails and the whole set of equations \eqref{loc_level_t} must be satisfied.

Let us write straight away the matrices $ S_0 $ and $ S_1 $ and then discuss their structure.

\begin{align}
 S_0=
 \renewcommand{\arraystretch}{0.7}
 \begin{pmatrix}
  1 & 0 & 0 & 0 & 0 & 0 & 0 & 0 \\
  0 & 1 & 0 & 0 & 0 & 0 & 1 & 0 \\
  0 & 0 & 1 & 0 & 0 & 0 & 0 & 0 \\
  0 & 0 & 0 & 1 & 0 & 0 & 0 & 0 \\
  1 & 0 & 0 & 0 & 0 & 0 & 0 & 1 \\
  0 & 0 & 0 & 0 & 0 & 0 & 1 & 0 \\
  0 & 0 & 0 & 0 & 1 & 0 & 0 & 0 \\
  0 & 0 & 0 & 0 & 0 & 1 & 0 & 0 \\
\end{pmatrix},
 S_1=
 \renewcommand{\arraystretch}{0.7}
 \begin{pmatrix}
  1 & 0 & 0 & 0 & 1 & 0 & 0 & 0 \\
  0 & 1 & 0 & 0 & 0 & 0 & 0 & 0 \\
  0 & 0 & 1 & 0 & 0 & 0 & 0 & 0 \\
  0 & 0 & 0 & 1 & 0 & 0 & 0 & 0 \\
  0 & 0 & 0 & 0 & 1 & 0 & 0 & 0 \\
  0 & 1 & 0 & 0 & 0 & 1 & 0 & 0 \\
  0 & 0 & 0 & 0 & 0 & 0 & 1 & 0 \\
  0 & 0 & 0 & 0 & 0 & 0 & 0 & 1 \\
 \end{pmatrix}
\end{align}
The blocks $ C_0 $ and $ C_1 $ are the  projection on $ e_1 \equiv (1,0,0,0)^T $ and $ e_2 \equiv (0,1,0,0)^T $. They satisfy $ C^T J_4 C  = C J_4 C^T = 0 $ and also $ C_1 C_0 = 0 $. To ensure $ C_1 \left( D_0A_1 \right)^2 C_0 \neq 0 $, $ \left( D_0A_1 \right)^2 $ must map $ e_1 $ into $ e_2 $, on the other hand to ensure  $ C_1D_0A_1C_0=0 $, $ D_0 A_1 $ must not map $ e_1 $ into $ e_2 $. This is achieved, for example by:
\begin{align}
D_0 A_1 =
 \begin{pmatrix}
  0_2 & J_2  \\
  \mathds{1}_2 & 0_2 \\
 \end{pmatrix},
(D_0 A_1)^2 =
 \begin{pmatrix}
  J_2 & 0_2  \\
  0_2 & J_2 \\
 \end{pmatrix} 
 =J_4
\end{align}
The matrix $ D_0 A_1 $ has been written in block form to show that this construction generalizes to higher dimensions, in fact in every dimension $ J_{2N} $ maps $ e_1 $ to $ e_2 $. At the same time $ C_0 $ and $ C_1 $ in higher dimensions are still the  projection on $ e_1  $ and $ e_2 $. As far as regards higher powers of  $ D_0 A_1 $, it is easy to see that $ (D_0 A_1)^4 = \mathds{1}$, therefore $ (D_0 A_1)^6 = (D_0 A_1)^2 $, in general $ \forall k \in \mathbb{N} $ $ (D_0 A_1)^{4k+2} = (D_0 A_1)^2 $. 
\end{Lemma}

\subsection{Absence of localization with $ N \gg \log L $}

The following theorem provides an upper bound for the probability that one-sided walls appear at a particular location. 
This upper bound implies that when $N \gg \log L$ a typical circuit has no localisation.

\begin{Theorem}
\label{lem:RightBlockSuffCond}
The conditions
\begin{equation} \label{loc_level_t}
 C_{x+1} \left( D_x A_{x+1} \right)^k  C_x = 0 
 \ \ \mbox{ for all }\ \  
 k \in \{0,1,2, \ldots \}\ ,    
\end{equation}
are sufficient to prevent all right-wards propagation past position $x$ at any time. The probability that this family of constrains holds is upper-bounded by
\begin{align}
& \pr \{ C_{x+1} \left( D_x A_{x+1} \right)^k C_x = 0 ,\ \forall k \in \mathbb{N}  \}  \nonumber \\
\leq\ &\pr \{ C_{x+1} C_x = 0 \} \ \leq\  \frac{2N + 1}{(1-2^{-2N})^{2N}}\, 2^{2N - 2N^2}.  \label{G5}
\end{align}
\end{Theorem}

\begin{proof}
This proof is clearer with reference to figure \ref{fig:loc_1} and \ref{fig:loc_2}.
The condition $C_{x+1} C_x = 0 $ prevents right-wards propagation for a single time-step, however (unless $C_x = 0 $) then $A_{x+1} C_x  \neq 0 $ and hence in subsequent time steps there could be right-wards propagation.
In the next time step, the only way for possible right-ward propagation to occur, that would not be blocked by the condition $C_{x+1} C_x = 0 $, is $C_{x+1} D_ x A_{x+1} C_x$, and so the additional requirement $C_{x+1} D_ x A_{x+1} C_x = 0$ prevents right-ward propagation.
Once again the same argument applies for subsequent time-steps, and hence we require that $ C_{x+1} \left( D_x A_{x+1} \right)^k  C_x = 0 $ for $ k \ge 2 $ ($k \in \mathbb{N}$). 
The bound given in \eqref{G5} is obtained from equation \eqref{eq:CF rank} with $ k=2N $ and $ r=2 $.
\end{proof}

\noindent\textit{Remark.} Theorem \ref{lem:RightBlockSuffCond} provides a sufficient condition. There are of course other potential conditions and mechanisms by which right-wards propagation is prevented.

\section{Discussion}
\label{sec:discussion}

\subsection{The scrambling time}
\label{sec:scrambling time}

In this section we argue that the time $t$ at which the evolution operator $W(t)$ maximally resembles a Haar unitary (Theorem \ref{thm:approx2Design}) is around the scrambling time $t_{\rm scr}$.
For this we note that there are two factors contributing to this resemblance: causality and recurrences.

{\bf Causality.} 
If $U$ is a Haar-random unitary then a local operator $A$ is mapped to a completely non-local operator $UAU^\dagger$ with high probability.
But in our model, the evolution $W(t) A W(t)^\dagger$ of a local operator $A$ is supported in its light cone, which only reaches the whole system at the scrambling time $t_{\rm scr}$.
Hence, for $W(t) A W(t)^\dagger$ to be a completely non-local operator we need $t\geq t_{\rm scr}$.
  
{\bf Recurrences.}
The powers $U^t$ of a Haar-random unitary $U\in {\rm SU}(d)$ lose their resemblance to a Haar unitary as $t$ increases.
  This can be quantified with the spectral form factor, which for a Haar unitary $U$ takes the small value $|\tr U|^2 \approx 1$, while for its powers it takes the larger value $|\tr U^t|^2 \approx t$. Specifically, we have 
\begin{align}\label{eq:form factor}
  &K_{\rm Haar}(t) =
  \int_{{\rm SU}(d)} \hspace{-7mm} dU\,
  \left| {\rm tr} U^t \right|^2 
  = \left\{ \begin{array}{ll}
    t \ \ \mbox{ if } 0< t < d\\
    d \ \ \mbox{ if } t\geq d
  \end{array} \right.\ .
\end{align}
That is, as time $t$ grows, the form factor of $U^t$ tends to that of Poisson spectrum (integrable system)
\begin{align}\label{eq:form factor 2}
  K_{\rm Poisson}(t) = d  \ \ \ \ \ \ \mbox{ for all } t > 0\ .
\end{align}
In our model the evolution operator $W(t)$ is never a Haar unitary, but its resemblance decreases as $t$ increases. 
In particular, the fact that the Clifford group is finite implies the existence of a recurrence time $t_{\rm rec}$ such that the evolution operator is trivial $W(t_{\rm rec}) = \unity$.

In summary, for $W(t)$ to maximally resemble a Haar unitary, the time $t$ should be the smallest possible to avoid recurrences, but still larger than $t_{\rm scr}$.
This argument explains why the ``long-time ensemble" does not resemble a random unitary, as found in \cite{Huang_2019}.
By the long-time ensemble we mean the set of unitaries $\{e^{-\I Ht}: t\in \mathbb R\}$ generated by a fixed Hamiltonian $H$.

\subsection{Is Clifford dynamics integrable or chaotic?}
\label{sec:int/chao}

In this section we argue that Clifford dynamics has some of the features of quasi-free boson and fermion systems, but at the same time, it displays a stronger chaos. 
For this reason we believe that Clifford dynamics is a very interesting setup to understand the landscape of quantum many-body phenomena.
Next we enumerate essential properties of Clifford dynamics: the first two are in common with quasi-free systems and the subsequent four are not.

{\bf Phase space description and classical simulability.}
  Clifford unitaries can be represented as symplectic transformations in a phase space (in a similar fashion to quasi-free bosons) of dimension exponentially smaller than the Hilbert space. The phase space structure of the Clifford group is described in Appendix \ref{app:CliffordPhaseSpace}.  This dimensional reduction allows to efficiently simulate the evolution of any Pauli operator (and many other relevant operators) with a classical computer.

{\bf Anderson localisation.} Clifford dynamics with disorder (meaning that each gate $U_x$ in Figure~\ref{fig:CircuitFigure} is statistically independent and identically distributed) displays a strong form of localisation, reminiscent of Anderson's localisation. Until now, this strong form of localisation has only been observed in free-particle systems. However, Clifford dynamics cannot be understood in terms of free particles.

{\bf Discrete time.}
  The Clifford phase space is a vector space over a finite field, hence evolution cannot be continuous in time. That is, we can have Floquet-type but not Hamiltonian-type dynamics. 
  The dynamical maps are symplectic matrices with $\mathbb Z_2$ entries, and these cannot be diagonalised.
This lack of eigenmodes prevents us from using many tools and intuitions of quasi-free systems.
  
{\bf No particles.}
  Some specific Clifford dynamics have gliders, which is the discrete-time analog of free particles. But the typical translation-invariant Clifford dynamics consists of fractal patterns~\cite{Gutschow_2010}, and in the non-translation invariant case (i.e.~disorder) we see patterns such as those in Figure~\ref{fig:localisation}. None of these patterns can be understood in terms of free or interacting particles.

{\bf Signatures of chaos.}
If we allow for fully non-local dynamics, quasi-free bosons and fermions cannot generate a 1-design. This is because their evolution operators commute with the number operator (bosons) or the parity operator (fermions). On the contrary, in the non-local case Clifford dynamics generates a 3-design \cite{Webb_2016,Zhu_2017}.
Hence we see that despite the above mentioned similarities with quasi-free systems, Clifford matter seems to display stronger chaos. 
However, chaotic dynamics can be diagnosed by a small (absolute) value of out-of-time order correlators (OTOC) \cite{Roberts_2017}, which is not observed in the Clifford case.
In fact, for any Clifford unitary $W$ and two Pauli operators $\sigma_\u ,\sigma_\v$ the OTOC at infinite temperature takes the maximum value $|\frac 1 d \tr( \sigma_\u W\sigma_\v W^\dagger \sigma_\u W\sigma_\v W^\dagger )| = 1$.
Incidentally, a small OTOC follows from being a 4-design but not a 3-design. 

{\bf Absence of local integrals of motion.} In the translation-invariant case some Clifford models \cite{Zimboras_2020} with local interactions have fully non-local integrals of motion.
This means that each operator that commutes with the evolution operator involves couplings which do not decay with the distance and act on an extensive number of sites (unbounded wight).




\subsection{Time-dependent vs time-independent circuits} 
\label{sec:t dep}

Time-dependent local quantum circuits (see Figure \ref{fig:RandomCircuitFigure}) have been used as a model for chaotic dynamics in numerous contributions \cite{Harrow_2018, Brandao_2016, Brandao_2016a, Harrow_2009, Hunter_Jones_2019}.
It has been proven that these circuits generate approximate $k$-designs where the order increases with time as $k \sim t^{1/10}$ (although the scaling is conjectured to be $k \sim t$ \cite{Hunter_Jones_2019}).
Some authors have attempted to model chaotic systems with conserved quantities by using time-dependent local circuits constrained so that each gate commutes with an operator of the form $Q = \sum_x \sigma_z^{(x)}$, where $x$ labels all sites \cite{Khemani_2018a, Hunter_Jones_2018}.
These $Q$-conserving circuits also generate approximate $k$-designs in the operator space orthogonal to $Q$, with $k$ increasing as time passes.

We argue that the dynamics of $Q$-conserving circuits is very different from time-independent circuits like the model we are studying, Figure \ref{fig:CircuitFigure}.
Despite the fact that in both cases there are conserved quantities ($Q$ and $W(t=1)$), $Q$-conserving time-dependent circuits do not have time-correlations nor recurrences (see Section~\ref{sec:scrambling time}).
This implies that they resemble Haar unitaries more and more as time goes on. Instead, as discussed in Section~\ref{sec:scrambling time}, time-independent dynamics loses its resemblance to Haar unitaries with time.

Previous works \cite{Nakata_2017, Brandao_2019} have constructed unitary designs with ``nearly time-independent'' dynamics.
This consists of an evolution where the Hamiltonian changes a small number of times, and it is time-independent in between changes.
A different line of work \cite{Prosen_2018_1, Prosen_2018_2, Prosen_2019, Bertini_2019, Prosen_2007} analyses disordered time-periodic dynamics with non-Clifford gates. 
These more general dynamics makes these models more chaotic than ours.
However, these works only prove that these models display certain aspects of Haar-random unitaries, instead of indistinguishability as captured by Theorem \ref{thm:approx2Design}.

\subsection{A variant of our model}

We define our model as having $L$ sites, with $N$ qubits per site, and nearest-neighbour interactions.
However, this is equivalent to say that it has $LN$ sites, with a single qubit per site, and $2N$-range interactions.
For this we use the fact that any Clifford gate of $2N$ qubits can be written as a circuit of depth $\mathcal O(N^2/ \log{N})$ \cite{Aaronson_2004, Gottesman_1998}. 
Hence, a dynamical period in the $LN$-site circuit decomposes into $\mathcal O(N^2 / \log N )$ elementary time steps.

\section{Conclusion and outlook}\label{sec:conc}


The dynamics of highly chaotic quantum systems, such as black holes \cite{Page_1993, Hayden_2007}, is often modelled with Haar-random unitaries, which allows for the exact calculation of relevant quantities.
This model is often justified by the fact that local random circuits \cite{Brandao_2016, Brandao_2016a, Harrow_2009} generate 2-designs. 
However, these circuits are time dependent, while presumably the dynamics of black holes are not \cite{Witten_1998}. 
In this work we make a step forward towards the justification of the Haar-unitary model of dynamics in quantum chaotic systems, by proving that the evolution operator of a time-periodic model cannot be distinguished from a random unitary in some physically relevant setups. 

An important question that remains open is whether local and time-independent (or time-periodic) dynamics can generate a 2-design. This amounts to not restricting the measurement in the discrimination process. The results in \cite{Prosen_2018_1, Prosen_2018_2, Prosen_2019, Bertini_2019, Prosen_2007} provide some hope in this direction.
However, we expect that the 2-design property is at best achieved around the scrambling time, and it fades away as time goes on (see discussion in Section \ref{sec:scrambling time} and in reference \cite{Brandao_2019}).
More generally, we would like to characterise which further properties of random unitaries are present in naturally-occurring dynamics.

\section{Acknowledgments}

We are  grateful to Nick Hunter-Jones, Oliver Lunt and Arijeet Pal for valuable discussions.
Tom Farshi acknowledges financial support by the Engineering and Physical Sciences Research Council (grant number EP/L015242/1).
Daniele Toniolo and Lluis Masanes acknowledge financial support by the UK's Engineering and Physical Sciences Research Council (grant number EP/R012393/1). 
Carlos Gonz\'alez acknowledges financial support by Spanish MINECO (project MTM2017-88385-P), MECD ``Jos\'e Castillejo'' program (CAS16/00339) and Programa Propio de I+D+i of the Universidad Polit\'ecnica de Madrid.
Research at Perimeter Institute is supported in part by the Government of Canada through the Department of Innovation, Science and Economic Development Canada and by the Province of  Ontario  through  the  Ministry  of  Economic  Development, Job Creation and Trade.

\section{Data availability}
The data that give rise to Figure \ref{fig:localisation} are available from the corresponding author
upon reasonable request.


\vspace{5mm}
\section*{Appendices}
\appendix




\section{Clifford dynamics and discrete phase space}
\label{app:CliffordPhaseSpace}
In this appendix, we first define the Pauli and Clifford groups and then present the phase-space description of Clifford dynamics. 
This description is known from previous works \cite{Koenig_2014,Aaronson_2004,Gottesman_1998} and we include it here for clarity of presentation. 

The Pauli sigma matrices together with the identity $\{ \unity, \sigma_x, \sigma_y, \sigma_z\}$ form a basis of the space of operators of one qubit $\mathbb C^2$.
Also, the sixteen matrices obtained by multiplying $\{ \unity, \sigma_x, \sigma_y, \sigma_z\}$ times the coefficients $\{1,\I,-1,-\I\}$ form a group. 
This is called the Pauli group of one qubit and it is denoted by $\mathcal P_1$.
The generalization to $n$ qubits is the following.

\begin{Definition}
The {\bf Pauli group} of $n$ qubits $\mathcal P_n$ is the set of matrices $\I^u \sigma_\u$ where
\begin{align}
  \label{eq:Pauli element}
  \sigma_\u
  =
  \bigotimes_{i=1}^n 
  (\sigma_x^{q_i} \sigma_z^{p_i}) \in {\rm U}(2^n)\ ,
\end{align}
for all phases $u \in \mathbb Z_4$ and vectors $\u = (q_1, p_1, q_2, p_2, \ldots, q_n, p_n ) \in \Z2^{2n}$.
We also define $\bar{\cal P}_n = \mathcal P_n /\{1,\I,-1,-\I\}$ which satisfies $\bar{\cal P}_n \cong \Z2^{2n}$.
\end{Definition}

\noindent
Here $\Z2^{2n}$ stands for a $2n$-dimensional vector space with addition and multiplication operations defined modulo $2$.
Using the identity $\sigma_z \sigma_x = -\sigma_x \sigma_z$ and the definition $\beta(\u, \u') = \sum_{i=1}^n p_i q_i'$ we obtain the multiplication and inverse rules
\begin{align}
  \label{eq:Pauli product}
  \sigma_\u \sigma_{\u'}
  &=
  (-1)^{\beta(\u, \u')} \sigma_{\u +\u'}\ ,
\\  \label{eq:Pauli inverse}
  \sigma_\u^{-1}
  &=
  (-1)^{\beta(\u, \u)} \sigma_\u\ .
\end{align}
The Pauli group \eqref{eq:Pauli element} is the discrete version of the Weyl group, or the \emph{displacement operators} used in quantum optics.
Concretely, if $\hat Q$ and $\hat P$ are quadrature operators (satisfying the canonical commutation relations $[\hat Q, \hat P] =\I \mathds{1} $) then we can write the analogy as
\begin{equation}
  \label{eq:analogy}
  \sigma_x^{q} \sigma_z^{p}
  \ \longleftrightarrow\ 
  e^{\I \hat P q} e^{\I\hat Q p}  \ ,
\end{equation}
where the phase space variables $(q,p)$ take values in $\Z2^2$ on the left of \eqref{eq:analogy}, and in $\mathbb R^2$ on the right.
This analogy also extends to the set of transformations that preserve the phase space structure.
Before characterizing these transformations let us  define the phase space associated to the Pauli group.

\begin{Definition}
  The {\bf discrete phase space} of $n$ qubits $\Z2^{2n}$ is the $2n$-dimensional vector space over the field $\Z2$, endowed with the symplectic (antisymmetric) bilinear form
\begin{equation}
  \label{inner product}
  \langle \u, \u'\rangle = \u^T J \u'\ ,
  \text{ where }\ 
  J = \bigoplus_{i=1}^n \left( \begin{array}{cc}
     0 & 1 \\
     1 & 0 
  \end{array}\right)\ ,
\end{equation}
for all $\u, \u' \in \mathbb Z_2^{2n}$.
Note that the form is indeed antisymmetric $\langle \u, \u'\rangle = \langle \u', \u\rangle = -\langle \u', \u\rangle \bmod 2$, which implies $\langle \u, \u\rangle =0$.
\end{Definition}

\noindent
Using the symplectic form \eqref{inner product} and the rules (\eqref{eq:Pauli product}-\eqref{eq:Pauli inverse}) we can write the commutation relations of the Pauli group as
\begin{align}
\label{eq:commutatation relation}
  \sigma_\u\, \sigma_{\u'}\, 
  \sigma_\u^{-1}\, \sigma_{\u'}^{-1} 
  =
  (-1)^{\langle \u, \u' \rangle}\ .
\end{align}
In analogy with the continuous (bosonic) phase space, in the following two definitions we introduce the transformations that preserve the symplectic form \eqref{inner product} and the Pauli group, respectively.

\begin{Definition}
\label{def:symplectic group}
The {\bf symplectic group} $\S_n$ is the set of matrices $S: \Z2^{2n} \to \Z2^{2n}$ such that
\begin{equation}
  \langle S\u, S\u' \rangle
  =
  \langle \u, \u' \rangle\ ,
\end{equation}
for all $\u, \u' \in \Z2^{2n}$.
This is equivalent to the condition $S^T J S = J \bmod 2$.
\end{Definition}

\begin{Definition}\label{def:Clifford}
  A unitary $U\in {\rm U}(2^n)$ is {\bf Clifford} if it maps each Pauli operator $\sigma\in \mathcal P_n$ to a Pauli operator $U \sigma U^\dagger \in \mathcal P_n$. Two Clifford unitaries $U,V$ are equivalent if there is a complex phase $\lambda$ with $|\lambda|=1$ such that $U= \lambda V$. The {\bf Clifford group} of $n$ qubits $\C_n$ is the set of equivalence classes of Clifford unitaries in ${\rm U}(2^n)$. 
\end{Definition}

\noindent
In this work we only consider the adjoint representation $\sigma \mapsto U\sigma U^\dagger$, hence, the phase $\lambda$ does not play any role.

\begin{Lemma}[Structure of $\C_n$]\label{lem:Clifford_phase_space}
  Each Clifford transformation $U\in \C_n$ is characterized by a symplectic matrix $S \in \S_n$ and a vector $\mathbf s \in \Z2^{2n}$ so that
\begin{equation}
  U \sigma_\u U^\dagger
  =
  \I^{\alpha[S,\u]}\, (-1)^{\langle\mathbf s, \u\rangle} \sigma_{S\u}\ ,
\end{equation}
where the function $\alpha$ takes values in $\mathbb Z_4$.
More precisely we have $\C_n \cong \bar{\cal P}_n \rtimes \S_n$.
\end{Lemma}
In this work the function $\alpha$ does not play any role, hence, we do not provide a characterization.

\begin{proof}
For each $U\in \C_n$ there are two functions 
\begin{align}
  s & : \Z2^{2n} \to \mathbb Z_4\ ,
  \\
  S & : \Z2^{2n} \to \Z2^{2n}\ ,
\end{align}  
such that
\begin{equation}
  U \sigma_\u U^\dagger = \I^{s[\u]}\, \sigma_{S[\u]}\ .
\end{equation}
Note that, at this point, we do not make any assumption about these functions, such as linearity. 
Using \eqref{eq:Pauli product} we obtain the equality between the following two expressions
\begin{align}
  \nonumber
  U \sigma_\u \sigma_{\u'} U^\dagger
  &= 
  (-1)^{\beta(\u, \u')}\, U \sigma_{\u+\u'} U^\dagger
  \\ \label{eq:ussU} &=
  (-1)^{\beta(\u, \u')}\, \I^{s[\u+\u']}\, \sigma_{S[\u+\u']}\ ,
  \\ \nonumber
  U \sigma_\u U^\dagger U \sigma_{\u'} U^\dagger
  &=
  (\I^{s[\u]}\, \sigma_{S[\u]})\, 
  (\I^{s[\u']}\, \sigma_{S[\u']})
  \\ \label{eq:usuusu} &=
  (-1)^{\beta(S\u, S\u')}\, \I^{s[\u]+s[\u']}\, \sigma_{S[\u]+S[\u']}\ ,
\end{align}
which implies the $\Z2$-linearity of the $S$ function. Hence, from now on, we write its action as a matrix $S[\u] = S\u$.
Next, if we impose the commutation relations of the Pauli group \eqref{eq:commutatation relation} as follows
\begin{align}
  \nonumber
  (-1)^{\langle \u, \u' \rangle}
  & =
  U \sigma_{\u}\, \sigma_{\u'}\,
  \sigma_{\u}^{-1}\, \sigma_{\u'}^{-1}\, U^{-1}
  \\ \nonumber & =
  (\I^{s[\u]} \sigma_{S\u})\, (\I^{s[\u']} \sigma_{S\u'})\, 
  (\I^{s[\u]} \sigma_{S\u})^{-1}\, (\I^{s[\u']} \sigma_{S\u'})^{-1}
  \\ & =
  (-1)^{\langle S\u, S\u' \rangle}\ ,
\end{align}
we find that the matrices $S$ are symplectic.
Conversely, it has been proven \cite{Gross_2019,Koenig_2014,Calderbank_1998,Gross_2006a} that for each symplectic matrix $S\in \mathcal S_n$ there is $U\in \C_n$ such that $U\sigma_\u U^\dagger \propto \sigma_{S\u}$ for all $\u$.

Now, let us obtain the set of pairs $(S,s)$ associated to the subgroup $\bar {\cal P}_n \subseteq  \C_n$.
Using \eqref{eq:commutatation relation} we see that the Clifford transformation $\sigma_\v \in \bar {\cal P}_n$ has $S=\unity$ and $s[\u] = 2\langle\v, \u\rangle$, for any $\v\in \Z2^{2n}$.
Next, let us prove the converse.
By equating \eqref{eq:ussU} and \eqref{eq:usuusu} with $S=\unity$, we see that any Clifford transformation $U$ with $S=\unity$ has a phase function $s$ satisfying 
\begin{equation}
  \label{eq:pseudolinear} 
  s[\u+\u'] = s[\u]+s[\u']\ ,
\end{equation}
for all pairs $\u, \u'$.
Also, since the map $\sigma_\u \to U\sigma_\u U^\dagger$ preserves the Hermiticity or anti-Hermiticity of $\sigma_\u$, the phase function in $U\sigma_\u U^\dagger = \I^{s[\u]} \sigma_\u$ has to satisfy $s[\u] \in \{0,2\}$ for all $\u$.
Combining this with \eqref{eq:pseudolinear} we deduce that, if $S=\unity$ then $s[\u] = 2\langle\v,\u\rangle$ for some vector $\v\in \Z2^{2n}$.
In summary, an element of the Clifford group belongs to the Pauli group if, and only if, there is a vector $\v\in \Z2^{2n}$ such that $S=\unity$ and $s[\u] = 2\langle\v,\u\rangle$.

Now let us show that $\C_n/\bar{\cal P}_n \cong \mathcal S_n$.
By definition, any Clifford element $U\bar{\cal P}_n U^\dagger \subseteq \bar{\cal P}_n$ satisfies $U \bar{\cal P}_n = \bar{\cal P}_n U$, hence $\bar{\cal P}_n \subseteq \C_n$ is a normal subgroup.
This allows us to allocate each element $U\in \C_n$ into an equivalence class $U\bar{\cal P}_n \subseteq \C_n$, and define a group operation between classes. 
In order to prove the isomorphism $\C_n/\bar{\cal P}_n \cong \mathcal S_n$, we need to check that two transformations $U,U'$ are in the same equivalence class ($\exists\, \v : U = U'\sigma_\v$) if and only if they have the same symplectic matrix $S=S'$. 
Identity \eqref{eq:commutatation relation} tells us that $U = U'\sigma_\v$ implies $S=S'$.
To prove the converse, let us assume that $U,U'$ have symplectic matrices $S=S'$.
Due to the fact $U^{-1}$ has symplectic matrix $S^{-1}$, the product $U^{-1} U'$ has symplectic matrix $S^{-1} S = \unity$.
As proven above, this implies that $U^{-1} U' \in \bar{\cal P}_n$,  and therefore both are in the same class.

Finally, for each symplectic matrix $S$ we define $\alpha[S,\u] = s[\u]$ where $s$ is the phase function of an arbitrarily chosen element in the equivalence class defined by $S$.
The phase function of the other elements in the class $S$ is $s[\u] = \alpha[S,\u] +2\langle\v, \u\rangle$ for all $\v\in \Z2^{2n}$.
\end{proof}

\section{Proof of Lemma \ref{lemma:order S}} \label{app:sympl_order}

\begin{proof}
We start considering $ \ln |\S_n| $. 
\begin{align} \nonumber
\ln |\S_n| &= \ln \left[ \prod_{i=1}^n (2^{2i}-1) \prod_{j=1}^n  2^{2j-1} \right] \\ \nonumber
&=\sum_{i=1}^n \ln(2^{2i}-1) + \sum_{j=1}^n \ln 2^{2j-1} \\  \nonumber
&=\sum_{i=1}^n \ln \left[2^{2i}(1-2^{-2i})\right] + \sum_{j=1}^n (2j-1)\ln 2 \\ \nonumber
&=\sum_{i=1}^n 2i \ln 2 + \sum_{i=1}^n \ln (1-2^{-2i}) + \sum_{j=1}^n (2j-1) \ln 2 \\ \nonumber
&=n(n+1)\ln 2 + \sum_{i=1}^n \ln (1-2^{-2i}) + n^2 \ln 2 \\ \label{log_appr}
&=n(2n+1)\ln 2 + \sum_{i=1}^n \ln (1-2^{-2i}). \\ \nonumber
\end{align}
We use $ \frac{x}{1+x} < \ln(1+x) < x $, with $x\neq 0 $ and $ x>-1$, to upper and lower bound the logarithm in \eqref{log_appr}. The corresponding bounds on $ |\S_n| $ are obtained after exponentiating
\begin{align}\nonumber
\sum_{i=1}^n \ln (1-2^{-2i}) < -\sum_{i=1}^n 2^{-2i}=  -\sum_{i=1}^n \frac{1}{4^i}= -\frac{1}{3} \left( 1-\frac{1}{4^n} \right)\\ 
\end{align}
$ b(n) $ is defined to be $ b(n) \equiv e^{-\frac{1}{3}(1-\frac{1}{4^n})} $, moreover $ b(n)< b(1) = e^{-\frac{1}{4}} < 0.78 $.

To obtain the lower bound of $ |S_n| $, from \ref{log_appr} we have:
\begin{align}\nonumber
\sum_{i=1}^n \ln (1-2^{-2i}) > -\sum_{i=1}^n \frac{2^{-2i}}{1-2^{-2i}}  
\end{align}
and $ a(n) \equiv e^{-\sum_{i=1}^n \frac{2^{-2i}}{1-2^{-2i}} } < b(n) $. We observe that: 
\begin{equation}
 a(n) \equiv e^{-\sum_{i=1}^n \frac{1}{2^{2i}-1} } \ge e^{-\frac{1}{3}\sum_{i=0}^{n-1} \frac{1}{2^{2i}} } > e^{-\frac{4}{9}} > 0.64
\end{equation}

\end{proof}



\section{Additional lemmas}
\label{app:AdditionalLemmas}

In this section, we include additional lemmas that are used in the proof of other results.

\begin{Lemma}
\label{lemma:subspace counting}
The number of $k$-dimensional subspaces of $\Z2^n$ is 
\begin{equation}
  \label{def:num subspaces}
  \N_k^n 
  = \prod_{i=0}^{k-1} \frac
  {2^n -2^i} {2^k -2^i}\ . 
\end{equation}
\end{Lemma}

\begin{proof}
Let us start by counting how many lists of $k$ linearly independent vectors $(\u_1, \ldots, \u_k)$ are in $\Z2^n$. 
The first vector $\u_1$ can be any element of $\Z2^n$ except the zero vector $\zero$, giving a total of $(2^n-1)$ possibilities.
Following that, $\u_2$ can be any element of $\Z2^n$ that is not contained in the subspace generated by $\u_1$, which is $\{\zero, \u_1\}$, giving $(2^n-2)$ possibilities.
Analogously, $\u_3$ can be any element of $\Z2^n$ that is not contained in the subspace generated by $\{\u_1, \u_2\}$, which is $\{\zero, \u_1, \u_2, \u_1 +\u_2\}$, giving  $(2^n-2^2)$ possibilities.
Following in this fashion we arrive at the following conclusion. The number of lists of $k$ linearly independent vectors is 
\begin{equation}
  \mathcal L_k^n
  =
  (2^n -2^0)(2^n -2^1)(2^n -2^2) \cdots (2^n - 2^{k-1})\ . 
\end{equation}  
It is important to note that many lists $(\u_1, \ldots, \u_k)$ generate the same subspace.
So, in order to obtain $\N_k^n$, we have to divide $\mathcal L_k^n$ by the number of lists which generate that same subspace.

First, we note that a list $(\u_1, \ldots, \u_n)$ is a basis of $\Z2^n$ with its vectors in a particular order.
Hence, $\mathcal L_n^n$ is the number of basis (in particular order) of $\Z2^n$. 
Second, we use the fact that the subspace of $\Z2^n$ generated by the list $(\u_1, \ldots, \u_k)$ is isomorphic to $\Z2^k$, so that, the number of basis (in a particular order) generating that subspace is $\mathcal L_k^k$.
Putting things together, we obtain $\N_k^n = \mathcal L_k^n /\mathcal L_k^k$, as in \eqref{def:num subspaces}
\end{proof}

\begin{Lemma}
\label{lemma:upperbound subspace}
Let $\N_k^n$ be the number of $k$-dimensional subspaces of $\Z2^n$; then we have
\begin{equation}
  2^{(n-k)k} 
  (1 - 2^{k-n})^k 
  \ \leq\  
  \N_k^n 
  \ \leq\ 
  2^{(n-k)k}
  \min\{2^k,4\}\ .
\end{equation}
\end{Lemma}

\begin{proof}
Taking Lemma \ref{lemma:subspace counting} and neglecting the negative terms in the numerator gives
\begin{align}
  \N_k^n 
  &= 
  \prod_{i=0}^{k-1} 
  \frac {2^n -2^i} {2^k -2^i}
  \leq
  \prod_{i=0}^{k-1} 
  \frac {2^n} {2^k -2^i}
  \\ &=
  \frac{2^{nk}}{2^{k^{2}}} 
  \prod_{i=0}^{k-1} \frac
  {1} {1 -2^{i-k}}
  =
  2^{(n-k)k}   
  \prod_{j=1}^{k} 
  \frac 1 {1 - 2^{-j}}
  \\ &\leq 
  2^{(n-k)k}   
  \prod_{j=1}^{\infty} 
  \frac 1 {1 - 2^{-j}}\ ,
\end{align}
where in the last inequality we have extended the product to infinity.
It turns out that this infinite product is the inverse of Euler's function $\phi$ evaluated at $1/2$, which has the value
\begin{equation}
  \phi(1/2)
  =
  \prod_{j=1}^{\infty} 
  (1 - 2^{-j})
  \approx 
  .28
  \geq \frac 1 4\ .  
\end{equation}
Combining the two above inequalities we obtain
\begin{equation}
  \N_k^n 
  \leq
  2^{(n-k)k}
  4\ .
\end{equation}
For the cases where $k=0,1$, we can improve this bound. 
When $k=0$ the coefficient is 1 by definition, and when $k=1$ the product $\prod_{i=0}^{k-1} (1 -2^{i-k})^{-1}$ evaluates to 2.
Hence, for $k=0,1$ we can replace 4 by $2^k$, and therefore this improvement is captured concisely by changing 4 to $\min \{ 2^k , 4 \}$.

We obtain the lower bound by instead neglecting the negative terms in the denominator
\begin{equation}
    \N_k^n \geq \prod_{i=0}^{k-1} \frac{2^n - 2^i}{2^k} = \frac{2^{nk}}{2^{k^2}} \prod_{i=0}^{k-1} (1 - 2^{i-n}) \ .
\end{equation}
The remaining product can be bounded using by
\begin{equation}
     \prod_{i=0}^{k-1} (1 - 2^{i-n}) 
     \geq  
     \prod_{i=0}^{k-1} (1 - 2^{k-n})  
     \geq 
     (1 - 2^{k-n})^k \ ,
\end{equation}
since $n \geq k>i$, and hence we get the final lower bound. 
\end{proof}

\begin{Lemma}
\label{lem:BinomialUpper}
The binomial coefficient can be bounded by
\begin{equation}
    \binom{k+r-1}{k} < \left( 1+r  \right)^k \leq (2r)^k \ .
\end{equation}
\end{Lemma}
\begin{proof}
We start with the bound 
\begin{equation}
    \binom{k+r-1}{k} = \prod_{i=1}^{k} \frac{r+k-i}{i} < \prod_{i=1}^{k} \left(1 + \frac{r}{i} \right)  \ .
\end{equation}
This follows from:
\begin{align}\label{first}
 \prod_{i=1}^k \frac{r+k-i}{i} = \frac{1}{k!} \prod_{i=1}^k (r+k-i) = \frac{1}{k!}(r+k-1)(r+k-2) \ldots r
\end{align}
\begin{align}\label{second}
 \prod_{i=1}^k \frac{r+i}{i} = \frac{1}{k!} \prod_{i=1}^k (r+i) = \frac{1}{k!}(r+k)(r+k-1) \ldots (r+1)
\end{align}
The order of the factors in the product in \eqref{second} has been inverted. It is easy to see by inspection that \eqref{first} lower bound \eqref{second}.
Further bounding we get:
\begin{align}
    \binom{k+r-1}{k}  <  \prod_{i=1}^{k} \left(1 + \frac{r}{i} \right) \leq (1+r)^k \leq (2 r)^k \ . 
\end{align}
\end{proof}

\begin{Lemma}
\label{S block permutation}
For any given $S\in \S_{2n}$ written in block form 
\begin{equation}
    S
    =
    \left( \begin{array}{cc}
     A & B \\
     C & D 
\end{array}\right) ,
\end{equation}
according to the decomposition $\Z2^{4n} = \Z2^{2n} \oplus \Z2^{2n}$, then
\begin{equation}
\begin{pmatrix}
     B & A \\
     D & C 
\end{pmatrix}
\ , \quad
\begin{pmatrix}
     C & D \\
     A & B 
\end{pmatrix}
\ , \text{ and} \quad
\begin{pmatrix}
     D & C \\
     B & A 
\end{pmatrix}
\ ,
\end{equation}
are all also symplectic matrices.
\end{Lemma}
\begin{proof}
Using the symplectic matrix
\begin{equation}
 M = 
    \begin{pmatrix}
     0 & \mathbbm{1} \\
     \mathbbm{1} & 0 
\end{pmatrix}\ ,
\end{equation}
we show, using the result for the product of symplectic matrices, that the three permuted versions of $S$ are also symplectic matrices via $MS$, $SM$, and $MSM$. 
\end{proof}

\begin{Lemma}
\label{eq:C=0iffB=0}
For any given $S\in \S_{2n}$ written in block form 
\begin{equation}
    S
    =
    \left( \begin{array}{cc}
     A & B \\
     C & D 
\end{array}\right) ,
\end{equation}
according to the decomposition $\Z2^{4n} = \Z2^{2n} \oplus \Z2^{2n}$, the following two properties hold:
\begin{align}
    B = 0 & \iff C = 0 \ ,  \\
    A = 0 & \iff D = 0 \ .
\end{align}
\end{Lemma}
\begin{proof}
First, we consider the single case where $C = 0$. 
Following the algorithm for generating a symplectic matrix in Lemma \ref{lemma:algorithm}, we see that $A$ must be ($2n\times2n$) symplectic matrix.
Hence, any choice for the columns of $B$ will have symplectic form of one with at least one column of the matrix $A$.
Therefore, to fulfil the symplectic constraints for the entire matrix $S$, the corresponding column of D must have symplectic form of one with a column of C. 
However, this is not possible since $C = 0$, therefore $B = 0$.
Finally, by Lemma \ref{S block permutation} this argument applies to each block.
\end{proof}

\bibliography{bibliography}

\end{document}